\documentclass[12pt]{article}

\usepackage[normalem]{ulem}

\usepackage[title]{appendix}
\usepackage{lscape}
\usepackage{amsfonts}
\usepackage{amssymb}
\usepackage{amsmath}
\usepackage{epstopdf}
\usepackage{float}
\usepackage{natbib}
\usepackage[margin=.75in]{geometry}
\usepackage[hang]{caption}
\usepackage{color}
\usepackage[pdftex]{graphicx}
\usepackage{rotating}

\usepackage{bbm}
\newcommand{\Perp}{\perp\!\!\!\perp}
\newcommand{\qed}{\hfill $\blacksquare$}
\renewcommand{\baselinestretch}{1.2}
\allowdisplaybreaks

\newtheorem{theorem}{Theorem}[section]
\newenvironment{proof}[1][Proof]{\begin{trivlist}
\item[\hskip \labelsep {\bfseries #1.}]}{\end{trivlist}}
\newtheorem{proposition}[theorem]{Proposition}
\newtheorem{lemma}[theorem]{Lemma}

\newtheorem{assumption}{Assumption}

\newtheorem{example}[theorem]{Example}
\newtheorem{definition}{Definition}

\DeclareMathOperator{\argmin}{arg\,min}
\newcommand{\cp}{\stackrel{\mathbb{P}}{\longrightarrow}}
\newcommand{\cd}{\stackrel{d}{\longrightarrow}}
\newcommand{\as}{\stackrel{a.s.}{\longrightarrow}}

\makeatletter
\newcommand*\rel@kern[1]{\kern#1\dimexpr\macc@kerna}
\newcommand*\widebar[1]{%
  \begingroup
  \def\mathaccent##1##2{%
    \rel@kern{0.8}%
    \overline{\rel@kern{-0.8}\macc@nucleus\rel@kern{0.2}}%
    \rel@kern{-0.2}%
  }%
  \macc@depth\@ne
  \let\math@bgroup\@empty \let\math@egroup\macc@set@skewchar
  \mathsurround\z@ \frozen@everymath{\mathgroup\macc@group\relax}%
  \macc@set@skewchar\relax
  \let\mathaccentV\macc@nested@a
  \macc@nested@a\relax111{#1}%
  \endgroup
}
\makeatother

\begin{document}

\title{A GMM approach to estimate\\the roughness of stochastic volatility\thanks{A previous draft of this article was circulated under the title ``Roughness in spot variance? A GMM approach for estimation of fractional log-normal stochastic volatility models using realized measures''. The paper was presented at the 2019 ``Frontiers in Quantitative Finance'' workshop at Copenhagen University, Denmark, the 2021 European Winter Meeting of the Econometric Society, the 2022 DAGStat conference in Hamburg, Germany, the 2022 Aarhus--Singapore Management University (SMU) joint workshop on Volatility, the 2022 Vienna--Copenhagen (VieCo) conference on Financial Econometrics in Copenhagen, Denmark, and in seminars at McMaster University, SMU, and University of Waterloo. We appreciate input from attendees at these venues. We are also grateful for insightful comments and suggestions made by the co-editor (Torben G. Andersen), the associate editor, and two anynomous referees at the Journal of Econometrics. In addition, we thank Mikkel Bennedsen, Christa Cuchiero, Masaaki Fukasawa, Jim Gatheral, Shin Kanaya, Jia Li, Roberto Ren\`{o}, Shuping Shi and Jun Yu for feedback. Christensen was partially funded by a grant from the Independent Research Fund Denmark (DFF 1028--00030B). This work was also supported by the Center for Research in Econometric Analysis of Time Series (CREATES). Please address correspondence to: kim@econ.au.dk.}}
	
\author{Anine E. Bolko\thanks{Department of Economics and Business Economics, Aarhus University, Fuglesangs All\'{e} 4, 8210 Aarhus V, Denmark.} \!\!\!\!\!
\and Kim Christensen\footnotemark[2] \!\!\!\!\!
\and Mikko S. Pakkanen\thanks{Department of Mathematics, Imperial College London, South Kensington Campus, London SW7 2AZ, UK.} $^{\text{,}}$ \kern-0.15cm \footnotemark[2] \!\!\!\!\!
\and Bezirgen Veliyev\footnotemark[2]}

\date{September 7, 2022}

\maketitle

\vspace*{-1.00cm}

\begin{abstract}
We develop a GMM approach for estimation of log-normal stochastic volatility models driven by a fractional Brownian motion with unrestricted Hurst exponent. We show that a parameter estimator based on the integrated variance is consistent and, under stronger conditions, asymptotically normally distributed. We inspect the behavior of our procedure when integrated variance is replaced with a noisy measure of volatility calculated from discrete high-frequency data. The realized estimator contains sampling error, which skews the fractal coefficient toward ``illusive roughness.'' We construct an analytical approach to control the impact of measurement error without introducing nuisance parameters. In a simulation study, we demonstrate convincing small sample properties of our approach based both on integrated and realized variance over the entire memory spectrum. We show the bias correction attenuates any systematic deviance in the parameter estimates. Our procedure is applied to empirical high-frequency data from numerous leading equity indexes. With our robust approach the Hurst index is estimated around $0.05$, confirming roughness in stochastic volatility.


\end{abstract}

\vfill

\thispagestyle{empty}

\pagebreak

\section{Introduction} \setcounter{page}{1}

Stochastic volatility (SV) models are pervasive in finance. Over the years, a variety of different models---each with its own dynamics---were developed, such as the log-normal model \citep*[e.g.][]{taylor:86a}, the square-root diffusion \citep*[e.g.][]{heston:93a}, or more complicated processes where volatility is driven by a non-Gaussian pure-jump component, e.g. \citet*{barndorff-nielsen-shephard:01a,todorov-tauchen:11a}.

In this paper, we investigate the log-normal SV model, which has been extensively studied in previous work, e.g. \citet*{alizadeh-brandt-diebold:02a, christoffersen-jacobs-mimouni:10a, hull-white:87a}. This class is a promising starting point, because the unconditional distribution of realized variance is close to log-normal \citep*[see e.g.][]{andersen-bollerslev-diebold-ebens:01a, andersen-bollerslev-diebold-labys:03a, christensen-thyrsgaard-veliyev:19a}.

While there is a general acceptance that log-normal volatility offers a decent description of return variation in financial asset prices, there is no consensus on the properties of the background driving Gaussian process. In the SV setting, it is often assumed to be a standard Brownian motion. The mean-reversion and volatility-of-volatility parameters then control both the local properties of volatility and also determine its long-run persistence. There are multiple papers dealing with estimation of the parameters of the log-normal SV model, for example using method of moments- or likelihood-based approaches \citep*[e.g.][]{taylor:86a, melino-turnbull:90a, duffie-singleton:93a, harvey-ruiz-shephard:94a, gallant-hsieh-tauchen:97a, fridman-harris:98a}. In the context of generalized method of moments (GMM) estimation, \citet*{andersen-sorensen:96a} offer further advice on how to select moment criteria and the weighting matrix in order to get good results in small samples.

When the driving process is a fractional Brownian motion, which neither is a semimartingale nor has independent increments, less is known. In this setting, part of the memory in volatility is reallocated to the fractal index or \citet*{hurst:51a} exponent. \citet*{comte-renault:98a} propose a log-normal SV model, where the Hurst exponent is larger than one-half; the value implied by a standard Brownian motion. This induces positive serial correlation---or long-memory---in the increments of the process. \citet*{bennedsen-lunde-pakkanen:22a, euch-rosenbaum:18a, fukasawa-takabatake-westphal:22a, gatheral-jaisson-rosenbaum:18a}, among others, study ``roughness'' as captured by a fractal index smaller than one-half, rendering volatility highly erratic---or anti-persistent---at short time scales. In the fractional setting, the Hurst index is typically pre-estimated based on a semi-parametric procedure, before the other parameters are recovered. While this may yield consistent parameter estimates, it is generally inefficient and may be severely biased in finite samples.

In this paper, we extend the classical GMM procedure to joint estimation of the parameters of the log-normal SV model with a general fractal index. An attractive feature of our paper is that moment expressions are derived in convenient integral form facilitating the implementation without recourse to simulation-based approaches. As in many papers preceeding this one, we appeal to the time series properties of integrated variance to construct our estimator, an idea pioneered by \citet*{bollerslev-zhou:02a}.

In practice, the integrated variance is unobserved. Realized variance, which is computed from high-frequency data, is a consistent estimator of integrated variance and often replaces it in the calculations. In previous work, this is done by showing convergence in probability or distribution of the parameter estimator in a double-asymptotic in-fill and long-span setting, such that the volatility discretization error is small enough to be ignored. In the subsequent applications, however, the volatility proxy enters directly in place of integrated variance.

Substituting the latent volatility with a proxy, however, entails a measurement error in finite samples, which obfuscates the underlying integrated variance dynamics. This can be detrimental to the estimation procedure if unaccounted for \citep*[e.g.][]{meddahi:02a, hansen-lunde:14a}, because the moment conditions imposed on integrated variance are generally not valid for the noisy proxy. \citet*{barndorff-nielsen-shephard:02a} employ a state-space system and the Kalman filter to smooth out realized variance prior to maximum quasi-likelihood estimation of their SV model, see also \citet*{meddahi:03a}. In this paper, we introduce a high-level assumption employing a realized measure to proxy for integrated variance following \citet*{patton:11a}. We construct a bias correction that controls for the measurement error and show how to embed it analytically into the GMM setting without introducing additional nuisance parameters, as opposed to \citet*{bollerslev-zhou:02a}.

Our proposed estimator is consistent and asymptotically normal. The main asymptotic theory is long-span with time going to infinity but high-frequency data sampled at a fixed frequency. As an aside, we complement the analysis by deriving the double-asymptotic result, where the estimation error correction is immaterial.

We investigate our estimator in a simulation study, where various configurations of a fractional log-normal SV model with Hurst parameter covering the rough, standard and long-memory setting are inspected. The parameter estimates are centered closely around the true values---suggesting our procedure is approximately unbiased---and relatively accurate, once the bias correction is adopted. In an empirical application, we study an extensive selection of major equity indexes and confirm rough volatility, even after ironing out the effect of noise in the volatility proxy. In those data we consistently locate a roughness parameter around $0.05$.

The rest of this paper is organized as follows. Section \ref{section:setting} presents the log-normal SV model and studies the properties of integrated variance within this framework. The GMM approach is introduced in Section \ref{sec:GMM_estimation}. Section \ref{sec:simulation-study-roughnessIV} examines the performance of our estimator in a Monte Carlo study. In Section \ref{section:empirical}, we apply the procedure to real data and compare our findings with the previous literature. We conclude in Section \ref{section:conclusion} and leave some theoretical derivations to the appendix. An online-only supplemental appendix contains further Monte Carlo results and empirical analysis.

\section{The setting} \label{section:setting}

We model the log-price of a financial asset, $X = (X_t)_{t \geq 0}$, as an adapted continuous-time stochastic process defined on a filtered probability space $(\Omega, \mathcal{F}, (\mathcal{F}_{t})_{t \geq 0}, \mathbb{P})$. We suppose a standard arbitrage-free market, in which asset prices are of semimartingale form \citep*[e.g.,][]{back:91a,delbaen-schachermayer:94a}. We assume $X$ can be described by an It\^{o} process:
\begin{equation} \label{eq:LogPriceDynamics}
X_{t} = X_{0} + \int_{0}^{t} \mu_{s} \text{d}s + \int_{0}^{t} \sigma_{s} \text{d}W_{s}, \qquad t \geq 0,
\end{equation}
where $X_{0}$ is $\mathcal{F}_{0}$-measurable, $\mu = ( \mu_{t})_{t \geq 0}$ is a predictable drift process, $\sigma = ( \sigma_{t})_{t \geq 0}$ is a c\`{a}dl\`{a}g volatility process and $W = (W_{t})_{t \geq 0}$ is a standard Brownian motion.

The spot variance $\sigma^{2} = (\sigma_{t}^{2})_{t \geq 0}$ is given by:
\begin{equation} \label{eq:LognormalInstantaneousVariance}
\sigma_{t}^{2} = \xi \exp \left(Y_{t} - \frac{1}{2} \kappa(0) \right), \qquad t \geq 0,
\end{equation}
where $\xi \in \Xi \subset (0, \infty)$ represents the unconditional mean of the stochastic variance, while $Y = (Y_{t})_{t \geq 0}$ is a mean zero stationary Gaussian process with autocovariance function (acf) $\kappa(u) = \text{cov}(Y_{0},Y_{u}) = \kappa_{ \phi}(u)$, $u \geq 0$, parameterized by $\phi \in \Phi \subset \mathbb{R}^{p}$. We assume $\Xi$ and $\Phi$ are compact, so that $\Theta = \Xi \times \Phi \subset \mathbb{R}^{p+1}$ is compact, and write $\theta = (\xi, \phi) \in \Theta$.

Note that we do not restrict volatility to a Markovian or semimartingale setting.\footnote{The log-normal distribution is invariant to (non-zero) power transformations. This implies that ``volatility,'' which in financial economics is more often associated with the standard deviation---or the square-root of the variance---is also log-normal if the variance is (and vice versa). Hence, volatility is applied loosely here to mean either variance or standard deviation. In the text, the concrete meaning is apparent from the context and should not be the cause of any confusion.} This is not a problem for absence of arbitrage and existence of an equivalent risk-neutral probability measure (although it is not unique in our setup), since volatility is not the price of a tradable asset.\footnote{In 2004, CBOE launched derivatives on the VIX index, which is a weighted average of implied volatility from a basket of S\&P 500 options, rendering volatility at least partially tradable (see, e.g., the white paper available at https://www.cboe.com/micro/vix/vixwhite.pdf).} The main limitation of \eqref{eq:LognormalInstantaneousVariance} is that it has continuous sample paths. Our model therefore excludes jumps in volatility, which are empirically relevant \citep*[e.g.,][]{bandi-reno:16a, eraker-johannes-polson:03a, todorov-tauchen:11a}. We leave a theoretical development of our framework in presence of volatility jumps to future research.\footnote{In the supplemental appendix, we simulate log-volatility as a jump-diffusion process driven by a standard Brownian motion with Hurst exponent $H = 1/2$. We implement the GMM procedure developed for the continuous sample path version of the model on it and observe that jumps in volatility do not cause any discernible spurious roughness. An explanation of this effect is offered there.}

To maintain a streamlined exposition, we also exclude a jump component in $X$. The theory should at least be robust to the addition of finite-activity jumps, but then one needs to pay attention to the practical implementation.\footnote{The realized variance introduced in \eqref{equation:rv} is not a consistent estimator for integrated variance in the presence of price jumps. In such instances, the truncated realized variance of \citet*{mancini:09a} or the bipower variation of \citet*{barndorff-nielsen-shephard:04b} can be exploited. Below, we derive a bias correction for realized variance and bipower variation to account for measurement error. The correction developed for realized variance also applies to the truncated version, at least so long as the jumps are finite-activity, \citep*[see, e.g.,][Proposition 1, which does not require semimartingale volatility and therefore is applicable in our setting]{li-todorov-tauchen:17a}.}

The integrated variance on day $t$ is defined as:
\begin{equation} \label{equation:iv}
IV_{t} = \int_{t-1}^{t} \sigma_{s}^{2} \text{d}s, \qquad t \in \mathbb{N},
\end{equation}
and holds information on the parameters of the model. Our estimation procedure exploits this by measuring integrated variance on a daily basis. The subscript $t$ indicates the end of a day. We later substitute integrated variance with a realized measure of volatility computed from intraday high-frequency data of $X$.

We exploit the dynamics of integrated variance in this paper. This is in contrast to the application of spot variance in previous work, e.g. \citet*{bennedsen-lunde-pakkanen:22a, fukasawa-takabatake-westphal:22a, gatheral-jaisson-rosenbaum:18a}. While spot variance is more ideal, it is associated with numerous pitfalls in practice. Firstly, spot variance estimation requires ultra high-frequency data, which may not readily be available. Even if they are, sampling at the highest frequency may induce an accumulation of microstructure noise that can distort the analysis \citep*[e.g.,][]{hansen-lunde:06b}. The calculation of microstructure noise-robust estimators is complicated and they suffer from poor rates of convergence \citep*[e.g.,][]{barndorff-nielsen-hansen-lunde-shephard:08a, jacod-li-mykland-podolskij-vetter:09a, zhang-mykland-ait-sahalia:05a}. Secondly, intraday spot variance is driven by a pronounced deterministic diurnal pattern, which needs to be controlled for if the properties of the underlying stochastic process are to be uncovered \citep*{andersen-bollerslev:97b, andersen-bollerslev:98b}. Working at the daily frequency sidesteps this problem. Thirdly, spot variance estimators converge at a slow rate---relative to estimators of integrated variance---and, in the context of our model, often lack associated CLTs. The smoothing entailed by integrating spot variance overcomes this issue to some extent.

\subsection{Moment structure of integrated variance}

In this section, we derive the moment structure of integrated variance in \eqref{equation:iv} within the framework of the general log-normal SV process \eqref{eq:LogPriceDynamics} -- \eqref{eq:LognormalInstantaneousVariance}. This serves as the foundation of our GMM procedure to estimate the parameters of the model.

The starting point is the moment conditions:
\begin{equation} \label{equation:moment-condition}
\mathbb{E} \big[g(IV_{t}; \theta_{0} ) \big] = 0,
\end{equation}
for a measurable function $g$, where $\theta_{0} \in \Theta$ is the true parameter vector.

The expectation in \eqref{equation:moment-condition} can be hard to calculate for most choices of $g$. While spot variance is log-normal, the integrated variance is a sum of correlated log-normal random variables, and figuring out its distribution is a highly nontrivial statistical problem.\footnote{The distribution of a sum of correlated log-normals can be approximated by the Fenton-Wilkinson method, where one replaces it with a single log-normal random variable. However, the approximation is often inaccurate.} This also renders maximum likelihood estimation of the model complicated.

As the next theorem illuminates, the computation in \eqref{equation:moment-condition} is manageable for raw moments, where the order of the expectation operator and volatility integral can be reversed.

\begin{theorem} \label{theo:MomentsIVGeneral} Suppose that \eqref{eq:LogPriceDynamics} -- \eqref{eq:LognormalInstantaneousVariance} hold. Then, the integrated variance process $(IV_{t})_{t \in \mathbb N}$ is stationary with the following first and second-order moment structure:
\begin{align} \label{equation:iv-moment}
\begin{split}
\mathbb{E} [IV_{t}] &= \xi, \\[0.10cm]
\mathbb{E} [IV_{t}IV_{t+ \ell}] &= \xi^{2} \int_{0}^{1}(1 - y) \Big[ \exp \big( \kappa( \ell + y) \big) + \exp \big( \kappa( | \ell - y |) \big) \Big] \text{\upshape{d}}y,
\end{split}
\end{align}
for $\ell \in \mathbb{N} \cup \{ 0 \}$.

The third and fourth moment of integrated variance are:
\begin{align}
\mathbb{E} [IV_{t}^{3}] &= \,\,\, 6 \xi^{3} \int_{0}^{1} \int_{0}^{x} (1-x) f(x,y) \mathrm{d}y \mathrm{d}x, \\[0.10cm]
\mathbb{E} [IV_{t}^{4}] &= 24 \xi^{4} \int_{0}^{1} \int_{0}^{x} \int_{0}^{y} (1-x) g(x,y,z) \mathrm{d}z \mathrm{d}y \mathrm{d}x,
\end{align}
where
\begin{align}
\begin{split}
f(x,y) &= \exp \big( \kappa(|x-y|)+ \kappa(|x|)+ \kappa(|y|) \big), \\[0.10cm]
g(x,y,z) &= \exp \big( \kappa(|x-y|)+ \kappa(|x-z|)+ \kappa(|y-z|)+ \kappa(|x|)+ \kappa(|y|)+ \kappa(|z|) \big). \\[0.10cm]
\end{split}
\end{align}
In addition, suppose the following conditions hold:
\begin{enumerate}
\item[(a)] \label{item:lem-limit} $\lim_{\ell \rightarrow \infty} \kappa( \ell) = 0$,
\item[(b)] \label{item:lem-cont} there exists an integrable function $\phi : [-1,1] \rightarrow \mathbb{R}$ such that $\displaystyle \frac{ \kappa( \ell+y)}{ \kappa( \ell)} \rightarrow \phi(y)$ as $\ell \rightarrow \infty$ for any $y \in [-1,1]$,
\item[(c)] \label{item:lem-bound} ${\displaystyle \limsup_{\ell \rightarrow \infty} \sup_{y \in [-1,1]} \bigg| \frac{ \kappa( \ell+y)}{ \kappa(\ell)} \bigg|< \infty}$.
\end{enumerate}
Then, as $\ell \rightarrow \infty$:
\begin{equation}
\mathbb{E} \big[(IV_{t}- \xi)(IV_{t+ \ell}- \xi)] \sim \xi^{2} \kappa( \ell) \int_{-1}^{1}(1-|y|) \phi(y) \text{\upshape{d}}y,
\end{equation}
where $f( \ell) \sim g( \ell)$ denotes asymptotic equivalence, i.e. $f( \ell)/g( \ell) \rightarrow 1$ as $\ell \rightarrow \infty$.
\end{theorem}
The proof of Theorem \ref{theo:MomentsIVGeneral} exploits a convenient change of variables technique, which we introduce in Lemma \ref{lemma:integration} in the appendix. This allows to express the $r$th moment of integrated variance as an ($r$-1)-dimensional integral.

Without proof, we conjecture a general formula for $\mathbb{E} [IV_{t}^{r}]$ by induction:
\begin{equation}
\mathbb{E} [IV_{t}^{r}] = r! \, \xi^{r} \int_{0}^{1} \int_{0}^{x_{1}} \cdots \int_{0}^{x_{r-2}} (1-x_{1}) g(x_{1}, \ldots, x_{r-1}) \mathrm{d} x_{r-1} \ldots \mathrm{d}x_{1},
\end{equation}
where $r \geq 2$ is a positive integer, and
\begin{equation}
g(x_{1}, \ldots, x_{r-1}) = \exp \bigg( \sum_{1 \leq i < j \leq r-1} \kappa(|x_{i} - x_{j}|) + \sum_{i = 1}^{r - 1} \kappa(|x_{i}|) \bigg).
\end{equation}
In many realistic models, the above integrals do not possess analytic solutions and need to be approximated or solved numerically. In the fractional SV model entertained below, the acf of spot variance has a sharp incline near the origin, which gets steeper for smaller $H$. In that model, even the third moment can be prohibitively time-consuming to calculate, at least within the grasp of our computing powers. This makes higher-order moments unwieldy to work with in practice. Furthermore, because the distribution of integrated variance is generally heavy-tailed and highly right-skewed, such moments are also hard to estimate. As if that was not bad enough, in the noisy proxy setting the moment conditions incorporate estimation error. The updated expectations involve a convolution of the integrated variance and measurement error. This is tough to deal with even for raw moments. As a practical compromise, our estimation procedure therefore relies on low-order moments.

In the last part of Theorem \ref{theo:MomentsIVGeneral}, condition (a) is a minimal condition that is necessary and sufficient for a stationary Gaussian process to be ergodic. This follows from the classical result of \citet*{maruyama:49a}. It is evidently true for the models in this paper. Condition (b) and (c) are more technical and restrict the oscillation of $\kappa( \ell)$ as $\ell \rightarrow \infty$.

As an illustration, suppose there exists $\ell_{0}>0$ such that
\begin{equation}\label{equation:covariance-asymptotic}
\kappa( \ell) = \ell^{- \beta} e^{- \rho \ell} L( \ell), \quad \ell \geq \ell_{0},
\end{equation}
where $\beta \geq 0$ and $\rho \geq 0$ with $\min(\beta, \rho)>0$, and for some slowly varying function $L : (0, \infty) \rightarrow (0, \infty)$, i.e. a function for which $\displaystyle \lim_{x \rightarrow \infty} \frac{L(cx)}{L(x)} = 1$ for all $c>0$. For example, if $L(x)$ converges to a strictly positive limit as $x \rightarrow \infty$, then it is evidently slowly varying. Appealing to the uniform convergence theorem for slowly varying functions \citep*[Theorem 1.5.2]{bingham-goldie-teugels:89a}, condition (a)--(c) can be verified with $\phi(y) = e^{- \rho y}$.

\subsection{The fractional SV model} \label{section:fractionalSV}

While the theoretical results developed in this paper apply more broadly to general log-normal SV models, in our illustrations---both the simulation study and empirical application---we zoom in on the fractional SV (fSV) model, in which the volatility is assumed to be the exponential of a fractional Ornstein-Uhlenbeck (fOU) process:
\begin{equation} \label{equation:fOU}
Y_{t} = \nu \int_{0}^{t} e^{- \lambda (t-s)} \text{\upshape{d}}B_{s}^{H}, \quad t \geq 0,
\end{equation}
where $\nu, \lambda > 0$, and $B^{H} = (B_{t}^{H})_{t \geq0}$ is a fractional Brownian motion (fBm) with Hurst index $H \in (0,1)$.\footnote{A fBm started at the origin ($B_{0}^{H} = 0$) with Hurst exponent $H \in (0,1)$ is a centered Gaussian process with covariance function $E \big[B_{t}^{H} B_{s}^{H} \big] = \frac{1}{2} \Big( |t|^{2H} + |s|^{2H} - |t-s|^{2H} \Big)$. It can be constructed as a weighted infinite moving average of past increments to a standard Brownian motion following the \citet*[][Definition 2.1]{mandelbrot-vanness:68} representation: $\displaystyle B_{t}^{H} = \frac{1}{ \Gamma(H+1/2)} \bigg\{ \int_{- \infty}^{0} \Big[ (t-s)^{H-1/2} - (-s)^{H-1/2} \Big] \text{d}B_{s} + \int_{0}^{t} (t-s)^{H-1/2} \text{d}B_{s} \bigg\}$, where $\Gamma(\cdot)$ is the Gamma function. The process is self-similar of index $H$ and has stationary but not independent increments (except for $H = 1/2$). Its sample paths are locally H\"{o}lder continuous up to order $H$. As readily seen, the fBm reduces to a standard Brownian motion for $H = 1/2$.}
This is a standard log-normal SV model for $H = 1/2$. The fractional version was introduced by \citet*{comte-renault:98a} in a long-memory setting $(H>1/2)$ and recently in a rough setting $(H<1/2)$ by \citet*{gatheral-jaisson-rosenbaum:18a}. Our main interest in this model is that we aim to assess the empirical level of the Hurst exponent, when all parameters are estimated jointly. Hence, in our Monte Carlo analysis we pay particular attention to the accuracy of the estimation error in $H$.\footnote{In Appendix \ref{appendix:gamma-bss}, we provide the associated analysis for GMM estimation of an alternative log-normal SV model, namely the Brownian semistationary process.}

The acf of the fOU model was derived in the following convenient form by \citet*{garnier-solna:18a}, which---when adapted to the present parameterization of the model and assuming $Y_{0}$ is drawn from its stationary distribution---can be expressed as:
\begin{equation} \label{equation:fSV-covariance}
\kappa( \ell) = \frac{ \nu^{2}}{ 2 \lambda^{2H}} \left( \frac{1}{2} \int_{- \infty}^{ \infty} e^{-|y|} | \lambda \ell + y|^{2H} \mathrm{d}y - | \lambda \ell|^{2H} \right), \quad \ell \geq 0.
\end{equation}
For $H = 1/2$ this expression reduces to:
\begin{equation} \label{equation:exponential-acf}
\kappa( \ell) = \frac{ \nu^{2}}{ 2 \lambda} e^{- \lambda \ell},
\end{equation}
which is the standard formula for the log-normal SV model.

As apparent from \eqref{equation:exponential-acf}, the acf decays exponentially for $H = 1/2$. However, \eqref{equation:fSV-covariance} has a hyperbolic rate of decay for other values of $H$, since $\kappa( \ell) = O \big( \ell^{2(H-1)} \big)$ as $\ell \rightarrow \infty$. This follows from \citet*{cheridito-kawaguchi-maejima:03a}. These properties are transferred to the integrated variance process by the last result in Theorem 2.1. Moreover, the acf is integrable for $H \leq 1/2$ but not integrable for $H > 1/2$.

The variance of $Y_{t}$ is:
\begin{equation} \label{equation:variance-of-Y}
\kappa(0) = \frac{ \nu^{2}}{2 \lambda^{2H}} \Gamma(1+2H).
\end{equation}
The fSV model conforms to \eqref{equation:covariance-asymptotic} with $\beta = 2(1-H)$ and $\rho=0$  for $H \in (0,1/2) \cup (1/2,1)$, by Theorem 2.3 of \citet*{cheridito-kawaguchi-maejima:03a}, and with $\beta = 0$ and $\rho = \lambda$ for $H = 1/2$. We stress that $L$ need not be given in closed form, as the proof of \eqref{equation:covariance-asymptotic} amounts to checking that $L( \ell) \equiv \frac{ \kappa( \ell)}{ \ell^{- \beta} e^{- \rho \ell}}$ is slowly varying, based on the asymptotic behavior of $\kappa( \ell)$ as $\ell \rightarrow \infty$.

\section{GMM Estimation} \label{sec:GMM_estimation}

In this section, for technical convenience we define all processes also for negative time indices.

\subsection{Assumptions and examples}

As described above, the spot variance $\sigma^{2} = ( \sigma_{t}^{2})_{t \in \mathbb{R}}$ depends on the parameter vector $\theta = (\xi, \phi) \in \Theta$. We write $\mathbb{P}_{\theta}$ for the probability measure induced by $\theta$ and $\mathbb{E}_{ \theta}$ is the corresponding expectation operator (we sometimes omit $\theta$ when there is no risk of confusion). We denote by $\mathcal{F}^{\sigma}$ the $\sigma$-algebra generated by $\sigma^{2}$ or, equivalently, $Y$.

We now introduce our main assumption about $Y$.

\begin{assumption}\label{assum:gauss} The Gaussian process $Y$ and its covariance function $\kappa$ satisfy the following conditions:
\begin{itemize}
\item[(i)] $Y$ has continuous sample paths for any $\phi \in \Phi$,
\item[(ii)] $(u, \phi) \mapsto \kappa_{ \phi}(u)$ is a continuous function.
\end{itemize}
\end{assumption}
Condition (i) is natural for stationary Gaussian processes, since if $Y$ was discontinuous, its sample paths would in fact be \emph{unbounded} almost surely \citep*{belyaev:61a}. Condition (ii) ensures that the moments of the model are continuous with respect to $\theta$. It is worth pointing out that neither these conditions nor the stationarity of $Y$ say much about the long-term behavior of volatility. We return to this in Assumption \ref{assum:error}. In the fOU process, condition (i) has been shown in Proposition 3.4 of \citet*{kaarakka-salminen:11a}, while condition (ii) can be verified by applying the dominated convergence theorem to \eqref{equation:fSV-covariance}.

As the main object of interest, integrated variance, is not observable in practice, it needs to be estimated. We strive for a general framework applicable to realized measures \emph{at large}, while still remaining analytically tractable. We postulate that we observe a noisy proxy of $IV_{t}$ given by
\begin{equation} \label{eq:IV-proxy}
\widehat{IV}_{t} = IV_{t} + \varepsilon_{t},
\end{equation}
where $\varepsilon_{t}$ is a random variable capturing the measurement error, which needs to adhere to a set of stylized technical conditions given in Assumption \ref{assum:error}. Such a high-level approach to describing measurement error between a realized measure and the corresponding integrated variance is reminiscent to what \citet*{patton:11a} uses for the analysis of noisy volatility proxies in the context of forecast evaluation.

To formalize our assumptions about the process $( \varepsilon_{t})_{t \in \mathbb{Z}}$, we require a filtration $\mathcal{F}^{ \sigma, \varepsilon}_{t} = \mathcal{F}_t^{ \sigma} \vee \mathcal{F}_{t}^{ \varepsilon}$, where $\mathcal{F}_t^{ \varepsilon} = \sigma \big( \{ \varepsilon_{t}, \varepsilon_{t-1}, \ldots \} \big)$, $t \in \mathbb{Z}$, is the $\sigma$-algebra generated by the errors up to time $t$. We also introduce a key assumption about the joint long-term behavior of $(IV_{t})_{t \in \mathbb{Z}}$ and $( \varepsilon_{t})_{t \in \mathbb{Z}}$.

\begin{assumption}\label{assum:error} The processes $(IV_{t})_{t \in \mathbb{Z}}$ and $( \varepsilon_{t})_{t \in \mathbb{Z}}$ satisfy the following conditions:
\begin{itemize}
\item[(i)] $(IV_{t}, \varepsilon_{t})_{t \in \mathbb{Z}}$ is a stationary and ergodic process under $\mathbb{P}_{ \theta}$ for any $\theta \in \Theta$,
\item[(ii)] $\theta \mapsto c( \theta) \equiv \mathbb{E}_{ \theta}[ \varepsilon_{1}^{2}]$ is a finite-valued, continuous function on $\Theta$,
\item[(iii)] $\mathbb{E}_{ \theta}[ \varepsilon_{t} \mid \mathcal{F}^{ \sigma, \varepsilon}_{t-1}] = 0$ for any $t \in \mathbb{Z}$ and any $\theta \in \Theta$.
\end{itemize}
\end{assumption}
The $\sigma^{2}$ and $(IV_{t})_{t \in \mathbb{Z}}$ processes readily inherit the stationarity and ergodicity of $Y$. The joint ergodicity of $(IV_{t}, \varepsilon_{t})_{t \in \mathbb{Z}}$ is more delicate, since even if $(IV_{t})_{t \in \mathbb{Z}}$ and $( \varepsilon_{t})_{t \in \mathbb{Z}}$ are ergodic on their own and mutually independent, it does not follow that $(IV_{t}, \varepsilon_{t})_{t \in \mathbb{Z}}$ is ergodic \citep*[see][Exercise 5.13]{lindgren:06a}. But if additionally $(IV_{t})_{t \in \mathbb{N}}$ or $( \varepsilon_{t})_{t \in \mathbb{Z}}$ is weakly mixing, then their joint ergodicity holds \citep*[see][Exercise 5.14]{lindgren:06a}. That said, in practical applications the mutual independence of $(IV_{t})_{t \in \mathbb{Z}}$ and $( \varepsilon_{t})_{t \in \mathbb{Z}}$ is too strong an assumption, since the level of measurement error typically depends on the underlying level of volatility. Condition (iii) is a martingale-difference property for $( \varepsilon_{t})_{t \in \mathbb{Z}}$, which implies $\mathbb{E}_{ \theta}[ \varepsilon_{t}] = 0$, i.e. the proxy $\widehat{IV}_{t}$ is unbiased. This is obviously a somewhat stylized assumption, which is not exactly satisfied by many realized measures. However, we can expect it to hold approximately and, in any case, it is crucial for the analytical tractability of the setup.\footnote{\citet*{meddahi:02a} studies realized variance under a class of log-normal volatility models including drift. He finds the mean of the measurement error to be negligible at a 5-minute sampling frequency.}

We now demonstrate that a particular structural form of the error term, $\varepsilon_{t}$, conveniently accommodates concrete realized measures as proxies in an approximate sense while satisfying Assumption \ref{assum:error}. More specifically, let
\begin{equation} \label{eq:error-term-structural}
\varepsilon_{t} = h \big(Z_{t},( \sigma_{s+t-1}^{2})_{s \in [0,1]} \big), \quad t \in \mathbb{Z},
\end{equation}
where $Z_{t}$, $t \in \mathbb{Z}$, are i.i.d. $d$-dimensional random vectors, for some $d \in \mathbb{N}$, that are independent of $\mathcal{F}^{ \sigma}$ and $h : \mathbb{R}^{d} \times C([0,1]) \rightarrow \mathbb{R}$ is a continuous functional such that
\begin{equation} \label{eq:exp-zero}
\mathbb{E}_\theta[h(Z_1,f)] = 0,
\end{equation}
for any $\theta \in \Theta$ and $f \in C([0,1])$. Then condition (i) in Assumption \ref{assum:error} can be proved using standard ergodic theory arguments, see, e.g., \citet*[Section 5.4]{lindgren:06a}, while condition (iii) is readily implied by \eqref{eq:exp-zero}. We can check condition (ii) on a case-by-case basis below by computing $c( \theta) = \mathbb{E}_{ \theta}[ \varepsilon_{1}^{2}]$ explicitly and invoking condition (ii) of Assumption \ref{assum:gauss} to establish its continuity in $\theta$.

In the following examples, we construct $\varepsilon_{t}$ and $Z_{t}$ only for $t \in \mathbb{N}$, but we can extend them to negative indices by stationarity.

\begin{example}[Realized variance, CLT approximation] \label{exm:CLT}
Suppose that we estimate the integrated variance $IV_{t}$ with the realized variance \citep*[see, e.g.,][]{andersen-bollerslev:98a,barndorff-nielsen-shephard:02a}:
\begin{equation} \label{equation:rv}
RV_{t}^{n} = \sum_{i=1}^{n} \big(X_{t-1+ \frac{i}{n}} - X_{t-1+ \frac{i-1}{n}} \big)^{2},
\end{equation}
for any $t \in \mathbb{N}$. Under standard technical conditions, the central limit theorem (CLT)
\begin{equation} \label{eq:RV-CLT}
\sqrt{n}(RV_{t}^{n} - IV_{t}) \xrightarrow[n \rightarrow \infty]{d_{ \mathrm{st}}} \sqrt{2} \int_{t-1}^{t} \sigma_{s}^{2} \mathrm{d} B_{s}^{ \perp},
\end{equation}
holds jointly for all $t \in \mathbb{N}$, where $\stackrel{d_{ \mathrm{st}}}{\longrightarrow}$ denotes stable convergence in distribution and $(B^{ \perp}_{s})_{s \geq 0}$ is a Brownian motion independent of $X$ and $\sigma$. Note that the limiting random variables $\sqrt{2} \int_{t-1}^{t} \sigma_{s}^{2} \mathrm{d} B_{s}^{ \perp}$, $t \in \mathbb{N}$, are conditionally independent given $\mathcal{F}^{ \sigma}$ with
\begin{equation}
\sqrt{2} \int_{t-1}^{t} \sigma_{s}^{2} \mathrm{d} B_{s}^{ \perp} \, \bigg| \, \mathcal{F}^{ \sigma} \sim N(0,2IQ_{t}), \quad t \in \mathbb{N},
\end{equation}
where
\begin{equation}
IQ_t = \int_{t-1}^t \sigma_{s}^{4} \mathrm{d}s
\end{equation}
is the integrated quarticity. Thus,
\begin{equation}
Z_{t} = \frac{ \int_{t-1}^{t} \sigma_{s}^{2} \mathrm{d} B_{s}^{ \perp}}{IQ_{t}^{1/2}} \sim N(0,1), \quad t \in \mathbb{N},
\end{equation}
are both mutually independent and independent of $\mathcal{F}^\sigma$.

Informally, the CLT \eqref{eq:RV-CLT} says that, for any $t \in \mathbb{N}$,
\begin{equation} \label{eq:clt-asy}
RV_{t}^{n} \stackrel{d}{ \approx} IV_{t} + \left( \frac{2}{n} IQ_{t} \right)^{1/2} Z_{t}
\end{equation}
for large $n$, where ``$\stackrel{d}{ \approx}$'' denotes approximate equality in distribution, as used, e.g., in \citet*[Section 1.2]{zhang-mykland-ait-sahalia:05a}. Thus, for any $t \in \mathbb{N}$, the proxy $\widehat{IV}_{t} = IV_{t} + \varepsilon_{t}$ with $\varepsilon_{t} = \left( \frac{2}{n} IQ_{t} \right)^{1/2} Z_{t}$ approximates $RV_{t}^{n}$ for large $n$. Such a proxy is analogous to what \citet*{fukasawa-takabatake-westphal:22a} employ in their estimation framework. We can represent the error term as $\varepsilon_{t}$ in the form \eqref{eq:error-term-structural} using the continuous functional
\begin{equation}
h(z,f) =  \bigg( \frac{2}{n} \int_{0}^{1} f(s)^{2} \mathrm{d}s \bigg)^{1/2} z, \quad z \in \mathbb{R}, \quad f \in C([0,1]).
\end{equation}
Then \eqref{eq:exp-zero} holds given that $Z_{1} \sim N(0,1)$. We can compute $c( \theta)$ explicitly using Tonelli's theorem. The expression is reported in Table \ref{table:error-term} and $\theta \mapsto c( \theta)$ is evidently continuous under Assumption \ref{assum:gauss}.
\end{example}

\begin{example}[Realized variance, no drift or leverage effect] \label{exm:indep}
In general, the measurement error $RV_{t}^{n} - IV_{t}$ is analytically hard to analyze unless we resort to asymptotic approximation with $n \rightarrow \infty$ as in Example \ref{exm:CLT} (see also the comments below). However, in a simple specific case, we can actually work with the exact error $\varepsilon_{t} = RV_{t}^{n} - IV_{t}$, that is $\widehat{IV}_{t} = RV_{t}^{n}$, without losing analytical tractability.

Namely, suppose that the log-price $X = (X_{t})_{t \geq 0}$ of the asset follows a drift-free It\^{o} process
\begin{equation}
X_{t} = X_{0} + \int_{0}^{t} \sigma_{s} \mathrm{d}W_{s}, \quad t \geq 0,
\end{equation}
where  $W = (W_{t})_{t \geq 0}$ is a standard Brownian motion independent of $\mathcal{F}^{ \sigma}$, i.e. ruling out any dependence between $W$ and the spot variance process $\sigma^{2}$, stemming from the leverage effect for instance \citep*[e.g.,][]{christie:82a}. Then, for any $t \in \mathbb{N}$,
\begin{align} \label{eq:sum-gamma}
\begin{split}
RV_{t}^{n} - IV_{t} &= \sum_{i=1}^{n} \Bigg( \bigg( \int_{ \frac{i-1}{n}+t-1}^{ \frac{i}{n}+t-1} \sigma_{s} \mathrm{d}W_{s} \bigg)^{2} - \int_{ \frac{i-1}{n}+t-1}^{ \frac{i}{n}+t-1} \sigma_{s}^{2} \mathrm{d}s \Bigg) \\[0.10cm]
&= \sum_{i=1}^{n} (Z_{t,i}^{2}-1) \int_{ \frac{i-1}{n}}^{ \frac{i}{n}} \sigma_{s+t-1}^{2} \mathrm{d}s,
\end{split}
\end{align}
where
\begin{equation} \label{eq:exact-gauss}
Z_{t,i} = \frac{ \int_{ \frac{i-1}{n}+t-1}^{ \frac{i}{n}+t-1} \sigma_{s} \mathrm{d}W_{s}}{ \big( \int_{ \frac{i-1}{n}+t-1}^{ \frac{i}{n}+t-1} \sigma_{s}^{2} \mathrm{d}s \big)^{1/2}}, \quad t \in \mathbb{N}, \quad i = 1, \ldots,n.
\end{equation}
Since $W$ is independent of $\mathcal{F}^{ \sigma}$, conditional on $\mathcal{F}^{ \sigma}$ the random variables $Z_{t,i}$, $t \in \mathbb{N}$, $i = 1, \ldots, n$, are mutually independent and follow a standard normal distribution. Consequently, they are i.i.d. standard normal also unconditionally and independent of $\mathcal{F}^{ \sigma}$.

Thanks to \eqref{eq:sum-gamma}, we can represent the measurement error $\varepsilon_{t} = RV_{t}^{n} - IV_{t}$ in the form \eqref{eq:error-term-structural}
via the functional
\begin{equation}
h \big((z_{1}, \ldots, z_{n}),f \big) = \sum_{i=1}^{n} (z_{i}^{2} -1) \int_{ \frac{i-1}{n}}^{ \frac{i}{n}} f(s) \mathrm{d}s, \quad (z_{1}, \ldots, z_{n}) \in \mathbb{R}^{n}, \quad f \in C([0,1]),
\end{equation}
and i.i.d. random vectors
\begin{equation}
Z_{t} = (Z_{t,1}, \ldots, Z_{t,n}), \quad t \in \mathbb{N},
\end{equation}
with components given by \eqref{eq:exact-gauss}, so that $d = n$. The property \eqref{eq:exp-zero} then holds, while an integral functional representation of $c( \theta)$ is given in Table \ref{table:error-term} and its continuity in $ \theta$ follows from the dominated convergence theorem under Assumption \ref{assum:gauss}.
\end{example}
In the presence of a leverage effect, the moments of the measurement error $\varepsilon_{t} = RV_{t}^{n} - IV_{t}$ can be analyzed using Malliavin calculus and chaos expansions, see, e.g., \citet*{peccati-taqqu:11a}. However, the resulting formulae are not in any form convenient for numerical implementation, which is why we do not pursue this approach further here.

\begin{example}[Bipower variation, CLT approximation] \label{exm:bipower} In the context of Example \ref{exm:CLT}, the realized variance can be substituted with the bipower variation estimator of \citet*{barndorff-nielsen-shephard:04b}, which is defined as:
\begin{equation} \label{equation:bv}
BV_{t}^{n} = \frac{ \pi}{2} \sum_{i=2}^{n} \big|X_{t-1+ \frac{i}{n}}-X_{t-1+ \frac{i-1}{n}} \big| \big|X_{t-1+ \frac{i-1}{n}}-X_{t-1+ \frac{i-2}{n}} \big|, \quad n \in \mathbb{N},
\end{equation}
for any $t \in \mathbb{N}$. Under standard technical conditions
\begin{equation} \label{eq:BV-CLT}
\sqrt{n}(BV_{t}^{n} - IV_{t}) \xrightarrow[n \rightarrow \infty]{d_{ \mathrm{st}}} \sqrt{ \frac{ \pi^{2}}{4}+ \pi-3} \int_{t-1}^{t} \sigma_{s}^{2} \mathrm{d} B_{s}^{ \perp},
\end{equation}
jointly for all $t \in \mathbb{N}$, where the structure of the limit is identical to the one in \eqref{eq:RV-CLT}. $BV_{t}^{n}$ is then approximated for large $n$ by the proxy $\widehat{IV}_{t} = IV_{t} + \varepsilon_{t}$ with error term $\varepsilon_{t} = \Big( \frac{ \frac{ \pi^{2}}{4}+ \pi-3}{n}IQ_{t} \Big)^{1/2}Z_{t}$, where $Z_t$, $t \in \mathbb{N}$, are as in Example \ref{exm:CLT}. Retracing the arguments in Example \ref{exm:CLT}, we can then show that $\varepsilon_{t}$ can be cast in the form \eqref{eq:exact-gauss}, so Assumption \ref{assum:error} holds.
\end{example}

\begin{table}
\begin{center}
\caption{Formulae for $c( \theta) = \mathbb{E}_{ \theta}[ \varepsilon_{1}^{2}]$. \label{table:error-term}}
\smallskip
\begin{tabular}{llcc}
\hline \hline
Proxy & Setting & Example & $c( \theta) = c( \xi, \phi)$ \\
\hline
Realized variance & CLT & \ref{exm:CLT} & $\frac{2 \xi^{2}}{n} \exp \big( \kappa_{ \phi}(0) \big)$ \\[0.10cm]
                  & No drift or leverage & \ref{exm:indep} & $\frac{4 \xi^{2}}{n} \int_{0}^{1} (1-y) \exp \big( \kappa_{ \phi} ( \frac{y}{n}) \big) \text{d}y$ \\[0.10cm]
Bipower variation & CLT & \ref{exm:bipower} & $\frac{\big( \frac{ \pi^{2}}{4}+ \pi -3 \big) \xi^{2}}{n} \exp \big( \kappa_{ \phi}(0) \big)$ \\[0.10cm]
\hline \hline
\end{tabular}
\smallskip
\begin{scriptsize}
\parbox{0.92\textwidth}{\emph{Note.}
In the case of Example \ref{exm:indep}, we use Theorem \ref{theo:MomentsIVGeneral} to derive the expression.}
\end{scriptsize}
\end{center}
\end{table}

The appealing feature of bipower variation is that it remains consistent for integrated variance in the presence of price jumps. However, in the latter setting the CLT in \eqref{eq:BV-CLT} ceases to hold and the limiting distribution of $BV_{t}^{n}$ is not mixed Gaussian \citep*{vetter:10a}. Even without price jumps the asymptotic theory of bipower variation generally needs to impose smoothness conditions on volatility, typically in the form of an It\^{o} semimartingale structure or, possibly, a long-memory process. To our knowledge, roughness is ruled out or at least remains undetermined. Hence, the approximate bias correction from Example \ref{exm:bipower} should be applied with caution. Nevertheless, it can serve as a heuristic double check.\footnote{In the supplemental appendix, we repeat the simulation analysis and empirical application with bipower variation. The results do not change much compared to those for realized variance reported below.}

Table \ref{table:error-term} summarizes the derivations for the measurement error. We reiterate that Example \ref{exm:CLT} and Example \ref{exm:bipower} are based on $n \rightarrow \infty$, while any implementation is always done with a finite $n$. In contrast, Example \ref{exm:indep} is valid for all $n$ in absence of drift and leverage. In the simulations, we compare both expressions for realized variance to gauge their impact on the estimation.

\subsection{Consistency} \label{section:consistency}

In this section, we turn to the consistency of our GMM estimator. We take $\theta \in \Theta$, $t \in \mathbb{Z}$ and $\ell \in \mathbb{Z}$ and introduce the moment structure of the $IV_{t}$ process, which is defined by:
\begin{equation}
g_{0}^{(1)}( \theta) = \mathbb{E}_{ \theta}[IV_{t}], \quad g_{0}^{(2)}( \theta) = \mathbb{E}_{ \theta}[IV_{t}^{2}], \quad
g_{ \ell}( \theta) = \mathbb{E}_{ \theta}[IV_{t} IV_{t- \ell}],
\end{equation}
for $\ell = 1, \ldots, k$, which we collect in the column vector
\begin{equation}
G( \theta) = \big(g_{0}^{(1)}( \theta),g_{0}^{(2)}( \theta),g_{1}( \theta), \ldots,g_{k}( \theta) \big)^{ \intercal}.
\end{equation}
Here, $^{ \intercal}$ is the transpose operator. We also define
\begin{align}
\begin{split}
\mathbb{IV}_{t} &= (IV_{t}, IV_{t}^{2}, IV_{t} IV_{t-1}, \ldots,IV_{t} IV_{t-k} \big)^{ \intercal}, \\[0.10cm]
\widehat{ \mathbb{IV}}_{t} &= \big( \widehat{IV}_{t}, \widehat{IV}_{t}^{2}, \widehat{IV}_{t} \widehat{IV}_{t-1}, \ldots, \widehat{IV}_{t} \widehat{IV}_{t-k} \big)^{ \intercal},
\end{split}
\end{align}
which by condition (i) of Assumption \ref{assum:error} are stationary and ergodic processes.

By condition (ii)-(iii) of Assumption \ref{assum:error}:
\begin{align} \label{equation:proxy-moment}
\begin{split}
\mathbb{E}_{ \theta} \big[ \widehat{IV}_{t} \big] &= g_{0}^{(1)}( \theta), \\[0.10cm]
\mathbb{E}_{ \theta} \big[ \widehat{IV}_{t} \widehat{IV}_{t- \ell} \big] &=
\begin{cases}
g_{0}^{(2)}( \theta) + c( \theta), & \ell = 0, \\[0.10cm]
g_{ \ell}( \theta), & \ell \neq 0.
\end{cases}
\end{split}
\end{align}
The noisy proxy changes the second moment in \eqref{equation:proxy-moment} compared to integrated variance in \eqref{equation:iv-moment}. The term $c( \theta)$ is induced by the noise. This observation was also made by \citet*{bollerslev-zhou:02a}. They account for the error-in-variables by including an additive nuisance parameter to the second-order moment condition, which is estimated independently of the structural parameters. In contrast, we incorporate the error variation directly into the model as a function of $\theta$, which avoids the need of an extra parameter.

The first- and other second-order moments are unbiased, due to the linearity of the expectation operator and because the errors are mean zero and serially uncorrelated. In principle, we can thus avoid the negative impact of measurement errors by excluding $g_{0}^{(2)}( \theta)$ from the moment selection. More generally, however, it is often preferable to add the variance or absolute value to the moment conditions, because low-order moments are highly informative about the parameters of SV models \citep*{andersen-sorensen:96a}. To avoid any systematic deviance in the estimated values of the parameters, it is then necessary to correct the appropriate entries in the moment vector as detailed above (dealing with the measurement error is of course much more complicated for the absolute value than for the square).
	
We propose to compare the sample moments of $\widehat{IV}_{t}$ to a corrected moment function
\begin{equation} \label{equation:bias-correction}
G_{c}( \theta) = G( \theta) + \big(0,c( \theta),0, \ldots,0 \big)^{ \intercal}.
\end{equation}
We define a random function:
\begin{equation}
\widehat{m}_{T}( \theta) = \frac{1}{T} \sum_{t=1}^{T} \widehat{\mathbb{IV}}_{t} - G_{c}( \theta),
\end{equation}
which, in view of \eqref{equation:proxy-moment}, has
\begin{equation}
\mathbb{E}_{ \theta_{0}} \big[ \widehat{m}_{T}( \theta) \big] = G_{c}( \theta_0) - G_{c}( \theta) \equiv m( \theta),
\end{equation}
so that
\begin{equation}
m( \theta_{0}) = 0.
\end{equation}
Our GMM estimator is then given by
\begin{equation} \label{equation:gmm-estimator}
\widehat{ \theta}_{T} = \underset{ \theta}{ \argmin} \ \widehat{m}_{T}( \theta)^{ \intercal} \mathbb{W}_{T} \widehat{m}_{T}( \theta),
\end{equation}
where $\mathbb{W}_{T}$ is a random $(k+2) \times (k+2)$ weight matrix.

We need additional conditions for the consistency of $\widehat{ \theta}_{T}$. Firstly, we introduce a standard assumption about the limiting behavior of $\mathbb{W}_{T}$.

\begin{assumption} \label{assum:weight-mat-limit}
$\mathbb{W}_{T} = A_{T}^{ \intercal} A_{T}$ for a random $(k+2) \times (k+2)$ matrix $A_{T}$, which under $\mathbb{P}_{ \theta_{0}}$ converges almost surely to a non-random matrix $A$ as $T \rightarrow \infty$.
\end{assumption}
Secondly, we assume that the parameters are identifiable.

\begin{assumption} \label{assumption:identification}
$A m( \theta) = 0$ if and only if $\theta = \theta_{0}$.
\end{assumption}
Assumption \ref{assumption:identification} is a standard identification condition in GMM that ensures uniqueness of the solution. It is hard to check when moments are not given in algebraic form, see e.g. \citet*{barboza-viens:17a, newey-mcfadden:94a} in the GMM setting or \citet*{corradi-distaso:06a, todorov:09a} in the context of estimation of SV models.

An alternative route to inspect identification is to perform a rank test on the Jacobian matrix. This amounts to verifying that $\nabla_{ \theta} A m( \theta_{0})$ has full column rank. The latter is equivalent to Assumption \ref{assumption:identification} if the moment conditions are linear in the parameters. In our setting, this is not the case (except for the mean). Hence, the rank condition can only help to identify parameters locally in a neighbourhood of a solution candidate. Nevertheless, developing a formal rank test for local identification in the fSV model presents a challenge. This requires that we derive the asymptotic distribution of $T^{-1/2} \nabla_{ \theta} A_{T} \widehat{m}_{T}( \theta)$, e.g. \citet*{wright:03a}. We do not pursue the idea here but leave it to future research. Instead, in Appendix \ref{appendix:identification} we offer some alternative insights about identification in the fSV model based on the notion of equality of an sequence of moment conditions. As noted there, $\xi$ and $H$ are identified, whereas $\nu$ and $\lambda$ are identified only through their ratio.

\begin{theorem} \label{consistency:gmmcorrection}
Suppose Assumptions \ref{assum:gauss} -- \ref{assumption:identification} hold. As $T \rightarrow \infty$
\begin{equation}
\widehat{ \theta}_{T} \as \theta_{0}.
\end{equation}
\end{theorem}
In the above, our analysis assumed that the number of observations per day, $n$, is fixed and then relies on the noisy proxy idea. Now, following \citet*{bollerslev-zhou:02a, corradi-distaso:06a, todorov:09a}, we also cover the theory of the GMM estimator in a double-asymptotic setting with $T \rightarrow \infty$ and $n \rightarrow \infty$.

To this end, we denote with $V_{t}^n$ some consistent realized measure of integrated variance (e.g., realized variance, bipower variation, or truncated realized variance). For fixed $k \in \mathbb{N}$, we set
\begin{align}
\mathbb{V}_{t}^{n} &= \big(V_{t}^{n}, (V_{t}^{n})^{2}, V_{t}^{n} V_{t-1}^{n}, \ldots, V_{t}^{n} V_{t-k}^{n} \big)^{ \intercal},
\end{align}
with associated sample moments
\begin{equation}
\widetilde{m}_{n,T}( \theta) = \frac{1}{T} \sum_{t=1}^{T} \mathbb{V}_{t}^{n} - G( \theta),
\end{equation}
where we employ the moments of $G( \theta)$ instead of the corrected version $G_{c}( \theta)$, which is of no consequence for the following result since $n \rightarrow \infty$.

Then,
\begin{equation}
\widetilde{ \theta}_{n,T} = \argmin_{ \theta \in \Theta} \widetilde{m}_{n,T}( \theta)^{ \intercal} \mathbb{W}_{n,T} \widetilde{m}_{n,T}( \theta),
\end{equation}
is our GMM estimator.

In this setting, we replace Assumption \ref{assum:error} with the following requirement.

\begin{assumption} \label{assum:doubleasymp} The processes $(IV_{t})_{t \in \mathbb{Z}}$ and $(V_{t}^{n})_{t \in \mathbb{Z}, n \in \mathbb{N}}$ admit the following:
\begin{itemize}
\item[(i)] $(IV_{t})_{t \in \mathbb{Z}}$ is a stationary and ergodic process under $\mathbb{P}_{ \theta}$ for any $\theta \in \Theta$,
\item[(ii)] $\sup_{t \in \mathbb{Z}} \mathbb{E}_{ \theta_{0}}[(V_{t}^{n} - IV_{t})^{2}] \rightarrow 0$ as $n \rightarrow \infty$.
\end{itemize}
\end{assumption}

\begin{theorem} \label{theorem:GMM}
Suppose Assumptions \ref{assum:gauss} and \ref{assum:weight-mat-limit} -- \ref{assum:doubleasymp} hold. As $T \rightarrow \infty$ and $n \rightarrow \infty$
\begin{equation}
\widetilde{ \theta}_{n,T} \cp \theta_{0}.
\end{equation}
\end{theorem}
This result is equivalent to Theorem 1 (and Corollary 1) in \citet*{todorov:09a} and Theorem 1 in \citet*{corradi-distaso:06a}.

In Appendix \ref{remark:uniform-bound}, we show that under a boundedness condition on the drift and volatility:
\begin{equation}
\sup_{t \in \mathbb{Z}} \mathbb{E}[(RV_{t}^{n}-IV_{t})^{2}] \leq C n^{-1},
\end{equation}
for some $C>0$. Hence, in this setting Assumption \ref{assum:doubleasymp} holds for $RV_{t}^{n}$.

\subsection{Asymptotic normality}

To establish asymptotic normality of our GMM estimator, for technical reasons we assume that under $\mathbb{P}_{ \theta_{0}}$ the Gaussian process $Y$ admits a causal moving average representation
\begin{equation} \label{equation:moving-average}
Y_{t} = \int_{- \infty}^{t} K(t-u) \text{d}B_{u}, \quad t \in \mathbb{R},
\end{equation}
for a two-sided standard Brownian motion $B = (B_{t})_{t \in \mathbb{R}}$ and measurable kernel $K : (0, \infty) \rightarrow \mathbb{R}$ such that $\int_{0}^{ \infty} K(u)^{2} \text{d}u < \infty$. We can extend $K$ to the entire real line by setting $K(u) = 0$ for $u \leq 0$ when necessary. \eqref{equation:moving-average} is not restrictive, since a stationary Gaussian process admits such a representation under weak conditions. In particular, the moving average structure exists if and only if $Y$ satisfies a mild, albeit technical, condition known as \textit{pure non-determinism}, see \citet*[Satz 5]{karhunen:50a} and \citet*[Section 4.5]{dym-mckean:76a}. The fSV model incorporated in this paper adheres to form \eqref{equation:moving-average}, since the fOU process has such a representation \citep*[e.g.,][]{barndorff-nielsen-basse-o-connor:11a}.

The asymptotic behavior of $K(u)$ as $u \rightarrow \infty$ governs the long-term memory of $Y$. To derive the asymptotic normality of our GMM estimator, we need to constrain that memory.

\begin{assumption} \label{assumption:memory}
$K(u) = O(u^{- \gamma})$ as $u \rightarrow \infty$ for some $\gamma > 1$.
\end{assumption}
\citet*{garnier-solna:18a} showed that the kernel in the moving average representation of the fOU process is asymptotically, as $u \rightarrow \infty$, proportional to $u^{H-3/2}$ for $H (0, 1/2) \cup (1/2, 1)$. Moreover, the definition in \eqref{equation:fOU} with $H = 1/2$ implies $K(u) = \nu e^{-\lambda u} = o(u^{- \gamma})$, for all $\gamma > 1$, as $u \rightarrow \infty$. Thereby, the fSV model requires $H \leq 1/2$ to be covered by Assumption \ref{assumption:memory}, allowing for rough volatility but ruling out the long-memory version.

We believe the constraint in Assumption \ref{assumption:memory} is nearly optimal in the sense that if $K(u)$ is asymptotically proportional to $u^{- \gamma}$ for $\gamma \in (0,1)$, e.g. with the fSV model for $H > 1/2$, then asymptotic normality with a standard rate of convergence ceases to hold. In this case, we can show that the expression for the asymptotic covariance matrix in our central limit theorem of Proposition \ref{prop:clt} does not converge. It is possible that a non-central limit theorem with a non-standard rate of convergence holds, as commonly encountered in the realm of long-memory processes, see, e.g., \citet*{taqqu:75a}. Proving such an extension is rather non-trivial, however, and therefore beyond the scope of the present exposition.\footnote{It may be possible to extend the permissible range of $H$ values by first-differencing the realized measure
series. However, since the present estimator is consistent for any $H \in (0,1)$ and the empirical results presented below do not contain estimates near the region $H>1/2$, we do not pursue this extension here.}

Additionally, we introduce stronger assumptions about the error process $(\varepsilon_{t})_{t \in \mathbb{Z}}$. In what follows, we write $\| X \|_{L^{2}( \mathbb{P}_{ \theta})} = \mathbb{E}_{ \theta} \big[X^{2} \big]^{1/2}$ for any square integrable random variable $X$ and work with the filtrations $\mathcal{F}^{ \widehat{ \mathbb{IV}}}_{t} = \sigma \big\{ \widehat{ \mathbb{IV}}_{t}, \widehat{ \mathbb{IV}}_{t-1}, \ldots \big \}$ and $\mathcal{F}^{B, \varepsilon}_{t} = \sigma\{ \varepsilon_{t}, \varepsilon_{t-1}, \ldots \} \vee \sigma\{ B_{u} : u \leq t \}$, $t \in \mathbb{Z}$.

\begin{assumption} \label{assum:error-clt} The processes $B$ and $( \varepsilon_{t})_{t \in \mathbb{Z}}$ satisfy the following conditions:
\begin{itemize}
\item[(i)] $\mathbb{E}[ \varepsilon_{1}^{4}] < \infty$,
\item[(ii)] \label{item:epsilon-memory} $\left\| \mathbb{E}_{ \theta_{0}} \Big[ \varepsilon_{r}^{2} \mid \mathcal{F}_{0}^{ \widehat{ \mathbb{IV}}} \Big] - \mathbb{E}_{ \theta_{0}}[ \varepsilon_{1}^{2}] \right\|_{L^{2}( \mathbb{P}_{ \theta_{0}})} = O(r^{- \gamma+1/2})$ as $r \rightarrow \infty$,
\item[(iii)] \label{item:W-incr} $B$ has independent increments with respect to $( \mathcal{F}^{B, \varepsilon}_{t})_{t \in \mathbb{Z}}$ (i.e., for any $t \in \mathbb{Z}$ the process $(B_{u}-B_{t})_{u \geq t}$ is independent of $\mathcal{F}^{B, \varepsilon}_{t}$).
\end{itemize}
\end{assumption}

Condition (ii) constrains the memory in the squared measurement error. In the high-frequency setting, the measurement error usually depends on volatility (as demonstrated in Example \ref{exm:CLT} -- \ref{exm:bipower}). So here Assumption \ref{assumption:memory} implies condition (ii), see Proposition \ref{prop:exm-memory} in Appendix \ref{app:prop-clt}. Condition (iii) ensures that the measurement error does not anticipate future increments of the driving Brownian motion $B$, which is not very restrictive.

The next result presents the CLT for the sample mean of our statistic.

\begin{proposition}\label{prop:clt}
Suppose that Assumptions \ref{assum:gauss}, \ref{assum:error}, \ref{assumption:memory}, and \ref{assum:error-clt} hold. Then, as $T \rightarrow \infty$, under $\mathbb{P}_{ \theta_{0}}$,
\begin{equation} \label{eq:clt-conv}
\sqrt{T} \widehat{m}_{T} \big( \theta_{0} \big) \overset{d}{ \longrightarrow} N \big(0, \Sigma_{ \widehat{\mathbb{IV}}} \big),
\end{equation}
where
\begin{equation}
\Sigma_{ \widehat{ \mathbb{IV}}} = \sum_{ \ell = -\infty}^{ \infty} \Gamma_{ \widehat{ \mathbb{IV}}}( \ell),
\end{equation}
and
\begin{equation}
\Gamma_{ \widehat{ \mathbb{IV}}}( \ell) = \mathbb{E}_{ \theta_{0}} \Big[ \big( \widehat{ \mathbb{IV}}_{1} - G_{c}( \theta_{0}) \big) \big( \widehat{ \mathbb{IV}}_{1+ \ell} - G_{c}( \theta_{0}) \big)^{ \intercal} \Big].
\end{equation}
\end{proposition}
The proof of Proposition \ref{prop:clt} builds on Theorem 4.10 from \citet*[][Section 4.2]{merlevede-peligrad-utev:19a}, where a martingale approximation CLT is derived for an $L^{2}$-mixingale of size -1/2. Theorem 24.5 in \citet*{davidson:94a} states a related CLT for an $L^{1}$-mixingale of size -1. In other words, we impose a slightly stronger moment restriction to attain a better size condition. The rate improvement is important in our context, as a mixingale of size -1 does not cover the fOU with $H \in (0,1/2)$. Moreover, $L^{1}$ integrability does not add anything here, because the process is always square integrable. That is, in our setting \citet*{merlevede-peligrad-utev:19a} deliver the right foundation. Conversely, with a CLT that only requires a mixingale of size -1, we need to additionally assume that $\gamma > 3/2$ in Assumption \ref{assumption:memory}. As the decay exponent of the fOU is $H - 3/2$, such a theory does not here apply for any $H > 0$.

A final assumption for the CLT of our GMM estimator is presented next. Here, we introduce the function $\mathbf{g}: \mathbb{R}^{k+2} \times \Theta \rightarrow \mathbb{R}$ via $\mathbf{g}(x, \theta) = x-G_{c}( \theta)$.

\begin{assumption} \label{Ass:GMMAssumptionsCLT}
It holds that

\begin{itemize}
\item[(i)] $\theta_{0}$ is an interior point of $\Theta$.
\item[(ii)] $J^{ \intercal} \mathbb{W} J$ is non-singular, where $J = \mathbb{E}_{ \theta_{0}} \Big[ \nabla_{ \theta} \mathbf{g}( \widehat{ \mathbb{IV}}_{1}, \theta_{0}) \Big]$ and $\mathbb{W}=A^{ \intercal}A$.
\item[(iii)] The function $\theta \mapsto \mathbf{g}(x, \theta)$ is continuously differentiable. In addition, $\mathbb{E}_{ \theta_{0}} \Big[ \| \mathbf{g}( \widehat{ \mathbb{IV}}_{1}, \theta_{0}) \|^{2} \Big] < \infty$ and $\mathbb{E}_{ \theta_{0}} \Big[ \sup_{ \theta \in \Theta} \| \nabla_{ \theta} \mathbf{g}( \widehat{ \mathbb{IV}}_{1},\theta) \| \Big] < \infty$.
\end{itemize}
\end{assumption}
Now, we are ready to present the asymptotic distribution of $\widehat{ \theta}_{T}$.

\begin{theorem} \label{Thm:GMMCLTCorrection}
Suppose Assumptions \ref{assum:gauss} --  \ref{assumption:identification} and \ref{assumption:memory} -- \ref{Ass:GMMAssumptionsCLT} hold. As $T \rightarrow \infty$,
\begin{equation} \label{equation:clt}
\sqrt{T} \big( \widehat{ \theta}_{T}- \theta_{0} \big) \overset{d}{ \longrightarrow} N \big(0,(J^{ \intercal} \mathbb{W} J)^{-1} J^{ \intercal} \mathbb{W} \Sigma_{ \widehat{ \mathbb{IV}}} \mathbb{W} J (J^{ \intercal} \mathbb{W} J)^{-1} \big).
\end{equation}
\end{theorem}
As usual, to minimize the asymptotic variance in \eqref{equation:clt} and derive the efficient GMM estimator---for given moment conditions---we choose an optimal weight matrix as (the inverse of) a consistent estimator of $\Sigma_{ \widehat{ \mathbb{IV}}}$. We propose a \citet*{newey-west:87a} HAC-type estimator:
\begin{equation} \label{equation:hac}
\widehat{ \Sigma}_{T} = \widehat{ \Gamma}(0)+ \sum_{\ell = 1}^{T-1} w( \ell/L) \left[ \widehat{ \Gamma}( \ell)+ \widehat{ \Gamma}( \ell)^{ \intercal} \right],
\end{equation}
where $w$ is a weight function, $L = o(T^{1/2})$ is the lag length, and
\begin{equation} \label{equation:autocovariance}
\widehat{ \Gamma}( \ell) = \frac{1}{T} \sum_{t=1}^{T-l} \big( \widehat{ \mathbb{IV}}_{t} - G_{c}( \widehat{ \theta}_{T}) \big) \big( \widehat{ \mathbb{IV}}_{t+ \ell} - G_{c}( \widehat{ \theta}_{T}) \big).
\end{equation}
Following \citet{davidson:20a}, we impose a weak regularity condition on $w$ that is fulfilled by a large class of common choices in practice, such as the Bartlett or Parzen kernel.
\begin{assumption} \label{assumption:kernel}
It holds that
\begin{itemize}
\item[(i)] $w(0) = 1$ and $\sup_{x \geq 0} |w(x)| < \infty$,
\item[(ii)] $w$ is continuous at 0,
\item[(iii)] $\int_{0}^{ \infty} \bar{w}(x) \mathrm{d}x < \infty$, where $\bar{w}(x)=\sup_{y \geq x} |w(y)|$.
\end{itemize}
\end{assumption}

\begin{theorem} \label{Thm:GMMCLTHAC}
Suppose Assumptions \ref{assum:gauss} --  \ref{assumption:identification} and \ref{assumption:memory} -- \ref{assumption:kernel} hold. As $T \rightarrow \infty$,
\begin{equation}
\widehat{ \Sigma}_{T} \overset{ \mathbb{P}}{ \longrightarrow} \Sigma_{ \widehat{ \mathbb{IV}}}.
\end{equation}
\end{theorem}
Setting $\mathbb{W}_{T} = \widehat{ \Sigma}_T^{-1}$, it follows that
\begin{equation} \label{equation:efficient-gmm}
\sqrt{T} \big( \widehat{ \theta}_{T}- \theta_{0} \big) \overset{d}{ \longrightarrow} N \Big(0, \big(J^{ \intercal} \Sigma_{ \widehat{ \mathbb{IV}}}^{-1} J \big)^{-1}  \Big).
\end{equation}
To make inference on the parameters we plug-in $\widehat{ \Sigma}_T^{-1}$ in \eqref{equation:efficient-gmm}.

By standard results for quadratic forms of multivariate normal variables, the minimized objective function value times the sample size has an asymptotic chi-square distribution:
\begin{equation} \label{equation:J-test}
\mathcal{J}_{ \text{HS}} = T \widehat{m}_{T}( \widehat{ \theta}_{T} )^{ \intercal} \widehat{ \Sigma}_T^{-1} \widehat{m}_{T}( \widehat{ \theta}_{T} ) \overset{d}{ \longrightarrow} \chi^{2}(k-p+1).
\end{equation}
where $k-p+1$ is the number of overidentifying restrictions. This facilitates a Sargan-Hansen omnibus specification test of the model.

To finish this section, we study the CLT of our GMM estimator in the double-asymptotic setting, where $T \rightarrow \infty$ and $n \rightarrow \infty$, such that the discretization error is negligible.

\begin{assumption} \label{assum:doubleasympCLT} The processes $(IV_{t})_{t \in \mathbb{Z}}$ and $(V_{t}^{n})_{t \in \mathbb{Z}, n \in \mathbb{N}}$ satisfy the following conditions:
\begin{itemize}
\item[(i)] $(IV_{t})_{t \in \mathbb{Z}}$ is a stationary and ergodic process under $\mathbb{P}_{ \theta}$ for any $\theta \in \Theta$,
\item[(ii)] $\sup_{t \in \mathbb{Z}} \mathbb{E}_{ \theta_{0}} \Big[ \big( \sqrt{T}(V_{t}^{n}-IV_{t}) \big)^{2} \Big] \rightarrow 0$ as $T \rightarrow \infty$ and $n \rightarrow \infty$.
\end{itemize}
\end{assumption}
In this setting, we again introduce a HAC-type estimator $\widehat{ \Sigma}_{n,T}$, which merely substitutes $\widehat{ \mathbb{IV}}_{t} - G_{c}$ with $\mathbb{V}_{t}^{n} - G$ in \eqref{equation:autocovariance}, and then we take $\mathbb{W}_{n,T} = \widehat{ \Sigma}_{n,T}^{-1}$.

\begin{theorem} \label{Thm:GMMCLTdouble}
Suppose Assumptions \ref{assum:gauss}, \ref{assum:weight-mat-limit} --\ref{assumption:identification}, \ref{assumption:memory}, \ref{Ass:GMMAssumptionsCLT} -- \ref{assum:doubleasympCLT} hold. As $T \rightarrow \infty$ and $n \rightarrow \infty$,
\begin{equation}
\sqrt{T} \big( \widetilde{ \theta}_{n,T} - \theta_{0} \big) \overset{d}{ \longrightarrow} N \big(0, (\widetilde{J}^{ \intercal} \Sigma_{\mathbb{IV}}^{-1} \widetilde{J})^{-1} \big),
\end{equation}
where
\begin{equation}
\Sigma_{ \mathbb{IV}} = \sum_{ \ell = - \infty}^{ \infty} \Gamma_{ \mathbb{IV}}( \ell),
\end{equation}
and
\begin{equation}
\Gamma_{ \mathbb{IV}}( \ell) = \mathbb{E}_{ \theta_{0}} \big[( \mathbb{IV}_{1}-G( \theta_{0}))( \mathbb{IV}_{1+ \ell}-G( \theta_{0}))^{ \intercal} \big]
\end{equation}
with $\widetilde{J} = \mathbb{E}_{ \theta_{0}} \big[ \nabla_{ \theta} \mathbf{g}( \mathbb{IV}_{1}, \theta_{0}) \big]$.
\end{theorem}
In closing, we remark that Theorem \ref{theorem:GMM} and \ref{Thm:GMMCLTdouble} also hold for the bias-corrected estimator $\hat{ \theta}_{T}$ provided it fulfills the additional conditions imposed by Assumption \ref{assum:doubleasymp}(ii) for consistency or Assumption \ref{assum:doubleasympCLT}(ii) for asymptotic normality, as $T \rightarrow \infty$ and $n \rightarrow \infty$. If that is so, the bias correction vanishes under the double-asymptotic framework. In practice, even if $n$ is large and the bias term is small, we recommend to add the correction as a precaution, as nothing is lost by doing so.

\bigskip

\noindent \textbf{Remark} \ \textit{In comparison to \citet*{fukasawa-takabatake-westphal:22a} (FTW), our article differs in several aspects. Firstly, FTW employ quasi-likelihood based on a Whittle approximation, while we propose GMM estimation. Secondly, FTW show consistency, but we further prove asymptotic normality of our estimator. Thirdly, many of our theoretical results apply more broadly to general log-normal SV models, whereas FTW concentrate on the fSV model. Fourthly, our approach is fully parametric, whereas FTW is a semi-parametric procedure, where the dynamic of the drift is left unspecified. Moreover, their implementation relies on the first-difference of log-realized variance to enforce mean zero. Hence, FTW estimate two parameters (volatility-of-volatility and the Hurst index), whereas we recover a four-dimensional parameter vector (including the mean and mean reversion coefficient). Fifthly, FTW also control for measurement error in the volatility proxy, but they employ a double-asymptotic CLT for log-realized variance with the length of the time interval approaching zero. In the implementation, however, they still rely on daily realized variance with 5-minute sampling frequency, thereby mixing properties of spot and integrated variance.}

\section{Simulation study} \label{sec:simulation-study-roughnessIV}

In the above, we developed a full-blown large sample GMM framework for estimation of the log-normal fSV model with a general Hurst index. We now review the finite sample properties of our approach. The aim is to assess the accuracy of the procedure in a realistic setup. We inspect both the infeasible setting where estimation is based on integrated variance and a feasible implementation relying on realized variance. For the latter, we gauge the performance both with and without the quarticity correction in \eqref{equation:bias-correction}.

We assume the log-price, $X_{t}$, evolves as a driftless It\^{o} process:
\begin{equation}
\text{d}X_{t} = \sigma_{t} \text{d}W_{t}, \quad t \geq 0, \label{equation:sim-X}
\end{equation}
with initial condition $X_{0} \equiv 0$. Here, $\sigma_{t}$ is the spot volatility and $W_{t}$ is a standard Brownian motion. We discretize $X$ via an Euler scheme.

The log-variance, $Y_{t} = \ln( \sigma_{t}^{2})$, is a fOU process:
\begin{equation}
\text{d}Y_{t} = -\lambda(Y_{t}- \eta) \text{d}t + \nu \text{d}B_{t}^{H} \label{equation:sim-Y},
\end{equation}
where $B_{t}^{H}$ is a fbM. We assume $W \Perp B^{H}$, so there is no leverage effect.

The SDE in \eqref{equation:sim-Y} is solved to get a more convenient expression for $Y$:
\begin{equation}
Y_{t} = \eta + (Y_{t- \Delta} - \eta) e^{- \lambda \Delta} + \nu \int_{t- \Delta}^{t} e^{- \lambda (t-s)} \text{d}B_{s}^{H}.\footnote{The ploy is as always to use It\^{o}'s Lemma with the integrating factor $e^{\lambda t} Y_{t}$. The math is a bit more involved here though, since we are dealing with a fBm, where a stochastic calculus may not exist. Nevertheless, it goes through in this particular instance, see, e.g., \citet*{cheridito-kawaguchi-maejima:03a}.}
\end{equation}
The stochastic integral is approximated as $\int_{t- \Delta}^{t} e^{- \lambda (t-s)} \text{d}B_{s}^{H} \simeq e^{- \lambda \Delta / 2} (B_{t}^{H} - B_{t- \Delta}^{H})$ meaning that increments to a discretely sampled fBm are required. These can be produced in many ways to get an exact discretization, e.g. Cholesky factorization or circulant embedding \citep*[see][]{asmussen-glynn:07a}. While the former has complexity $O(x^{3})$, the latter entails a markedly lower budget of $O(x\log(x))$ and is our preferred algorithm.

We draw 10,000 independent replications of this model with a path length of $T = \text{4,000}$ days as a default.
In each simulation, the log-variance process is started at random from its stationary distribution, $Y_{0} \sim N \big( \eta, \text{var}(Y_{t}) \big)$, where $\text{var}(Y_{t})$ is given by \eqref{equation:variance-of-Y}. To get an almost continuous-time realization of the processes and minimize discretization bias, we partition $[t-1,t]$, for $t = 1, \ldots T$, into $N = \text{23,400}$ discrete points of length $\Delta = 1/N$. In the US equity market, this roughly amounts to a 16-year sample of the stock price recorded every second in a $6.5$-hour trading day.

Our procedure is inspected on several distinct sets of parameters to gauge its robustness. Throughout, we set $\eta = \ln( \xi) - 0.5 \text{var}(Y_{t})$, where $\xi = \mathbb{E}(\sigma_{t}^{2}) = 0.0225$. This ensures the unconditional mean of the variance process is identical across settings and implies an annualized standard deviation $\sigma_{t}$ of about 15\% on average, close to the aggregate level of volatility in the empirical data analyzed in Section \ref{section:empirical}. As we are particularly attentive to estimation of the Hurst index, we choose $H = [0.05,0.10,0.30,0.50,0.70]$ as in \citet*{fukasawa-takabatake-westphal:22a}, thus covering both the rough, standard and long-memory case. We calibrate $\lambda$ and $\nu$ to minimize the distance from the model-implied autocovariance of integrated variance at lag 0, 20 and 50 to the associated sample autocovariances of realized variance of the .SPX (that is, the S\&P 500) equity index, after controlling for sampling error in the noisy proxy, which are subsequently rounded off to the nearest convenient values.

The parameters are presented in Table \ref{table:sim-parameter-estimation}. A realization of the spot and integrated variance processes from each model are plotted in Figure \ref{figure:sigma}. While the pathwise properties of volatility are notably different at a microscopic scale, they are much harder to discriminate after we integrate them up to the daily horizon.

\begin{figure}[t!]
\begin{center}
\caption{Sample path of spot and integrated variance. \label{figure:sigma}}
\begin{tabular}{cc}
\small{Panel A: log(spot variance).} & \small{Panel B: log(integrated variance).} \\
\includegraphics[height=0.4\textwidth,width=0.48\textwidth]{{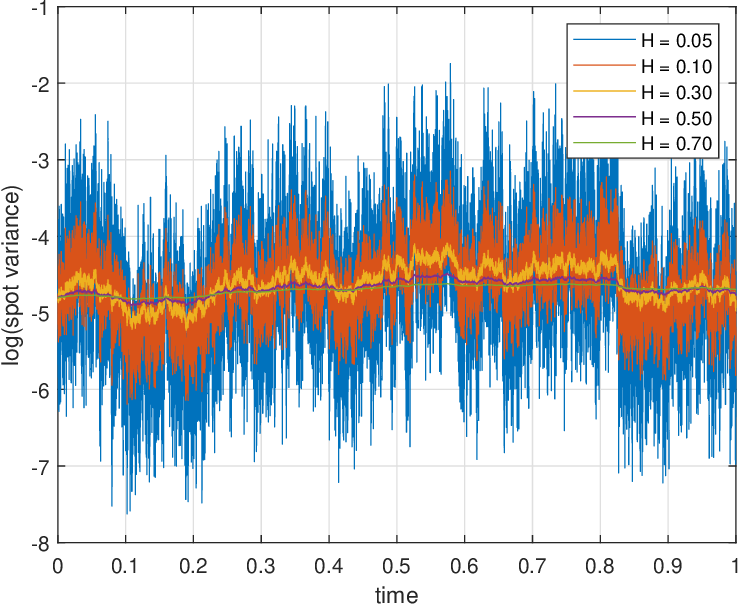}} &
\includegraphics[height=0.4\textwidth,width=0.48\textwidth]{{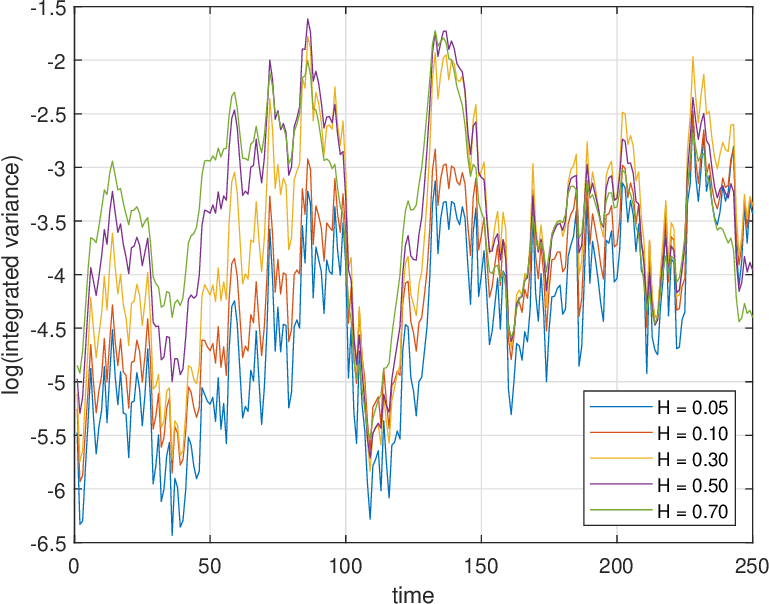}} \\
\end{tabular}
\begin{scriptsize}
\parbox{\textwidth}{\emph{Note.} In Panel A, we simulate a sample path of the log-spot variance for a single day as a function of $H$. In Panel B, we show the associated integrated variance dynamics over 250 trading days.}
\end{scriptsize}
\end{center}
\end{figure}

In addition to integrated variance we also collect realized variance with $n = 78$, i.e. with $5$-minute data. The advantage of this choice is that there is no concern about microstructure noise at this sampling frequency in practice. The input to the optimizer is therefore either $(IV_{t})_{t=1}^{T}$ or $(RV_{t}^{n})_{t=1}^{T}$. We restrict the description of the implementation details below to the feasible setting with realized variance.

The parameter vector is $\theta_{0} = (\xi, \lambda, \nu, H)$, which we estimate via the gradient-based non-linear least squares Matlab function \texttt{lsqnonlin}. We employ the default search algorithm ``trust-region-reflective'' with a tolerance level of $10^{-6}$.

We launch the engine at initial values determined as follows: $\xi$ is started at the average realized variance, i.e. $\widebar{RV} = T^{-1} \sum_{t=1}^{T} RV_{t}^{n}$. To set $H$ and $\nu$ we exploit the auxiliary two-stage procedure proposed in \citet*{gatheral-jaisson-rosenbaum:18a}, which relies on the scaling law:
\begin{equation} \label{equation:scaling-law}
\gamma_{h} \equiv \mathbb{E} \big[ | Y_{t+h} - Y_{t} |^{q} \big] \rightarrow K_{q} \nu^{q} |h|^{qH},
\end{equation}
as $h \rightarrow 0$, where $\displaystyle K_{q} = 2^{q/2} \frac{\Gamma( \frac{q+1}{2})}{ \sqrt{ \pi}}$ is the $q$'th moment of the absolute value of a standard normal random variable. This entails a log-linear relationship between $\gamma_{h}$ and $|h|$: $\ln \left( \gamma_{h} \right) = \ln \left( K_{q} \nu^{q} \right) + qH \ln \left( |h| \right)$. We employ $RV_{t}^{n}$ as a proxy for the instantaneous variance and substitute the left-hand side of \eqref{equation:scaling-law} by the sample mean:
\begin{equation}
\hat{ \gamma}_{h} = \frac{1}{T-m} \sum_{t=1}^{T-m} \lvert \ln(RV_{t+h}^{n}) - \ln(RV_{t}^{n}) \rvert^{q},
\end{equation}
for $h = 1, \dots, m$. $H$ and $\nu$ are then estimated by OLS with $q = 2$ and $m = 6$. The results are rather robust against this configuration. At last, $\lambda$ is pre-estimated such that the theoretical variance of $Y_{t}$ equals the sample variance of $\ln(RV_{t}^{n})$.

As shown in Table \ref{table:sim-parameter-estimation}, the initial values display very low variation between replications, but they are often highly biased. For instance, using $IV_{t}$ the starting point of $H$ increases with the true value, but as expected it is too high on average, whereas for $RV_{t}^{n}$ it is largely unaffected by the actual roughness of the model.

As such, there is a lot of work left to the GMM procedure. We match the sample average of $RV_{t}^{n}$ with the mean of $IV_{t}$ and the $\ell$'th sample autocovariance of $RV_{t}^{n}$ with the autocovariance structure of $IV_{t}$---available based on numerical integration of \eqref{equation:iv-moment} together with \eqref{equation:fSV-covariance}---with $\ell = [0,1,2,3,5,20,50]$. To motivate this choice, note that the short-term behavior of integrated variance and the impact of measurement error are described by rapid changes of autocovariances at a short time scale. Hence, we select 0, 1, 2, 3 and 5 into the set of lags. Additionally, to capture the medium- to long-term persistence of volatility, we include autocovariances at large lags. Their variation is slower, however, and it suffices to select a sparser subset, so we add lags 20 and 50.\footnote{In a robustness check, and to better capture the persistence of log-variance with $H = 0.7$, we also attempted to include lag 100 and 200. However, the results did not change much. This is consistent with \citet*{andersen-sorensen:96a}, who note that estimation of SV models does not always improve by adding more information.} In total, this yields eight moment conditions with four overidentifying restrictions.

While the above may seem arbitrary, it follows the previous literature both in terms of number of moment conditions per parameter and with its emphasis on first- and second-order moments \citep*[see, e.g.,][]{bollerslev-zhou:02a, corradi-distaso:06a, todorov:09a}. The selected autocovariances---taken from an infinite set of possibilities---are also meant to maximize the efficiency of our GMM estimator, while minimizing partially redundant information in the moment conditions \citep*[e.g.,][]{breusch-qian-schmidt-wyhowski:99a, hall-inoue-jana-shin:07a}. On the one hand, we perceive short-run autocovariances to be very informative about the parameters of the model. On the other hand, we expect nearby autocovariances to exhibit higher correlation. By picking a variety of mostly short lags and employing a few lags farther out as a variance reduction device, we attempt to exploit the structure of the problem as much as possible without resorting to formal econometric analysis. An alternative approach to address this problem is \citet*{carrasco-florens:00a}, who extend GMM to a continuum of moment conditions. However, the latter entails additional technical details and complicates the implementation further, and we therefore postpone it to future research.

We employ iterated GMM with a maximum of three iterations. In the first step, a preliminary estimate of $\theta_{0}$ is acquired by setting $\mathbb{W}_{T}$ equal to the identity matrix. In subsequent stages, in accordance with Theorem \ref{Thm:GMMCLTHAC}, the weight matrix is recalculated as in \eqref{equation:hac} with a Parzen kernel to ensure positive semi-definiteness and automatic lag selection based on \citet*{andrews:91a} using an approximating ARMA(1,1) structure for realized variance. The benefit of iterated GMM is its invariance to the initial weighting matrix. Typically, only a single iteration is required to converge, but when volatility is really rough, an extra computation is sometimes helpful due to poor starting values.

\begin{sidewaystable}
\begin{small}
\setlength{ \tabcolsep}{0.25cm}
\begin{center}
\caption{Parameter estimation of the log-normal fSV model in simulated data. \label{table:sim-parameter-estimation}}
\vspace*{-0.25cm}
\begin{tabular}{lccccccccc}
\hline \hline
parameter & value && \multicolumn{2}{c}{integrated variance} && \multicolumn{4}{c}{realized variance} \\
&&&&&&& uncorrected & \multicolumn{2}{c}{corrected estimate} \\
&&& initial value & estimate && initial value & estimate & exact & approximate \\ \cline{4-5} \cline{7-10}
Panel A: \\
$\xi$ & 0.0225 && 0.0225 (0.0014) & 0.0214 (0.0014) && 0.0225 (0.0014) & 0.0207 (0.0015) & 0.0213 (0.0014) & 0.0213 (0.0014) \\
$\lambda$ & 0.0050 && 0.0326 (0.0068) & 0.0072 (0.0057) && 0.0260 (0.0063) & 0.0075 (0.0062) & 0.0072 (0.0056) & 0.0081 (0.0061) \\
$\nu$ & 1.2500 && 0.4580 (0.0058) & 1.3806 (0.3543) && 0.5233 (0.0066) & 2.6245 (0.7784) & 1.3834 (0.3942) & 1.0704 (0.1129) \\
$H$ & 0.0500 && 0.2661 (0.0101) & 0.0460 (0.0202) && 0.2208 (0.0098) & 0.0168 (0.0121) & 0.0466 (0.0212) & 0.0646 (0.0171) \\
Panel B: \\
$\xi$ & 0.0225 && 0.0225 (0.0009) & 0.0216 (0.0010) && 0.0225 (0.0009) & 0.0211 (0.0011) & 0.0216 (0.0009) & 0.0216 (0.0010) \\
$\lambda$ & 0.0100 && 0.0381 (0.0065) & 0.0117 (0.0059) && 0.0286 (0.0059) & 0.0084 (0.0049) & 0.0117 (0.0055) & 0.0124 (0.0056) \\
$\nu$ & 0.7500 && 0.3683 (0.0046) & 0.7767 (0.0871) && 0.4390 (0.0055) & 1.7914 (0.5549) & 0.7750 (0.1066) & 0.7275 (0.0700) \\
$H$ & 0.1000 && 0.2998 (0.0104) & 0.0939 (0.0212) && 0.2363 (0.0099) & 0.0258 (0.0143) & 0.0950 (0.0242) & 0.1036 (0.0215) \\
Panel C: \\
$\xi$ & 0.0225 && 0.0226 (0.0031) & 0.0190 (0.0025) && 0.0226 (0.0031) & 0.0182 (0.0025) & 0.0189 (0.0025) & 0.0189 (0.0025) \\
$\lambda$ & 0.0150 && 0.0318 (0.0054) & 0.0187 (0.0150) && 0.0225 (0.0046) & 0.0098 (0.0087) & 0.0187 (0.0149) & 0.0188 (0.0149) \\
$\nu$ & 0.5000 && 0.3624 (0.0045) & 0.5019 (0.0454) && 0.4282 (0.0052) & 0.7131 (0.1794) & 0.4998 (0.0621) & 0.4982 (0.0594) \\
$H$ & 0.3000 && 0.4380 (0.0108) & 0.2781 (0.0490) && 0.3620 (0.0108) & 0.1759 (0.0530) & 0.2793 (0.0589) & 0.2802 (0.0583) \\
Panel D: \\
$\xi$ & 0.0225 && 0.0226 (0.0037) & 0.0191 (0.0031) && 0.0226 (0.0037) & 0.0168 (0.0028) & 0.0190 (0.0031) & 0.0190 (0.0031) \\
$\lambda$ & 0.0350 && 0.0424 (0.0053) & 0.0437 (0.0264) && 0.0240 (0.0040) & 0.0161 (0.0163) & 0.0473 (0.0357) & 0.0473 (0.0357) \\
$\nu$ & 0.3000 && 0.2484 (0.0032) & 0.2990 (0.0236) && 0.3320 (0.0039) & 0.4712 (0.1947) & 0.2982 (0.0332) & 0.2981 (0.0331) \\
$H$ & 0.5000 && 0.5715 (0.0106) & 0.4839 (0.0638) && 0.4256 (0.0114) & 0.2725 (0.0948) & 0.4912 (0.0933) & 0.4913 (0.0932) \\
Panel E: \\
$\xi$ & 0.0225 && 0.0226 (0.0057) & 0.0195 (0.0049) && 0.0226 (0.0057) & 0.0159 (0.0040) & 0.0193 (0.0049) & 0.0193 (0.0049) \\
$\lambda$ & 0.0700 && 0.0521 (0.0056) & 0.0828 (0.0315) && 0.0207 (0.0037) & 0.0188 (0.0258) & 0.0794 (0.0373) & 0.0794 (0.0373) \\
$\nu$ & 0.2000 && 0.1691 (0.0023) & 0.2004 (0.0147) && 0.2749 (0.0032) & 0.3571 (0.2121) & 0.2055 (0.0208) & 0.2054 (0.0208) \\
$H$ & 0.7000 && 0.6777 (0.0098) & 0.6803 (0.0670) && 0.4254 (0.0120) & 0.3079 (0.1318) & 0.6580 (0.1000) & 0.6580 (0.1000) \\
\hline \hline
\end{tabular}
\smallskip
\begin{scriptsize}
\parbox{0.98\textwidth}{\emph{Note.} 
We simulate 10,000 replications of a fractional Ornstein-Uhlenbeck process $\text{d}Y_{t} = -\lambda(Y_{t}- \eta) \text{d}t + \nu \text{d}B_{t}^{H}$ on $[0,T]$ with $T = \text{4,000}$ and a discretization step of $\Delta = 1/\text{23,400}$. 
The true model parameters $\theta_{0} = (\xi, \lambda, \nu, H)$ appear in Panel A -- E, where $\xi = e^{ \eta + 0.5\text{var}(Y_{t})}$. 
We estimate $\theta_{0}$ with the GMM procedure developed in the main text, where the theoretical mean and autovariance (at lag 0, 1, 2, 3, 5, 20, and 50) of integrated variance is matched with the sample. 
The optimizer is launched with initial values from the two-stage procedure in \citet*{gatheral-jaisson-rosenbaum:18a}. 
We report the average of the initial values and the associated parameter estimates based on integrated variance (left) and realized variance (right). 
The latter is computed both excluding (``uncorrected'') and including (``corrected'') the bias correction in \eqref{equation:bias-correction}. The form of $c(\theta)$ is available in Table \ref{table:error-term} and is either based on the no drift and leverage assumption (``exact'') or the CLT-based approximation (``approximate''). 
Standard deviation across simulations appear in parenthesis. 
}
\end{scriptsize}
\end{center}
\end{small}
\end{sidewaystable}

The results are presented in Table \ref{table:sim-parameter-estimation}. We report the mean estimate (standard error in parenthesis) both for the initial and final parameter value, where all calculations are done across replica. The left-hand side shows the outcome based on integrated variance, whereas the right-hand side is for realized variance with and without the correction in \eqref{equation:bias-correction}. As a robustness check, we gauge both the exact solution and CLT-based approximation available in Table \ref{table:error-term}. The former applies, since there is no drift nor leverage in the model.

We begin by commenting on the estimation based on integrated variance. Several interesting findings emerge. In the infeasible setting, the GMM procedure returns parameter estimates that are close to their population counterparts, thus verifying the robustness and accuracy of our approach. Across the board, the typical estimate of $H$ tracks the true value with minor deviations. We do notice a minuscule underestimation of $\xi$ and overestimation of $\lambda$. The drift parameters pull the acf of integrated variance in opposite directions, causing an offsetting impact on the moment matching. Nevertheless, both estimates remain within about a Monte Carlo standard error of the true value. $\nu$ is typically recovered with little bias, but for $H = 0.05$ the estimates are shifted upward and exhibit high variation. The intuition is that it is tougher to recover the volatility-of-volatility parameter with severe roughness.

Turning attention to the feasible results for realized variance, we record a significant deterioration in the estimation of $H$ without the quarticity adjustment. As explained, realized variance is a noisy proxy for integrated variance, which induces ``illusive roughness'' and yields $H$ estimates that are vastly below target, when the measurement error is unaccounted for. Also, $\nu$ increases while $\lambda$ decreases to compensate for this effect. Including the bias correction rectifies this problem and leads to a huge improvement in all parameter estimates. It is assuring to see how the analysis for the exact correction aligns that integrated variance. Moreover, the distinction between the exact and approximate correction is often immaterial, but for $H \leq 0.10$ notable differences start to creep in. In particular, the latter produces smaller $\nu$ and larger $H$ estimates. This indicates that the CLT-based correction may be somewhat inaccurate in finite samples when volatility is very erratic, which is not entirely unexpected. However, it is important to note that this amounts to \textit{less} roughness. Overall, our simulations suggest the CLT-based approach is accurate enough even with 5-minute sampling.

\begin{figure}[t!]
\begin{center}
\caption{Kernel density estimate of standardized $H$. \label{figure:H}}
\begin{tabular}{cc}
\small{Panel A: integrated variance.} & \small{Panel B: realized variance.} \\
\includegraphics[height=0.4\textwidth,width=0.48\textwidth]{{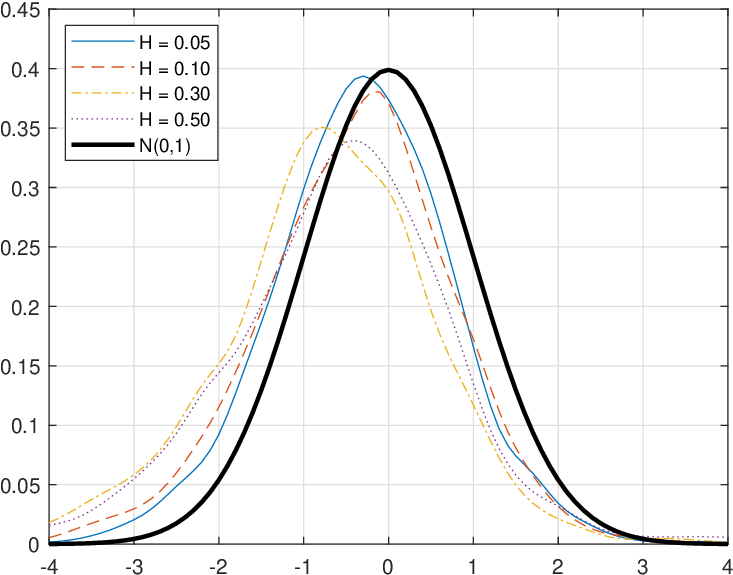}} &
\includegraphics[height=0.4\textwidth,width=0.48\textwidth]{{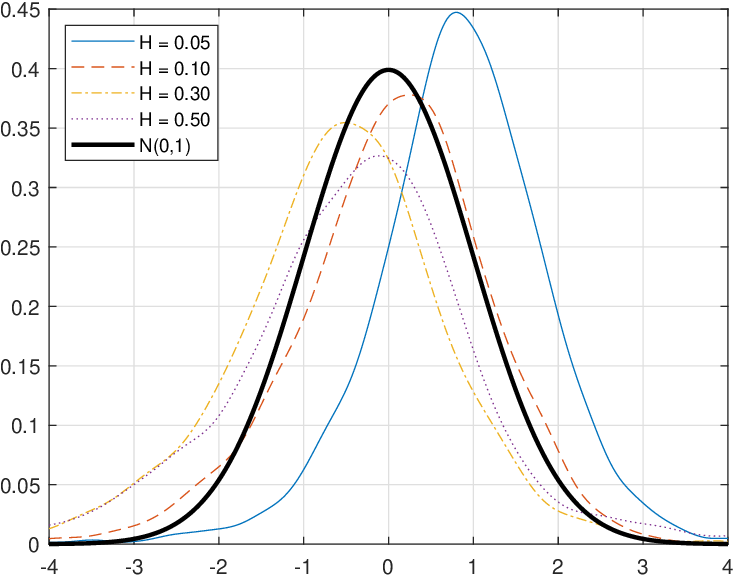}} \\
\end{tabular}
\begin{scriptsize}
\parbox{\textwidth}{\emph{Note.} We construct kernel smoothed densities of the statistic $(\widehat{H} - H) / \widehat{ \text{se}} ( \widehat{H})$, where $\widehat{ \text{se}}( \widehat{H})$ is the estimated standard error of $\widehat{H}$. In Panel A, we show the analysis for integrated variance, whereas Panel B is for realized variance with the CLT-based bias correction. The density function of a standard normal random variable is superimposed as a reference point.}
\end{scriptsize}
\end{center}
\end{figure}

To gauge the accuracy of our asymptotic theory, Figure \ref{figure:H} portrays kernel smoothed densities of the statistic $( \widehat{H} - H) / \widehat{ \text{se}} ( \widehat{H}) \overset{d}{ \longrightarrow} N(0,1)$, where $\widehat{ \text{se}}( \widehat{H})$ is standard error estimate extracted from \eqref{equation:hac} and the Jacobi matrix provided by the optimizer. The left-hand side depicts the outcome for integrated variance, whereas the right-hand side contains the normalized realized variance estimates with the CLT-based bias correction. The exact correction is omitted, as its graphs are virtually identical to those reported for integrated variance.\footnote{$H = 0.7$ is excluded from the figure as well, because the distribution theory does not cover the long-memory model. Indeed, the approximation is markedly worse in that setting.} As readily seen, the standard normal is a good description of the variation in the parameter estimates of $H$, apart from a modest off-centering with realized variance for $H = 0.05$.

In Figure \ref{figure:J-test}, we inspect the finite sample behavior of the Sargan-Hansen J-test of overidentifying restrictions. We contrast the empirical distribution function of the test statistic, $\mathcal{J}_{ \text{HS}}$, against its asymptotic distribution under the null hypothesis, which is $\chi^{2}(4)$.  Again, we show the integrated variance in Panel A and realized variance with the CLT-based bias correction in Panel B. The impression is that the test statistic lines up fairly well with the predicted values. However, the curves are located to the right of the reference line in the lower half of diagram and vice versa in the upper half. This implies an overconcentration of mass in the center of the empirical distribution, meaning the test statistic is slightly underdispersed. This leads to a conservative test with rejection rates that are below the nominal level. At the 5\% significance level, for example, the test statistic exceeds the critical value of 9.4877 about 3.6\% of the times for $H = 0.05$, which drops to 0.28\% for $H = 0.50$. Interestingly, the approximation is better for smaller values of $H$.

\begin{figure}[t!]
\begin{center}
\caption{Distribution of the J-test of overidentifying restrictions. \label{figure:J-test}}
\begin{tabular}{cc}
\small{Panel A: integrated variance.} & \small{Panel B: realized variance.} \\
\includegraphics[height=0.4\textwidth,width=0.48\textwidth]{{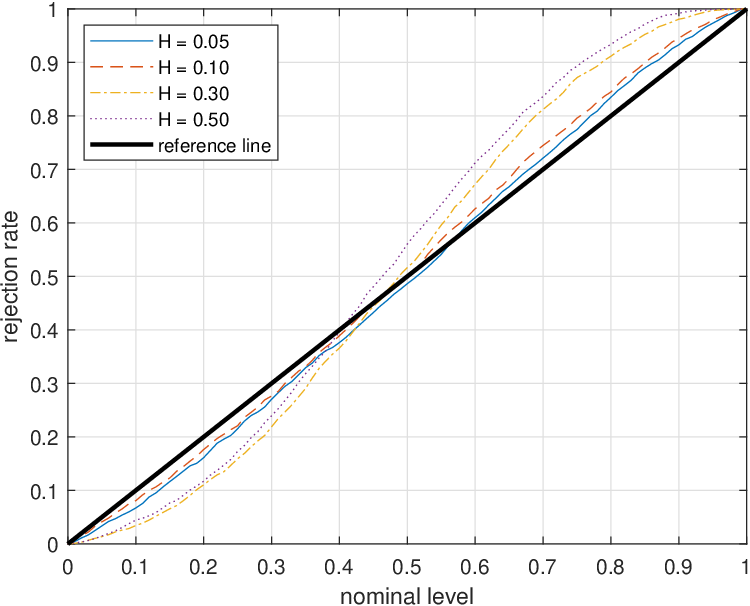}} &
\includegraphics[height=0.4\textwidth,width=0.48\textwidth]{{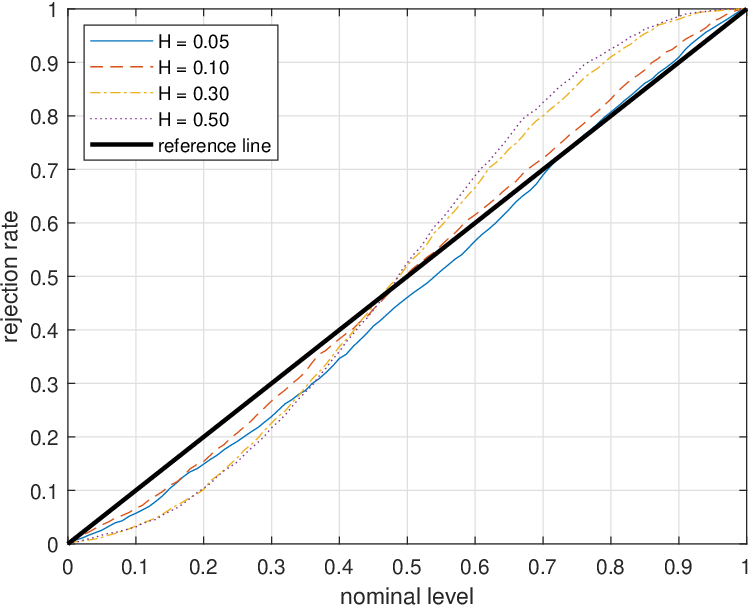}} \\
\end{tabular}
\begin{scriptsize}
\parbox{\textwidth}{\emph{Note.} We construct P-P plots to describe the variation in the J-test of overidentifying restrictions. The test statistic, $\mathcal{J}_{ \text{HS}}$, is asymptotically $\chi^{2}(4)$-distributed under the null. In Panel A, we show the analysis for integrated variance, whereas Panel B is for realized variance with the CLT-based bias correction. The 45-degree line is superimposed as a reference point.}
\end{scriptsize}
\end{center}
\end{figure}

\section{Empirical application} \label{section:empirical}

The log-normal fSV model is estimated from empirical high-frequency data covering a comprehensive selection of asset return series. We downloaded version 0.3 of the Oxford-Man Institute's ``realized library'' via: \texttt{https://realized.oxford-man.ox.ac.uk/}. The website tracks thirty-one leading stock indexes covering major financial markets. At the end of each trading day, the library is refreshed with information from Thomson Reuters DataScope Tick History and several nonparametric volatility estimators are calculated and appended to the database. We here employ the daily realized variance defined in \eqref{equation:rv}. In line with our comments above, we decide on a 5-minute sampling frequency to suppress microstructure noise. As the trading hours of each stock exchange varies, this corresponds to $n$ falling between 78 and 102 for most indexes, equivalent to an opening period of 6.5 -- 8.5 hours.

As it is, the database contains several data entries we suspect are erroneous. For instance, realized variance is occasionally identically equal to zero. While not impossible due to price discreteness, it is implausible for liquid securities. We therefore remove these from the sample. In addition, on February 7 2010 a realized variance corresponding to an annualized volatility of 250\% is reported for the .OMXHPI equity index. We searched both Factiva and Google for relevant news articles. There is nothing immediate to suggest anything out of the extraordinary occurred in the Finnish stock market that day. However, without access to the underlying high-frequency data, it is difficult to reconcile what is causing such an abnormal deviation in realized variance. The volatility in the subsequent days are back to normal levels, so it is probably an isolated outlier. Removal of this single observation is enough to raise the first-order sample autocorrelation of the .OMXHPI realized variance series from 0.209 to 0.592.

To account for such irregularities, we further discard realized variance estimates differing more than 30 mean absolute deviations from the average realized variance calculated from 50 observations on a rolling window centered around, but excluding, the data point under investigation. This is a light filter that removes none but the most egregious data. The median number of data deleted for each index with this algorithm is two. In the US equity market (.DJI, .IXIC, .RUT and .SPX) nothing from the financial crisis is flagged as outlying. Only June 24 2015 and August 24 2016 are deleted. Both these days correspond to infamous flash crashes that ravage volatility estimates \citep*[e.g.][]{christensen-oomen-reno:22a}.

An overview of the remaining data is presented in Table \ref{table:hurst-exponent}. It reports the starting date of each index and the sample size. We include information up to 31 July 2019 and exclude .KSE og .STI from our investigation, as there are sizable gaps in their data series.

\begin{sidewaystable}
\begin{footnotesize}
\setlength{ \tabcolsep}{0.19cm}
\begin{center}
\caption{Parameter estimation of the log-normal fSV model in stock index data.
\label{table:hurst-exponent}}
\vspace*{-0.25cm}
\begin{tabular}{lllccccccccccccc}
\hline \hline
&&&&&&&&&& \multicolumn{4}{c}{GMM estimate} \\
\cline{11-14}
code & index & location & start date & sample size & $n$ & $\widebar{RV}$ && $\rho_{1}$ && $\xi$ & $\lambda \times 100$ & $\nu$ & $H$ && $\mathcal{J}_{ \text{HS}}$ \\
\hline 
.AEX & AEX & Netherlands & 2000--01 & 4,990 & 102 & 0.028 && 0.743 && 0.024 & 0.026 & 1.800 & 0.032 && 0.726 \\
.AORD & All Ordinaries & Australia & 2000--01 & 4,941 & 72 & 0.011 && 0.575 && 0.008 & 4.873 & 2.345 & 0.035 && 0.128 \\
.BFX & BEL 20 & Belgium & 2000--01 & 4,987 & 102 & 0.022 && 0.644 && 0.018 & 0.012 & 1.938 & 0.021 && 0.603 \\
.BSESN & BSE Sensex & India & 2000--01 & 4,852 & 75 & 0.035 && 0.625 && 0.030 & 0.012 & 2.029 & 0.021 && 0.621 \\
.BVLG & PSI All-Share & Portugal & 2012--10 & 1,731 & 96 & 0.011 && 0.630 && 0.010 & 0.027 & 1.893 & 0.015 && 0.764 \\
.BVSP & Bovespa & Brazil & 2000--01 & 4,819 & 84 & 0.038 && 0.718 && 0.032 & 0.103 & 1.852 & 0.025 && 0.847 \\
.DJI & DJIA & USA & 2000--01 & 4,907 & 78 & 0.027 && 0.677 && 0.021 & 0.010 & 2.008 & 0.021 && 0.773 \\
.FCHI & CAC 40 & France & 2000--01 & 4,989 & 102 & 0.033 && 0.668 && 0.029 & 0.108 & 1.837 & 0.035 && 0.789 \\
.FTMIB & FTSE MIB & Italy & 2009--06 & 2,581 & 90 & 0.029 && 0.658 && 0.023 & 0.018 & 1.877 & 0.016 && 0.303 \\
.FTSE & FTSE 100 & United Kingdom & 2000--01 & 4,933 & 102 & 0.028 && 0.541 && 0.022 & 0.020 & 2.211 & 0.019 && 0.645 \\
.GDAXI & DAX & Germany & 2000--01 & 4,965 & 102 & 0.041 && 0.702 && 0.037 & 0.021 & 2.009 & 0.027 && 0.922 \\
.GSPTSE & TSX Composite & Canada & 2002--05 & 4,314 & 78 & 0.020 && 0.599 && 0.012 & 0.001 & 2.254 & 0.017 && 0.464 \\
.HSI & Hang Seng & Hong Kong & 2000--01 & 4,791 & 78 & 0.024 && 0.677 && 0.017 & 0.001 & 2.102 & 0.008 && 0.322 \\
.IBEX & IBEX 35 & Spain & 2000--01 & 4,957 & 102 & 0.035 && 0.657 && 0.033 & 0.172 & 1.894 & 0.030 && 0.878 \\
.IXIC & Nasdaq 100 & USA & 2000--01 & 4,906 & 78 & 0.030 && 0.696 && 0.020 & 0.002 & 1.445 & 0.029 && 0.361 \\
.KS11 & KOSPI & South Korea & 2000--01 & 4,814 & 72 & 0.030 && 0.763 && 0.024 & 0.003 & 1.643 & 0.027 && 0.661 \\
.MXX & IPC Mexico & Mexico & 2000--01 & 4,907 & 78 & 0.020 && 0.487 && 0.016 & 0.007 & 2.442 & 0.010 && 0.478 \\
.N225 & Nikkei 225 & Japan & 2000--02 & 4,762 & 72 & 0.025 && 0.683 && 0.023 & 0.183 & 1.847 & 0.029 && 0.564 \\
.NSEI & Nifty 50 & India & 2000--01 & 4,849 & 75 & 0.030 && 0.598 && 0.029 & 0.092 & 2.146 & 0.032 && 0.955 \\
.OMXC20 & OMXC20 & Denmark & 2005--10 & 3,431 & 96 & 0.030 && 0.641 && 0.021 & 0.306 & 2.247 & 0.026 && 0.671 \\
.OMXHPI & OMX Helsinki & Finland & 2005--10 & 3,467 & 102 & 0.028 && 0.592 && 0.018 & 0.001 & 1.475 & 0.028 && 0.360 \\
.OMXSPI & OMX Stockholm & Sweden & 2005--10 & 3,468 & 102 & 0.024 && 0.562 && 0.015 & 0.001 & 2.432 & 0.012 && 0.487 \\
.OSEAX & Oslo Exchange & Norway & 2001--09 & 4,462 & 101 & 0.031 && 0.634 && 0.023 & 0.001 & 2.317 & 0.012 && 0.379 \\
.RUT & Russel 2000 & USA & 2000--01 & 4,907 & 78 & 0.018 && 0.663 && 0.013 & 0.051 & 2.180 & 0.019 && 0.674 \\
.SMSI & Madrid General & Spain & 2005--07 & 3,587 & 101 & 0.031 && 0.648 && 0.029 & 0.218 & 2.008 & 0.030 && 0.591 \\
.SPX & S\&P 500 & USA & 2000--01 & 4,911 & 78 & 0.026 && 0.699 && 0.019 & 0.100 & 1.610 & 0.043 && 0.641 \\
.SSEC & Shanghai Composite & China & 2000--01 & 4,724 & 66 & 0.042 && 0.674 && 0.036 & 0.015 & 1.865 & 0.025 && 0.692 \\
.SSMI & Swiss Market Index & Switzerland & 2000--01 & 4,904 & 102 & 0.020 && 0.732 && 0.017 & 0.022 & 1.613 & 0.037 && 0.747 \\
.STOXX50E & EURO STOXX 50 & Europe & 2000--01 & 4,989 & 102 & 0.039 && 0.585 && 0.033 & 0.092 & 2.013 & 0.030 && 0.783 \\
\cline{1-1} \cline{7-7}  \cline{9-9} \cline{11-14} \cline{16-16}
Average &&&&&& 0.028 && 0.640 && 0.022 & 0.214 & 1.980 & 0.024 &&0.607\\
\hline \hline
\end{tabular}
\smallskip
\begin{scriptsize}
\parbox{0.98\textwidth}{\emph{Note.} 
``code'' is based on the Oxford-Man Institute's naming convention. 
``sample size' is the number of observations for the stock index. 
$n$ is the number of intraday returns at the 5-minute sampling frequency. 
$\widebar{RV}$ is the sample average realized variance. 
$\rho_{1}$ is the first-order autocorrelation of realized variance. 
$\xi$ is the average level of the spot variance process, 
$\lambda$ is the speed of mean reversion (multiplied by 100), 
$\nu$ is the volatility-of-volatility, 
while $H$ is the Hurst exponent. 
The cross-sectional average of each descriptive statistic is shown in the bottow row. 
$\mathcal{J}_{ \text{HS}}$ reports the $P$-value of the Sargan-Hansen test of overidentifying restrictions, which is asymptotically $\chi^{2}(4)$-distributed. 
}
\end{scriptsize}
\end{center}
\end{footnotesize}
\end{sidewaystable}

The GMM estimation follows the setup from the simulation section. Since we are dealing with equity data, we can expect both drift and leverage to be present, so the approximate correction from Table \ref{table:error-term} is employed. The right-hand side of Table \ref{table:hurst-exponent} shows the outcome for individual stock indexes, where the bottom row presents the cross-sectional average of each parameter. Looking at the table, the results are remarkably stable across assets. The $\bar{ \xi} = 0.022$ estimate corresponds to about 14.91\% annualized volatility in the aggregate stock market. We do observe a slight underestimation of the mean volatility level compared to the sample average realized variance, which is consistent with the Monte Carlo analysis.

The reported Hurst exponents suggest a very rough volatility process with an average level of $\bar{H} = 0.024$. This is on par with \citet*{bayer-friz-gatheral:16a} and \citet*{fukasawa-takabatake-westphal:22a} but slightly smaller compared to \citet*{bennedsen-lunde-pakkanen:22a} and \citet*{gatheral-jaisson-rosenbaum:18a}. A possible reason for this discrepancy is that the latter employ realized variance as a proxy for spot variance. However, the former is a consistent estimator of the integrated variance, which is much smoother than instantaneous variance (see Figure \ref{figure:sigma}). This ought to bias their $H$ estimates upwards. Our procedure does not suffer from that problem, as we directly compare realized variance to the dynamics of integrated variance in the fSV model, so the averaging ``cancels out.''

\begin{figure}[t!]
\begin{center}
\caption{Properties of the roughness estimate. \label{figure:roughness}}
\begin{tabular}{cc}
\small{Panel A: 90\% confidence interval.} & \small{Panel B: $\hat{H}$ versus $\rho_{1}$.} \\
\includegraphics[height=0.4\textwidth,width=0.48\textwidth]{{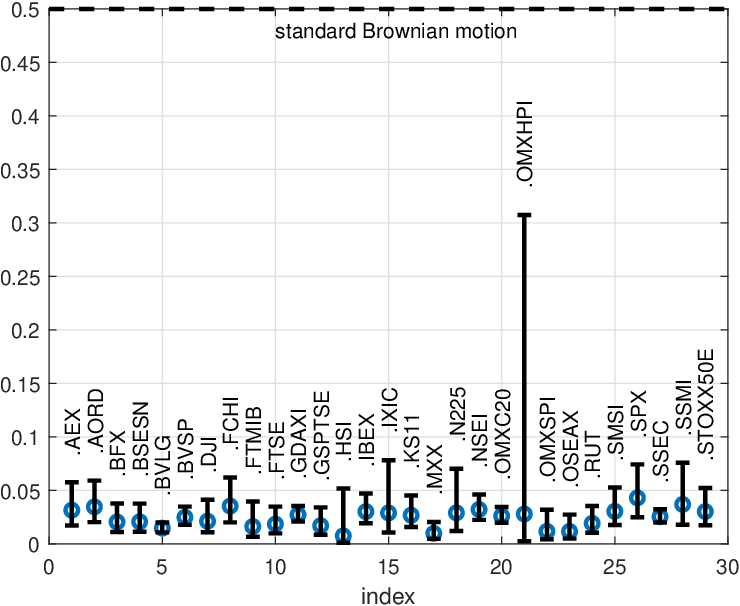}} &
\includegraphics[height=0.4\textwidth,width=0.48\textwidth]{{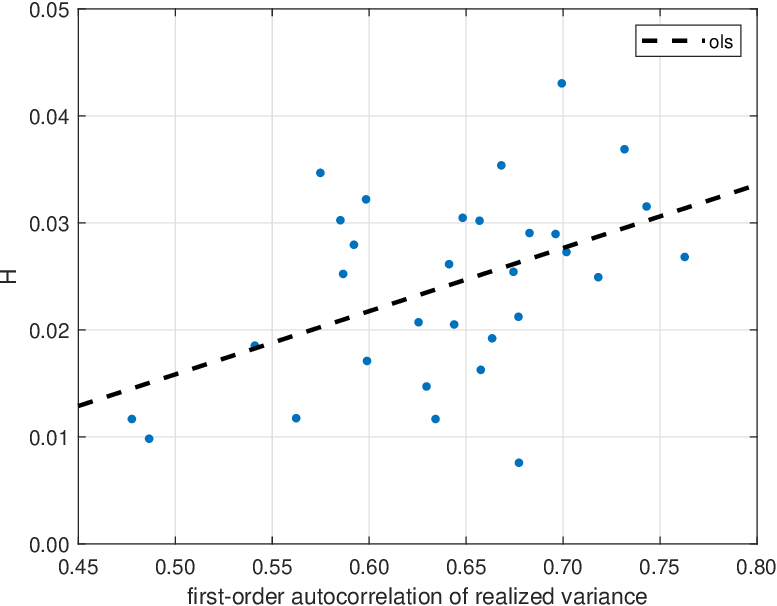}} \\
\end{tabular}
\begin{scriptsize}
\parbox{\textwidth}{\emph{Note.} In Panel A, we report 90\% CIs for $H$. The CIs are based on the log-based distribution theory to enforce non-negativity. This also makes the CIs asymmetric. In Panel B, we show a scatter plot of $\widehat{H}$ versus $\rho_{1}$ together with the fitted regression line $\widehat{H} = \hat{a} + \hat{b} \rho_{1}$.}
\end{scriptsize}
\end{center}
\end{figure}

In Panel A of Figure \ref{figure:roughness}, we report the empirical estimates of $H$ together with 90\% confidence intervals (CIs). The latter are constructed by exponentiating the CIs associated with the log-based distribution theory, where we apply the delta rule to conclude that $\displaystyle \frac{\sqrt{T}( \ln(\widehat{H}) - \ln(H))}{\widehat{ \text{se}}( \widehat{H})/ \widehat{H}} \overset{d}{ \longrightarrow} N(0,1)$ as $T \rightarrow \infty$.\footnote{This approach has the advantage of enforcing non-negativity on the CI for $H$.} Here, $\widehat{ \text{se}}( \widehat{H})$ is the standard error of $\widehat{H}$ extracted from the estimate of the asymptotic variance-covariance matrix in \eqref{equation:efficient-gmm}. Looking at the graph, we discover that $H$ is typically estimated very accurately, meaning realized variance is informative about the true level of roughness in the data. The exception is the Finnish stock market, where the range of plausible values is rather extensive. Irrespective of this, the upper bound of the CIs are far away from the level implied by a standard Brownian motion. In Panel B, we add a scatter plot of $\widehat{H}$ versus the persistence of realized variance, as measured by its first-order autocorrelation coefficient, $\rho_{1}$. This reveals a pronounced positive association between the series. A linear regression $\widehat{H} = a + b \rho_{1} + \epsilon$ shows an $R^{2} = 0.2120$ and a slope estimate $\hat{b} = 0.0591$ that is statistically significant with a $P$-value of 0.0091 (the estimated intercept $\hat{a} = -0.0137$ is highly insignificant).

To gauge the statistical fit of the model, we calculate the J-test for overidentifying restrictions. To reiterate, the test statistic $\mathcal{J}_{ \text{HS}}$ has an asymptotic $\chi^{2}(4)$-distribution under $\mathcal{H}_{0}$. The results appear in the last column of Table \ref{table:hurst-exponent}. Overall, the $P$-values are relatively high, so the fSV process does a good job in describing the data.

\begin{figure}[t!]
\begin{center}
\caption{Properties of .SPX realized variance. \label{figure:spx-realized-variance}}
\begin{tabular}{cc}
\small{Panel A: Time series.} & \small{Panel B: Autocorrelation function.} \\
\includegraphics[height=0.4\textwidth,width=0.48\textwidth]{{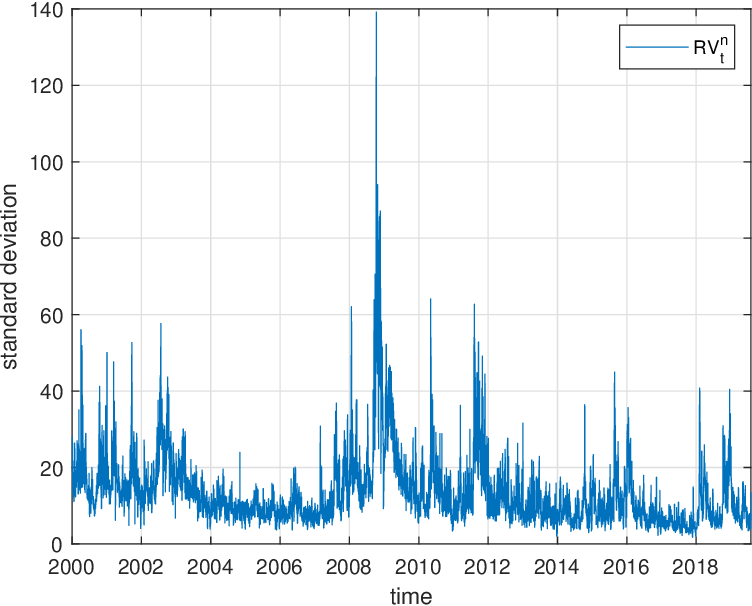}} &
\includegraphics[height=0.4\textwidth,width=0.48\textwidth]{{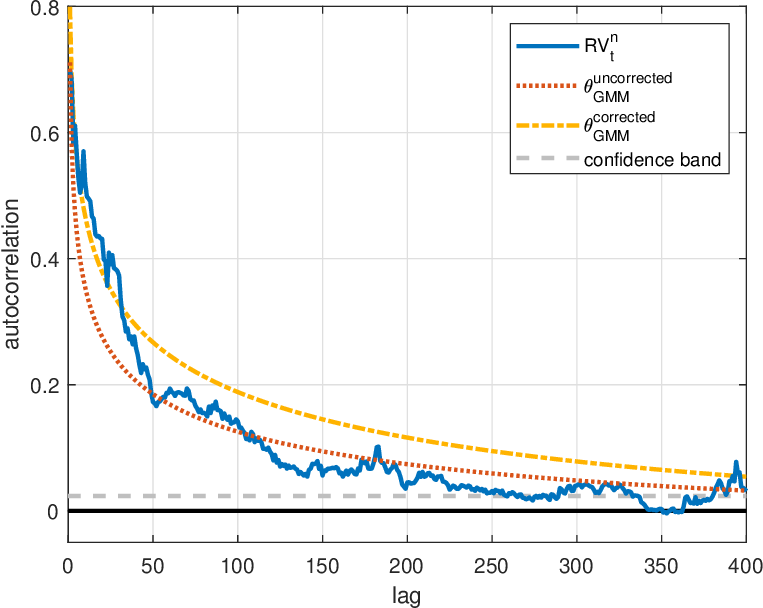}} \\
\end{tabular}
\begin{scriptsize}
\parbox{\textwidth}{\emph{Note.} In Panel A, we plot the realized variance of .SPX converted to standard deviation per annum. In Panel B, we show the sample acf of realized variance with a 95\% white noise confidence band. We compare this to the theoretical acf of the log-normal fSV model implied by the estimated parameter vector $\hat{ \theta}_{\text{GMM}}$, where the latter is reported without and with the bias correction from \eqref{equation:bias-correction}.}
\end{scriptsize}
\end{center}
\end{figure}

As an illustration of our findings we zoom in at .SPX, which represents the S\&P 500 index and is therefore related to developments in the US stock market. In Panel A of Figure \ref{figure:spx-realized-variance}, we show the realized variance of .SPX (the raw estimator has been converted to standard deviation per annum for convenience). It displays the customary volatility clustering present in most financial asset return series. In Panel B, we plot the first 400 lags of the associated acf of $RV_{t}^{n}$ together with a Bartlett one-sided 95\% white noise confidence band. It is suggestive of significant memory in integrated variance. The acf retains the impression of a hyperbolic decay, which is consistent with $H \neq 0.5$. The initial steep decline is indicative of roughness driving the serial correlation, in contrast to the much slower decay symptomatic of long-memory, where the latter entails a Hurst exponent above 0.5. As a comparison, we superimpose the model-implied acf recovered from the GMM parameter estimation, where the latter is shown both with and without the approximate bias correction in \eqref{equation:bias-correction}. The acf of the uncorrected estimator tracks the sample counterpart based on realized variance closely both at the short and long end. Meanwhile, the effect of the bias correction is to lift the acf higher, indicating a larger amount of memory in integrated variance. Note that this was to be expected, since the impact of measurement error in a time series is to attenuate the acf \citep*[e.g.,][]{hansen-lunde:14a}.\footnote{We also estimated the fSV model with a driving standard Brownian motion, i.e. pre-imposing $H = 0.5$. The remaining parameter estimates were $(\hat{ \xi}, \hat{ \lambda}, \hat{ \nu}) = (0.016,0.101,0.552)$, which broadly aligns with previous studies. Intuitively, to fit the sample acf of realized variance the GMM procedure has to select a larger mean-reversion parameter $\lambda$ to compensate for the extra memory induced by forcing $H$ to one-half.}

As a robustness check, we fetched from the NYSE Trade and Quote (TAQ) database a 5-minute transaction price series for the ticker symbol SPY; an exchange-traded fund tracking the S\&P 500 and equivalent to the .SPX index. The purpose is to verify the findings via our in-house data sources and to gauge their sensitivity with respect to price jumps. We pair SPY with the .SPX sample and construct realized variance, bipower variation and truncated realized variance. The latter are jump-robust estimators of integrated variance. In that ordering, the GMM procedure delivers the following $H$ estimates: 0.033, 0.038, 0.035. As expected, bipower variation and truncated realized variance return slightly larger measures of $H$, but the various estimates are practically speaking identical and consistent with the .SPX estimate from Table \ref{table:hurst-exponent}. The other parameter estimates are also close. Hence, the presence and degree of roughness is not dictated by the choice of noisy proxy, nor is it caused by application of the non-robust realized variance.

In sum, our empirical results point toward a very erratic volatility process in line with---or even exceeding---previous research. As these findings are not induced by microstructure noise nor discretization error, we are bound to conclude there is roughness in variance.

\section{Conclusion} \label{section:conclusion}

We propose a GMM framework for estimation of the log-normal SV model governed by a fractional Brownian motion. Our procedure is built from the dynamic properties of integrated variance, but it employs a realized measure of volatility computed from high-frequency as a noisy proxy. We explicitly account for the inherent measurement error in the selected estimator by adjusting an appropriate moment condition. We prove consistency and asymptotic normality our estimator in a classical long-span setting. A Monte Carlo study shows our routine is capable of recovering the parameters of the model across the entire memory spectrum. We implement the approach on high-frequency data from leading equity market indexes and confirm the presence of substantial roughness in the stochastic variance process, as consistent with recent findings in the literature.

\pagebreak

\appendix

\section{Proofs} \label{App:Proofs}

\subsection{Auxiliary Result}

To prove Theorem \ref{theo:MomentsIVGeneral}, we need the following auxiliary result that enables to express certain two-dimensional integrals in a one-dimensional form.

\begin{lemma} \label{lemma:integration}
Assume $f : [0, \infty) \rightarrow  \mathbb R$ is a continuous function and let $k \in \mathbb N$. Then,
\begin{equation*}
\int_{k-1}^{k} \int_{0}^{1} f(|s-t|) \mathrm{d}s \mathrm{d}t = \int_{0}^{1} (1-y) \big(f(|k-1-y|)+f(k-1+y) \big) \mathrm{d}y.
\end{equation*}
\end{lemma}
\begin{proof}
Write
\begin{equation*}
\int_{k-1}^k \int_{0}^{1} f(|s-t|) \mathrm{d}s \mathrm{d}t = \iint \limits_{[k-1,k] \times[0,1]} f(|s-t|) \mathrm{d}s \mathrm{d}t
\end{equation*}
and introduce the linear (bijective) change of variables:
\begin{equation*}
\begin{bmatrix}
s \\[0.10cm] t
\end{bmatrix}
=
\frac{1}{2} \begin{bmatrix}
(u+v) \\[0.10cm] (-u+v)
\end{bmatrix}
= \underbrace{ \frac{1}{2}
\begin{bmatrix}
1 & 1 \\[0.10cm] -1 & 1
\end{bmatrix}}_{ \equiv A}
\begin{bmatrix}
u \\[0.10cm] v
\end{bmatrix} \equiv
\begin{bmatrix} \varphi_{1}(u,v) \\[0.10cm] \varphi_{2}(u,v)
\end{bmatrix} \equiv \varphi(u,v).
\end{equation*}
Applying this to the inequalities $k-1 \leq s \leq k$ and $0 \leq t \leq 1$, we find they are equivalent to
\begin{equation*}
v \leq 2k - u, \quad v \geq 2(k-1)-u, \quad v \geq u, \quad \text{and} \quad v \leq u+2.
\end{equation*}
Therefore, the set
\begin{equation*}
B = \{ (u,v) \in \mathbb{R}^{2} : v \leq 2k - u, \quad v \geq 2(k-1)-u, \quad v \geq u, \quad \text{and} \quad  v \leq u+2\}
\end{equation*}
is mapped by $\varphi$ to $[k-1,k] \times[0,1]$. Note also that $B = B_{1} \cup B_{2}$, where
\begin{align*}
B_{1} &\equiv \{ (u,v) \in \mathbb{R}^{2} : k-2 \leq u < k-1 \quad \text{and} \quad 2(k-1)-u \leq v \leq u+2 \} \\[0.10cm]
B_{2} &\equiv \{ (u,v) \in \mathbb{R}^{2} : k-1 \leq u \leq k \quad \text{and} \quad u \leq v \leq 2k-u \}
\end{align*}
are disjoint. Now, the Jacobian $(D \varphi)(u,v)$ of $\varphi$ equals $A$ for any $(u,v) \in \mathbb{R}^2$, whereby
\begin{equation*}
| \det(D \varphi)(u,v)| = | \det(A)| = \frac{1}{2},
\end{equation*}
and since $s-t = \varphi_{1}(u,v) - \varphi_{2}(u,v) = \frac{1}{2} (u+v) - \frac{1}{2} (-u+v) = u$, we get by multivariate integration by substitution:
\begin{equation*}
\begin{split}
\iint \limits_{[k-1,k] \times[0,1]} f(|s-t|) \mathrm{d}s \mathrm{d}t &= \iint \limits_{ \varphi(B)} f(|s-t|) \mathrm{d}s \mathrm{d}t \\[0.10cm]
&= \iint \limits_{B} f(| \varphi_{1}(u,v) - \varphi_{2}(u,v)|) | \det(D \varphi)(u,v)| \mathrm{d}u \mathrm{d}v \\[0.10cm]
&= \frac{1}{2} \bigg( \iint \limits_{B_{1}} f(|u|) \mathrm{d}u \mathrm{d}v + \iint \limits_{B_{2}} f(|u|) \mathrm{d}u \mathrm{d}v \bigg).
\end{split}
\end{equation*}
Firstly,
\begin{align*}
\begin{split}
\iint \limits_{B_{1}} f(|u|) \mathrm{d}u \mathrm{d}v = \int_{k-2}^{k-1} \bigg( \int_{2(k-1)-u}^{u+2} f(|u|) \mathrm{d}v \bigg) \mathrm{d}u &= 2 \int_{k-2}^{k-1} \big(u-(k-2) \big) f(|u|) \mathrm{d}u \\[0.10cm]
&= 2 \int_{0}^{1} (1-y) f(|k-1-y|) \mathrm{d} y,
\end{split}
\end{align*}
via the substitution $y = k-1-u$. Secondly,
\begin{align*}
\begin{split}
\iint \limits_{B_{2}} f(|u|) \mathrm{d}u \mathrm{d}v = \int_{k-1}^{k} \bigg( \int_{u}^{2k-u} f(|u|) \mathrm{d}v \bigg) \mathrm{d}u &= 2 \int_{k-1}^{k}(k-u) f(u) \mathrm{d}u \\[0.10cm]
&= 2 \int_{0}^{1} (1-y) f(k-1+y) \mathrm{d}y,
\end{split}
\end{align*}
by substituting $y = u-(k-1)$ and noting $u \geq k-1 \geq 0$. Thus, the asserted formula follows. \qed
\end{proof}

\subsection{Proof of Theorem \ref{theo:MomentsIVGeneral}}

To prove the first part of the theorem, we note that since the variance process $( \sigma_{t}^{2})_{t \geq 0}$ is stationary, Fubini's theorem yields that
\begin{equation*}
\mathbb{E} [IV_{t}] = \int_{t-1}^{t} \mathbb{E} \big[ \sigma_{s}^{2} \big] \text{d}s = \mathbb{E} \big[ \sigma_{0}^{2} \big] = \xi.
\end{equation*}
We proceed with the second-order moments of $IV_{t}$ by noting that
\begin{align*}
\mathbb{E}[ \sigma_{t}^{2} \sigma_{s}^{2}] &= \xi^{2} \mathbb{E} \big[ \exp \big( Y_{t} + Y_{s}- \kappa(0) \big)] \\[0.10cm]
&= \xi^{2} \exp \big( \kappa(|t-s|) \big),
\end{align*}
where the last equation follows from $Y_{t} + Y_{s} \sim N(0, 2 \kappa(|t-s|) + 2 \kappa(0))$. We deduce that
\begin{align*}
\mathbb{E} [IV_{1} IV_{1+ \ell}] &= \int_{ \ell}^{ \ell+1} \int_{0}^{1} \mathbb{E} \big[ \sigma_{s}^{2} \sigma_{t}^{2} \big] \text{d}s \text{d}t \\[0.10cm]
&= \xi^{2} \int_{ \ell}^{ \ell+1} \int_{0}^{1} \exp( \kappa(|t-s|)) \text{d}s \text{d}t \\[0.10cm]
&= \xi^{2} \int_{0}^{1} (1-y) \big[ \exp( \kappa(| \ell-y|)) + \exp( \kappa( \ell+y)) \big] \text{d}y,
\end{align*}
as a consequence of Lemma \ref{lemma:integration}.

Next, we deal with the third moment of integrated variance:
\begin{equation*}
\mathbb{E}[ \sigma_{t}^{2} \sigma_{s}^{2} \sigma_{u}^{2}] = \xi^{3} \exp \Big(Y_{t}+Y_{s}+Y_{u}- \frac{3}{2} \kappa(0) \Big),
\end{equation*}
where $Y_{t}+Y_{s}+Y_{u}$ is Gaussian with mean zero and
\begin{equation*}
\text{var}(Y_{t}+Y_{s}+Y_{u}) = 3 \kappa(0) + 2 \big( \kappa(|t-s|)+ \kappa(|t-u| )+ \kappa(|s-u|) \big).
\end{equation*}
Then, it follows that
\begin{equation*}
\mathbb{E}[\sigma_t^2 \sigma_s^2 \sigma_u^2]=\xi^3 \exp \left( \kappa(|t-s| )+\kappa(|t-u| )+\kappa(|s-u| )   \right),
\end{equation*}
from which we deduce
\begin{align*}
\mathbb{E}[IV_{t}^{3}] &= \xi^{3} \int_{0}^{1} \int_{0}^{1} \int_{0}^{1} \mathbb{E}[  \sigma_{t}^{2} \sigma_{s}^{2} \sigma_{u}^{2}] \mathrm{d}u \mathrm{d}s \mathrm{d}t \\[0.10cm]
&= \xi^{3} \int_{0}^{1} \int_{0}^{1} \int_{0}^{1} \exp \big( \kappa(|t-s|)+ \kappa(|t-u|)+ \kappa(|s-u|) \big) \mathrm{d}u \mathrm{d}s \mathrm{d}t.
\end{align*}
As above, we exploit the symmetric structure of the problem. We start by reducing the three-dimensional integral applying the following two-dimensional change of variables:
\begin{equation*}
x = u-t, \quad y=s-t,
\end{equation*}
such that
\begin{align*}
\mathbb{E}[IV_{t}^{3}] &= \xi^{3} \int_{0}^{1} \int_{-t}^{1-t} \int_{-t}^{1-t} f(x,y) \mathrm{d}x \mathrm{d}y \mathrm{d}t = \xi^{3} \int_0^{1} \int_{-1}^{1} \int_{-1}^{1} 1_{\{-t \leq x,y \leq 1-t \}} f(x,y) \mathrm{d}x \mathrm{d}y \mathrm{d}t \\[0.10cm]
&= \xi^{3} \int_{-1}^{1} \int_{-1}^{1} c(x,y) f(x,y) \mathrm{d}x \mathrm{d}y,
\end{align*}
where $f(x,y) = \exp \big( \kappa(|x-y|) + \kappa(|x|) + \kappa(|y|) \big)$ and $c(x,y) = \int_{0}^{1} 1_{\{ -t \leq x, y \leq 1-t \}} \mathrm{d}t$. Note that $c(x,y)=c(-x,-y)$, which is seen from substituting $s=1-t$. Now, when $(x, y) \in [0,1] \times [0,1]$:
\begin{equation*}
c(x,y) = 1- \max(x,y).
\end{equation*}
Moreover, if $(x, y) \in [0,1] \times [-1,0]$, we conclude that $c(x,y) = \max(1-x+y,0)$. Inserting these terms and noting that $f(x,y) = f(-x,-y) = f(y,x)$:
\begin{align*}
\mathbb{E}[IV_{t}^{3}] &= 2 \xi^{3} \int_{0}^{1} \int_{0}^{1} (1- \max(x,y)) f(x,y) \mathrm{d}y \mathrm{d}x + 2 \xi^{3} \int_{0}^{1} \int_{-1}^{0} \max(1-x+y,0) f(x,y) \mathrm{d}y \mathrm{d}x \\[0.10cm]
&=4 \xi^{3} \int_{0}^{1} \int_{0}^{x} (1-x) f(x,y) \mathrm{d}y \mathrm{d}x + 2 \xi^{3} \int_{0}^{1} \int_{x-1}^{0} (1-x+y) f(x,y) \mathrm{d}y \mathrm{d}x \\[0.10cm]
&=6 \xi^{3} \int_{0}^{1} \int_{0}^{x} (1-x) f(x,y) \mathrm{d}y \mathrm{d}x.
\end{align*}
The last equality holds, since both double integrals in the above expression agree in value. To show this, we substitute $z=x-y$ and reexpress $f(x,z) = f(x, x-z)$, which yields
\begin{align*}
\int_{x-1}^{0} (1-x+y) f(x,y) \mathrm{d}y = \int_{x}^{1} (1-z) f(x,z) \mathrm{d}z.
\end{align*}
Finally, by exploiting $f(x,y)=f(y, x)$ again, we further conclude
\begin{align*}
\int_{0}^{1} \int_{0}^{x} (1-x) f(x,y) \mathrm{d}y \mathrm{d}x = \int_{0}^{1} \int_{x}^{1} (1-z) f(x,z) \mathrm{d}z \mathrm{d}x.
\end{align*}
To calculate the fourth moment of $IV_{t}$, we proceed as above:
\begin{equation*}
\mathbb{E}[ \sigma_{t}^{2} \sigma_{s}^{2} \sigma_{u}^{2} \sigma_{v}^{2}] = \xi^{4} \exp \big( \kappa(|t-s|)+ \kappa(|t-u|)+ \kappa(|s-u|)+ \kappa(|t-v|)+ \kappa(|s-v|)+ \kappa(|u-v|) \big).
\end{equation*}
Subsequently, to compute the four-dimensional integral we apply the change of variables:
\begin{equation*}
x=u-t, \quad y=s-t, \quad z=v-t,
\end{equation*}
which, by recalling the definition of $g$, implies
\begin{align*}
\mathbb{E}[IV_t^4] &=\xi^4 \int_0^1 \int_{-t}^{1-t} \int_{-t}^{1-t} \int_{-t}^{1-t} g(x,y,z)   d x dy dz dt \\
&=\xi^4 \int_{-1}^{1} \int_{-1}^{1} \int_{-1}^{1} d(x,y, z) g(x,y, z) d x dy dz,
\end{align*}
where $g(x,y,z) = \exp \big( \kappa(|x-y|)+ \kappa(|x-z|)+ \kappa(|y-z|)+ \kappa(|x|)+ \kappa(|y|)+ \kappa(|z|) \big)$ and $d(x,y, z)= \int_{0}^{1} 1_{\{ -t \leq x,y,z \leq 1-t \}} \mathrm{d}t$.

To compute this term, we split the cube $[-1, 1] \times [-1, 1] \times [-1, 1]$ into eight quadrants determined by the signs of the variables. And since $g(x,y,z) = g(-x,-y,-z)$ and $d(x,y,z) = d(-x,-y,-z)$, we concentrate our efforts to four quadrants, which are handled case-by-case analogously to the derivation of the third moment.

1. $(x,y,z) \in [0,1] \times [0,1] \times [0,1]$:

Here $d(x,y,z) = 1- \max(x,y,z)$. Since $g$ and $d$ are invariant under permutation of its three variables, it suffices to compute the integral for $z \leq y \leq x$. As a result,
\begin{align*}
\int_{0}^{1} \int_{0}^{1} \int_{0}^{1} d(x,y,z) g(x,y,z) \mathrm{d}z \mathrm{d}y \mathrm{d}x = 6 \int_{0}^{1} \int_{0}^{x} \int_{0}^{y} (1-x) g(x,y,z) \mathrm{d}z \mathrm{d}y \mathrm{d}x.
\end{align*}

2. $(x,y,z) \in [0,1] \times [0,1] \times [-1,0]$:

Assume $y \leq x$. Then, $d(x,y,z) = \max(1-x+z,0)$ and
\begin{align*}
\int_{0}^{1} \int_{0}^{1} \int_{-1}^{0} d(x,y,z) g(x,y,z) \mathrm{d}z \mathrm{d}y \mathrm{d}x &= 2 \int_{0}^{1} \int_{0}^{x} \int_{-1}^{0} \max(1-x+z,0) g(x,y,z) \mathrm{d}z \mathrm{d}y \mathrm{d}x \\[0.10cm]
&= 2 \int_{0}^{1} \int_{0}^{x} \int_{x-1}^{0} (1-x+z) g(x,y,z) \mathrm{d}z \mathrm{d}y \mathrm{d}x.
\end{align*}

3. $(x,y,z) \in [0,1] \times [-1,0] \times [0,1]$:

Assume $z \leq x$. Then, $d(x,y,z) = \max(1-x+y,0)$ and
\begin{equation*}
\int_{0}^{1} \int_{0}^{1} \int_{-1}^{0} d(x,y,z) g(x,y,z) \mathrm{d}y \mathrm{d}z \mathrm{d}x = 2 \int_{0}^{1} \int_{0}^{x} \int_{-1}^{0} \max(1-x+y,0) g(x,y,z) \mathrm{d}y \mathrm{d}z \mathrm{d}x,
\end{equation*}
which reduces to the integral over the second region, because $g(x,y,z) = g(x,z,y)$.

4. $(x,y,z) \in [-1,0] \times [0,1] \times [0,1]$:

Assume $z \leq y$. Then, $d(x,y,z) = \max(1-y+x, 0)$ and
\begin{equation*}
\int_{0}^{1} \int_{0}^{1} \int_{-1}^{0} d(x,y,z) g(x,y,z) \mathrm{d}x \mathrm{d}z \mathrm{d}y = 2 \int_{0}^{1} \int_{0}^{y} \int_{-1}^{0} \max(1-y+x,0) g(x,y,z) \mathrm{d}x \mathrm{d}z \mathrm{d}y,
\end{equation*}
which is identical to the integral over the third region due to $g(x,y,z)=g(y,x,z)$.

Next, we show that the second integral (and by implication the third and fourth) is equal to the first one. To this end, in the inner two integrals over the second region, we substitute $u=x-y$, $v=x-z$ and employ the identity $g(x,y,z) = g(x,x-z,x-y)$. This leads to
\begin{equation*}
\int_{0}^{x} \int_{x-1}^{0} (1-x+z) g(x,y,z) \mathrm{d}z \mathrm{d}y = \int_{0}^{x} \int_{x}^{1} (1-v) g(x,u,v) \mathrm{d}v \mathrm{d}u,
\end{equation*}
and hence
\begin{equation*}
\int_{0}^{1} \int_{0}^{x} \int_{x}^{1} (1-v) g(x,u,v) \mathrm{d}v \mathrm{d}u \mathrm{d}x = \int_{0}^{1} \int_{0}^{v} \int_{0}^{x} (1-v) g(v,x,u) \mathrm{d}u \mathrm{d}x \mathrm{d}v,
\end{equation*}
which concludes the proof of the claimed equality of integrals. Summing up the terms, we get a factor $24 = 2 \times 12$ in front of the integral.

As for the second part of Theorem  \ref{theo:MomentsIVGeneral}, note that from condition (a) there exists $\ell_{0} > 0$ such that $| \kappa(u)| \leq 1$ for any $u \geq \ell_{0}-1$. Denoting $\gamma_{ \ell+1,1} = \mathbb{E}[IV_{t} IV_{t+ \ell}]- \xi^{2}$, we thus find that
\begin{equation*}
\gamma_{ \ell+1,1} = \xi^{2} \int_{0}^{1} (1-y) \big( \exp \big( \kappa(| \ell-y|) \big) - 1 + \exp \big( \kappa( \ell+y) \big)-1 \big) \text{d}y.
\end{equation*}
Introducing $r(x) \equiv \exp(x) - 1 - x$, $x \in \mathbb{R}$ allows to further write
\begin{equation*}
\frac{ \gamma_{ \ell+1,1}}{ \xi^{2} \kappa( \ell)} = \underbrace{ \int_{0}^{1} (1-y) \bigg( \frac{ \kappa( \ell-y)}{ \kappa( \ell)} + \frac{ \kappa( \ell+y)}{ \kappa( \ell)} \bigg) \text{d}y}_{\equiv I_{1}} + \underbrace{ \int_{0}^{1} (1-y) \bigg( \frac{r( \kappa( \ell-y))}{ \kappa( \ell)} + \frac{r( \kappa( \ell+y))}{ \kappa( \ell)} \bigg) \text{d}y}_{\equiv I_{2}},
\end{equation*}
for any $\ell \geq \ell_{0}$.

As $|r(x)| \leq 3x^{2}$, $x \in [0,1]$, it follows that
\begin{equation*}
|I_{2}| \leq 3 \sup_{y \in [-1,1]} \bigg| \frac{ \kappa( \ell+y)}{ \kappa( \ell)} \bigg| \int_{0}^{1} \underbrace{(1-y)(| \kappa(| \ell-y|)|+| \kappa( \ell+y)|)}_{ \equiv v_{ \ell}(y)} \text{d}y,
\end{equation*}
where for any $y \in [0,1] : 0 \leq v_{ \ell}(y) \leq 1$, while $v_{ \ell}(y) \rightarrow 0$, as $\ell \rightarrow \infty$ by \eqref{item:lem-limit}. Applying the dominated convergence theorem and condition (c) implies that:
\begin{equation*}
\limsup_{ \ell \rightarrow \infty}|I_{2}| \leq 3 \limsup_{ \ell \rightarrow \infty} \sup_{\bar{y} \in [-1,1]} \bigg| \frac{ \kappa( \ell+\bar{y})}{ \kappa( \ell)} \bigg| \lim_{ \ell \rightarrow \infty} \int_{0}^{1} v_{ \ell}(y) \text{d} y = 0.
\end{equation*}
Finally, for $y \in [0,1]$ the integrand $u_{ \ell}(y) \equiv (1-y) \big( \frac{ \kappa( \ell-y)}{ \kappa( \ell)} + \frac{ \kappa( \ell+y)}{ \kappa( \ell)} \big)$ in $I_{1}$ is bounded uniformly in $\ell$ by some constant from condition (c), while
\begin{equation*}
\lim_{ \ell \rightarrow \infty} u_{ \ell}(y) = (1-y)( \phi(-y) + \phi(y)), \quad y \in [0,1],
\end{equation*}
by condition (b). Thus, by dominated convergence
\begin{equation*}
\lim_{ \ell \rightarrow \infty} I_{1} = \int_{0}^{1} (1-y)( \phi(-y) + \phi(y)) \text{d}y = \int_{-1}^1 (1-|y|) \phi(y) \text{d}y,
\end{equation*}
which concludes the proof. \qed

\subsection{Proof of Theorem \ref{consistency:gmmcorrection}}

We apply Theorem 2.1 of \citet*{hansen:82a}, the sufficient conditions of which are implied by our Assumptions \ref{assum:gauss} -- \ref{assumption:identification}. It remains to verify $G_{c}( \theta)$ is continuous in $\theta$, which also renders the random function $\widehat{m}_T( \theta)$ continuous in $\theta$. Next, note that the moduli of continuity of $\widehat{m}_{T}( \theta)$ and $G_{c}( \theta)$ coincide, so $G_{c}( \theta)$ being continuous readily implies the so-called first moment continuity of $\widehat{m}_T( \theta)$, see Definition 2.1 in \citet*{hansen:82a}.

To establish continuity of $G_{c}( \theta)$ in $\theta=( \xi, \phi)$, note that $c( \theta)$ is continuous by Assumption \ref{assum:error}, whereby it suffices to prove the continuity of $G(\theta)$. The first component of $G(\theta)$ is $g^{(1)}_0(\theta) = \xi$, which is evidently continuous, while the remaining components are given in integral form in Theorem \ref{theo:MomentsIVGeneral}. Their continuity is then a consequence of the dominated convergence theorem, given condition (ii) of Assumption \ref{assum:gauss}. \qed

\subsection{Proof of Theorem \ref{theorem:GMM}}

We introduce the notation:
\begin{align*}
\widetilde{Q}_{n,T}( \theta) &= \widetilde{m}_{n,T}( \theta)^{ \intercal} \mathbb{W}_{T} \widetilde{m}_{n,T}( \theta), \\[0.10cm]
Q( \theta) &= m( \theta)^{ \intercal} \mathbb{W} m( \theta).
\end{align*}
where $m( \theta) = G( \theta_{0}) - G( \theta)$ and $\mathbb{W} = A^{ \intercal}A$. The claim then follows under the conditions of Theorem 2.1 of \citet*{newey-mcfadden:94a}:
\begin{itemize}
\item[(i)] $Q( \theta)$ is uniquely minimized at $\theta_{0}$,
\item[(ii)] $\Theta$ is compact,
\item[(iii)] $\theta \rightarrow Q( \theta)$ is continuous, and
\item[(iv)] $\sup_{ \theta \in \Theta} | \widetilde{Q}_{n,T}( \theta) - Q( \theta)| \cp 0$.
\end{itemize}
We note condition (i) is implied by Assumption \ref{assumption:identification}, since for $\theta \neq \theta_{0}$:
\begin{equation*}
Q( \theta) = (A m( \theta))' A m( \theta) > 0 = Q( \theta_0).
\end{equation*}
Condition (ii) is immediate. We already showed condition (iii) in the proof of Theorem \ref{consistency:gmmcorrection}. Now, we pass to the last condition (iv). In view of the Cauchy-Schwarz inequality,
\begin{align*}
\big| \widetilde{Q}_{n,T}( \theta)-Q( \theta) \big| &\leq \big|( \widetilde{m}_{n,T}( \theta)-m( \theta))^{ \intercal} \mathbb{W}_{T} ( \widetilde{m}_{n,T}( \theta)-m( \theta)) \big| + \big|m( \theta)^{ \intercal} (\mathbb{W}_{T}+ \mathbb{W}_{T}^{ \intercal})( \widetilde{m}_{n,T}( \theta)-m( \theta)) \big| \\[0.10cm]
&+ \big|m( \theta)^{ \intercal} ( \mathbb{W}_{T}- \mathbb{W}) m( \theta) \big| \\[0.10cm]
&\leq \| \widetilde{m}_{n,T}( \theta)-m( \theta) \|^{2} \, \| \mathbb{W}_{T} \| + 2 \, \|m( \theta) \| \, \| \mathbb{W}_{T} \| \, \| \widetilde{m}_{n,T}( \theta)-m( \theta) \| \\[0.10cm]
&+ \|m( \theta) \|^{2} \, \| \mathbb{W}_{T}- \mathbb{W} \|.
\end{align*}
Then, in view of Assumption \ref{assum:weight-mat-limit}, it suffices to prove that
\begin{equation*}
\sup_{ \theta \in \Theta} \| \widetilde{m}_{n,T}-m( \theta) \| \cp 0, \quad \text{as } T \rightarrow \infty \text{ and } n \rightarrow \infty.
\end{equation*}
Let ${m}_{T}( \theta) = T^{-1} \sum_{t=1}^{T} \big[ \mathbb{IV}_{t} - G( \theta) \big]$. Since the convergence $\sup_{ \theta \in \Theta} \|m_{T}( \theta)-m( \theta) \| \cp 0$ was already covered by the proof of Theorem \ref{consistency:gmmcorrection} (setting $\varepsilon_{t} = c( \theta) = 0$), it remains to show
\begin{equation*}
\sup_{\theta \in \Theta} \| \widetilde{m}_{n,T}( \theta)-{m}_{T}( \theta) \| \cp 0, \quad \text{as } T \rightarrow \infty \text{ and } n \rightarrow \infty.
\end{equation*}
To this end, we observe that
\begin{align*}
\| \widetilde{m}_{n,T}( \theta)-{m}_{T}( \theta) \| &\leq \frac{1}{T} \sum_{t=1}^{T} \| \mathbb{V}_{t}^{n}- \mathbb{IV}_{t} \| \\[0.10cm]
&\leq \frac{1}{T} \sum_{t=1}^{T} \bigg[ |V_{t}^{n}-IV_{t}| + \sum_{j=0}^{k} | V_{t}^{n} V_{t-j}^{n} - IV_{t} IV_{t-j}| \bigg] \\[0.10cm]
&\leq \frac{1}{T} \sum_{t=1}^{T} |V_{t}^{n}-IV_{t}| (1+|V_{t}^{n}|+IV_{t})\\[0.10cm]
&+ \frac{1}{T} \sum_{t=1}^{T} \sum_{j=1}^{k} |V_{t}^{n}-IV_{t}| |V_{t-j}^{n}| + IV_{t}  |V_{t-j}^{n}-IV_{t-j}|.
\end{align*}
From Assumption \ref{assum:doubleasymp} and the Cauchy-Schwarz inequality, we deduce that:
\begin{equation*}
\mathbb{E} \bigg[ \sup_{ \theta \in \Theta} \| \widetilde{m}_{n,T}( \theta)-{m}_{T}( \theta) \| \bigg] \rightarrow 0, \quad \text{as } T \rightarrow \infty \text{ and } n \rightarrow \infty,
\end{equation*}
which was to be shown. \qed

\subsection{Verifying Assumption \ref{assum:doubleasymp} for realized variance} \label{remark:uniform-bound}

Suppose that $\sup_{s \in \mathbb{R}} \mathbb{E}[ \mu_{s}^{4}] + \sup_{s \in \mathbb{R}} \mathbb{E}[ \sigma_{s}^{4}] < \infty$. Then, there exists a constant $C$ such that
\begin{equation*}
\sup_{t \in \mathbb{Z}} \mathbb{E} \big[(RV_{t}^{n}-IV_{t})^{2} \big] \leq \frac{C}{n}.
\end{equation*}
To see this, we apply It\^{o}'s Lemma to get
\begin{equation*}
\big(X_{t-1+ \frac{i}{n}}-X_{t-1+ \frac{i-1}{n}} \big)^{2} = 2 \int_{t-1+ \frac{i-1}{n}}^{t-1+ \frac{i}{n}} (X_{s}-X_{t-1+ \frac{i-1}{n}}) \text{d} X_{s} + \int_{t-1+ \frac{i-1}{n}}^{t-1+ \frac{i}{n}} \sigma_{s}^{2} \text{d}s.
\end{equation*}
Consequently,
\begin{align*}
RV_{t}^{n} - IV_{t} &= 2 \sum_{i=1} \int_{t-1+ \frac{i-1}{n}}^{t-1+ \frac{i}{n}} (X_{s}-X_{t-1+ \frac{i-1}{n}}) \text{d}X_{s} \\[0.10cm]
&= 2 \sum_{i=1}^{n} \int_{t-1+ \frac{i-1}{n}}^{t-1+ \frac{i}{n}}(X_{s}-X_{t-1+ \frac{i-1}{n}}) \mu_{s} \text{d}s + 2 \sum_{i=1}^{n} \int_{t-1+ \frac{i-1}{n}}^{t-1+ \frac{i}{n}}(X_{s}-X_{t-1+ \frac{i-1}{n}}) \sigma_{s} \text{d}W_{s}.
\end{align*}
In turn, this combined with Cauchy–Schwarz and Jensen's inequality leads to
\begin{align*}
\mathbb{E} \big[(RV_{t} - IV_{t})^{2}] &\leq 4 \sum_{i=1}^{n} \int_{ \frac{i-1}{n}}^{ \frac{i}{n}} \mathbb{E}[(X_{s}-X_{ \frac{i-1}{n}})^{2} \mu_{s}^{2}] \text{d}s + 4 \sum_{i=1}^{n} \int_{ \frac{i-1}{n}}^{ \frac{i}{n}} \mathbb{E}[(X_{s}-X_{ \frac{i-1}{n}})^{2} \sigma_{s}^{2}] \text{d}s \\[0.10cm]
& \leq \frac{C}{n},
\end{align*}
where in the last inequality we exploited $\sup_{s \geq 0} \mathbb{E}[ \mu_{s}^{4}] + \sup_{s \geq 0 } \mathbb{E}[ \sigma_{s}^{4}] < \infty$ along with Burkholder's inequality: $\displaystyle \sup_{s \in \big[ \frac{i-1}{n}, \frac{i}{n} \big]} \mathbb{E} \big[(X_{s}-X_{ \frac{i-1}{n}})^{4} \big] \leq \frac{C}{n^{2}}$. \qed

\subsection{Proof of Proposition \ref{prop:clt}} \label{app:prop-clt}

The proof relies on a martingale approximation central limit theorem for stationary and ergodic processes \citep*[Section 4.2]{merlevede-peligrad-utev:19a}, and it requires some preparation. First, we state and prove a couple of generic, elementary lemmas.

\begin{lemma} \label{lem:exp-monotonicity}
Suppose $X$ is a random variable such that $\mathbb{E}[X^{2}] < \infty$ and let $\mathcal{F}$ and $\mathcal{G}$ be $\sigma$-algebras such that $\mathcal{F} \subset \mathcal{G}$. Then,
\begin{equation*}
\| \mathbb{E}[X \mid \mathcal{F}] \|_{L^{2}( \mathbb{P})} \leq \| \mathbb{E}[X \mid \mathcal{G}] \|_{L^{2}( \mathbb{P})}.
\end{equation*}
\end{lemma}

\begin{proof}
Since $\mathcal{F} \subset \mathcal{G}$, we get by the tower property of conditional expectations,
\begin{equation*}
\| \mathbb{E}[X \mid \mathcal{F}] \|_{L^{2}(\mathbb{P})}^{2} = \mathbb{E} \big[ \mathbb{E}[X \mid \mathcal{F}]^{2} \big] = \mathbb{E} \big[ \mathbb{E}[ \mathbb{E}[X \mid \mathcal{G}] \mid \mathcal{F}]^{2} \big].
\end{equation*}
Applying Jensen's inequality for conditional expectations,
\begin{equation*}
\mathbb{E} \big[ \mathbb{E}[X \mid \mathcal{G}] \mid \mathcal{F} \big]^{2} \leq \mathbb{E} \big[ \mathbb{E}[X \mid \mathcal{G}]^{2} \mid \mathcal{F} \big].
\end{equation*}
Hence,
\begin{equation*}
\mathbb{E} \big[ \mathbb{E}[ \mathbb{E}[X \mid \mathcal{G}]| \mathcal{F}]^{2} \big] \leq \mathbb{E} \big[ \mathbb{E}[ \mathbb{E}[X \mid \mathcal{G}]^{2}| \mathcal{F}] \big] = \mathbb{E} \big[ \mathbb{E}[X \mid \mathcal{G}]^{2} \big] = \| \mathbb{E}[X \mid \mathcal{G}] \|_{L^{2}( \mathbb{P})}^{2}.
\end{equation*} \qed
\end{proof}

\begin{lemma} \label{lem:lognormal}
Suppose that $X \sim N( \mu, \lambda^{2})$ for some $\mu \in \mathbb{R}$ and $\lambda >0$. Then,
\begin{equation*}
\mathbb{E} \big[(e^{X}-1)^{2} \big] \leq \big(e^{ \mu + \lambda^{2}}+1 \big)^{2} \big(8| \mu|+6 \lambda^{2} \big).
\end{equation*}
\end{lemma}

\begin{proof}
Note that
\begin{equation*}
\mathbb{E} \big[(e^{X}-1)^{2} \big] = e^{2( \mu + \lambda^{2})} - 2e^{ \mu + \frac{1}{2} \lambda^{2}} + 1 \leq e^{2( \mu + \lambda^{2})} + 2e^{ \mu + \lambda^{2}} + 1 = (e^{ \mu + \lambda^{2}} +1)^{2},
\end{equation*}
while
\begin{equation*}
e^{2( \mu + \lambda^{2})} - 2e^{ \mu + \frac{1}{2} \lambda^{2}} + 1 = e^{2( \mu + \lambda^{2})}-1 + 2(1-e^{ \mu + \frac{1}{2} \lambda^{2}}) \leq 8| \mu|+ 6 \lambda^{2} \leq \underbrace{(e^{ \mu+ \lambda^{2}} +1)^{2}}_{ \geq 1} (8| \mu|+ 6 \lambda^{2}),
\end{equation*}
for $| \mu|+ \lambda^{2} < \frac{1}{2}$ due to the elementary inequality $|e^{x} -1| \leq 2|x|$, for $|x| \leq 1$. However, if $| \mu|+ \lambda^{2} \geq \frac{1}{2}$, then $8| \mu|+ 6 \lambda^{2} \geq 1$, so the inequality holds also unconditionally. \qed
\end{proof}

Secondly, for convenience we formulate here a multivariate version of the martingale approximation central limit theorem, as a corollary of the results from Section 4.2 in \citet*{merlevede-peligrad-utev:19a}. However, we state the result using the concept of a mixingale.

\begin{definition}
A square-integrable (possibly multivariate) stochastic process $(\zeta_{t})_{t \in \mathbb{Z}}$, which is adapted to a filtration $( \mathcal{G}_{t})_{t \in \mathbb{Z}}$, is called an \emph{$L^{2}$-mixingale} of size $- \varphi_{0} \leq 0$ with respect to $( \mathcal{G}_{t})_{t \in \mathbb{Z}}$ if there exists $\varphi > \varphi_{0}$ such that for any $t \in \mathbb{Z}$,
\begin{equation} \label{eq:mixingale-def}
\| \mathbb{E}[ \zeta_{t+r} \mid \mathcal{G}_{t}] - \mathbb{E}[ \zeta_{t}] \|_{L^{2}( \mathbb{P})} = O(r^{- \varphi}), \quad \text{as } r \rightarrow \infty.
\end{equation}
Notably, if $( \zeta_{t})_{t \in \mathbb{Z}}$ is stationary and $( \mathcal{G}_{t})_{t \in \mathbb{Z}}$ is the natural filtration, i.e. $\mathcal{G}_{t} = \mathcal{F}^{ \zeta}_{t} \equiv \sigma \{ \zeta_{t}, \zeta_{t-1}, \ldots \}$ for any $t \in \mathbb{Z}$, then condition \eqref{eq:mixingale-def} reduces to
\begin{equation} \label{eq:mixingale-def-2}
\| \mathbb{E}[ \zeta_{r} \mid \mathcal{G}_{0}] - \mathbb{E}[ \zeta_{0}] \|_{L^{2}( \mathbb{P})} = O(r^{- \varphi}), \quad \text{as } r \rightarrow \infty.
\end{equation}
\end{definition}
The next lemma states the relevant mixingale CLT.

\begin{lemma} \label{lemma:mixingale-clt}
Suppose $(\zeta_{t})_{t \in \mathbb{Z}}$ is a multivariate, stationary and ergodic process, which is additionally an $L^{2}$-mixingale of size $- \frac{1}{2}$ in reference to $( \mathcal{F}^{ \zeta}_{t})_{t \in \mathbb{Z}}$. Then, as $T \rightarrow \infty$,
\begin{equation} \label{equation:mixingale-CLT}
\frac{1}{ \sqrt{T}} \sum_{t=1}^{T} \big( \zeta_{t} - \mathbb{E}[ \zeta_{0}] \big) \overset{d}{ \longrightarrow} N(0, \Sigma_{ \zeta}),
\end{equation}
where $\Sigma_{ \zeta} \equiv \sum_{ \ell = - \infty}^{ \infty} \mathbb{E} \big[( \zeta_{1} - \mathbb{E}[ \zeta_{0}])( \zeta_{1+ \ell} - \mathbb{E}[ \zeta_{0}])^{ \intercal} \big]$.
\end{lemma}

\begin{proof}
By the Cram\'{e}r--Wold device, we can reduce the multivariate convergence \eqref{equation:mixingale-CLT} to a univariate one by linear projection. To this end, assume that $(\zeta_{t})_{t \in \mathbb{Z}}$ is $d$-dimensional for some $d \in \mathbb{N}$ and let $a \in \mathbb{R}^d$ be arbitrary. Defining $\xi_{t} \equiv a^{ \intercal}( \zeta_{t} - \mathbb{E}[\zeta_{0}])$, $t \in \mathbb{Z}$, we deduce that $\mathcal{F}_{t}^{ \xi} = \sigma \{ \xi_{t}, \xi_{t-1}, \ldots \} \subset \mathcal{F}_{t}^{ \zeta}$ for any $t \in \mathbb{Z}$, whereby it follows that
\begin{equation*}
\big\| \mathbb{E}[ \xi_{r} \mid \mathcal{F}_{0}^{ \xi}] \big\|_{L^{2}( \mathbb{P})} \leq \|a\|_{ \mathbb{R}^{d}} \big\| \mathbb{E}[ \zeta_{r} \mid \mathcal{F}_{0}^{ \zeta} ] - \mathbb{E}[ \zeta_{0}] \big\|_{L^{2}( \mathbb{P})} = O(r^{- \varphi}),
\end{equation*}
as $r \rightarrow \infty$ for some $\varphi > 1/2$. Thus,
\begin{equation*}
\sum_{r=1}^{ \infty} r^{-1/2} \big\| \mathbb{E}[ \xi_{r} \mid \mathcal{F}_{0}^{ \xi} ] \big\|_{L^{2}( \mathbb{P})} < \infty,
\end{equation*}
so that, by Theorem 4.10 and Remark 4.13 of \citet{merlevede-peligrad-utev:19a},
\begin{equation*}
a^{ \intercal} \frac{1}{ \sqrt{T}} \sum_{t=1}^{T}( \zeta_{t} - \mathbb{E}[ \zeta_{0}]) = \frac{1}{ \sqrt{T}} \sum_{t=1}^{T} \xi_{t} \overset{d}{ \longrightarrow} N \bigg(0, \sum_{ \ell = - \infty}^{ \infty} \mathbb{E}[ \xi_{1} \xi_{1+ \ell}] \bigg)
\end{equation*}
as $T \rightarrow \infty$. It remains to note that
\begin{equation*}
\sum_{ \ell = - \infty}^{ \infty} \mathbb{E}[ \xi_{1} \xi_{1+ \ell}] = \sum_{ \ell = -\infty}^{ \infty} \mathbb{E}[a^{ \intercal}( \zeta_{1} - \mathbb{E}[ \zeta_{0}])( \zeta_{1+ \ell} - \mathbb{E}[ \zeta_{0}])^{ \intercal}a] = a^{ \intercal} \Sigma_{ \zeta}a,
\end{equation*}
and the claim follows. \qed
\end{proof}

To prove Proposition \ref{prop:clt}, building on Lemma \ref{lemma:mixingale-clt}, we still need the following technical lemma that estimates the memory of integrated variance.

\begin{lemma}\label{lem:cond-exp-decay}
Suppose that Assumption \ref{assum:gauss} and \ref{assumption:memory} hold. Moreover, suppose that $( \mathcal{F}_{t})_{t \in \mathbb{R}}$ is a filtration such that $B$ is adapted and has independent increments with respect to it (cf. condition (iii) in Assumption \ref{assum:error-clt}). Then, for any $p > 0$, $0 \leq s \leq t$ and $0 \leq s' \leq t'$,
\begin{itemize}
\item[(i)] $\Big\| \mathbb{E}_{ \theta_{0}} \Big[ \int_{s+r}^{t+r} \sigma_{u}^{p} \mathrm{d}u \mid \mathcal{F}_{0} \Big] - \mathbb{E}_{ \theta_{0}} \Big[ \int_{s}^{t} \sigma_{u}^{p} \mathrm{d}u \Big] \Big\|_{L^{2}( \mathbb{P}_{ \theta_{0}})} = O(r^{- \gamma+1/2})$,
\item[(ii)] $ \Big\| \mathbb{E}_{ \theta_{0}} \Big[ \int_{s+r}^{t+r} \sigma_{u}^{p} \mathrm{d}u \int_{s'+r}^{t'+r} \sigma_{u'}^{p} \mathrm{d}u' \mid \mathcal{F}_{0} \Big]- \mathbb{E}_{ \theta_{0}} \Big[ \int_{s}^{t} \sigma_{u}^{p} \mathrm{d}u \int_{s'}^{t'} \sigma_{u'}^{p} \text{d}u' \Big] \Big\|_{L^{2}( \mathbb{P}_{ \theta_{0}})}  = O(r^{- \gamma+1/2})$,
\end{itemize}
as $r \rightarrow \infty$.
\end{lemma}

\begin{proof}
We only prove (ii) as the proof of (i) is analogous. In explicit terms, Assumption \ref{assumption:memory} says that there exist constants $u_{0} \geq 0$ and $c > 0$ such that
\begin{equation} \label{eq:kernel-bound}
|K(u)| \leq c u^{- \gamma}, \quad u \geq u_{0}.
\end{equation}
Without loss of generality, assume $r \geq u_{0}$ from now on. By Tonelli's theorem,
\begin{align} \label{eq:first-tonelli}
\begin{split}
\mathbb{E}_{ \theta_{0}} \bigg[ \int_{s+r}^{t+r} \sigma_{u}^{p} \mathrm{d}u \int_{s'+r}^{t'+r} \sigma_{u'}^{p} \mathrm{d}u' \mid \mathcal{F}_{0} \bigg] &- \mathbb{E}_{ \theta_{0}} \bigg[ \int_{s}^{t} \sigma_{u}^{p} \text{d}u \int_{s'}^{t'} \sigma_{u'}^{p} \mathrm{d}u' \bigg] \\[0.10cm]
&= \int_{s}^{t} \int_{s'}^{t'} \big( \mathbb{E}_{ \theta_{0}}[ \sigma_{u+r}^{p} \sigma_{u'+r}^{p} \mid \mathcal{F}_{0}] - \mathbb{E}_{ \theta_{0}}[ \sigma_{u}^{p} \sigma_{u'}^{p}] \big) \mathrm{d}u \mathrm{d}u',
\end{split}
\end{align}
where
\begin{equation*}
\sigma_{v}^{p} \sigma_{v'}^{p} = \xi^{p} e^{- \frac{p}{2} \kappa(0)} \exp \bigg( \int_{- \infty}^{ \infty} K^{+}(v,v', \tau) \mathrm{d}B_{ \tau} \bigg)
\end{equation*}
with $\displaystyle K^{+}(v,v', \tau) = \frac{p}{2} \big[K(v- \tau) + K(v'- \tau) \big]$ for any $v,v' \geq 0$ (recall we set $K(v)=0$ for any $v \leq 0$). Subsequently,
\begin{equation*}
\mathbb{E}_{ \theta_{0}}[ \sigma_{u}^{p} \sigma_{u'}^{p}] = \xi^{p} e^{- \frac{p}{2} \kappa(0)} \exp \bigg( \frac{1}{2} \int_{- \infty}^{ \infty} K^{+}(u,u', \tau)^{2} \mathrm{d} \tau \bigg),
\end{equation*}
while, by the assumed properties of the Brownian motion $B$,
\begin{align*}
\mathbb{E}_{ \theta_{0}}&[ \sigma_{u+r}^{p} \sigma_{u'+r}^{p} \mid \mathcal{F}_{0}] \\[0.10cm]
&= \xi^{p} e^{- \frac{p}{2} \kappa(0)} \exp \bigg( \int_{- \infty}^{0} K^{+}(u+r,u'+r, \tau) \mathrm{d}B_{ \tau} \bigg) \mathbb{E} \bigg[ \exp \bigg( \int_{0}^{ \infty} K^{+}(u+r,u'+r, \tau) \mathrm{d} B_{ \tau} \bigg) \bigg] \\[0.10cm]
&= \xi^{p} e^{- \frac{p}{2} \kappa(0)} \exp \bigg( \int_{- \infty}^{0} K^{+}(u,u', \tau-r) \mathrm{d}B_{ \tau}  + \frac{1}{2} \int_{0}^{ \infty} K^{+}(u,u', \tau-r)^{2} \mathrm{d} \tau \bigg),
\end{align*}
using the property $K^{+}(u+r,u'+r, \tau) = K^{+}(u,u', \tau-r)$.

Therefore,
\begin{equation} \label{eq:corr-diff-1}
\mathbb{E}_{ \theta_{0}}[ \sigma_{u+r}^{p} \sigma_{u'+r}^{p} \mid \mathcal{F}_{0}] - \mathbb{E}_{ \theta_{0}}[ \sigma_{u}^{p} \sigma_{u'}^{p}]
= \xi^{2} e^{- \frac{p}{2} \kappa(0)} \exp \bigg( \frac{1}{2} \int_{- \infty}^{ \infty} K^{+}(u,u', \tau)^{2} \mathrm{d} \tau \bigg) \big( \exp(Y_{r}^{u,u'}) - 1 \big),
\end{equation}
with
\begin{equation} \label{eq:corr-diff-2}
Y_{r}^{u,u'} = \int_{- \infty}^{0} K^{+}(u,u', \tau-r) \mathrm{d}B_{ \tau} - \frac{1}{2} \overline{K}^{u,u'}(r) \sim N \bigg(- \frac{1}{2} \overline{K}^{u,u'}(r),\, \overline{K}^{u,u'}(r) \bigg)
\end{equation}
and $\overline{K}^{u,u'}(r) = \int_{- \infty}^{-r}  K^{+}(u,u', \tau)^{2} \mathrm{d} \tau = \int_{- \infty}^{0} K^{+}(u,u', \tau-r)^{2} \mathrm{d} \tau$. Applying Tonelli's theorem and Jensen's inequality to \eqref{eq:first-tonelli}, we conclude that:
\begin{align*}
\mathbb{E}_{ \theta_{0}} \Bigg[ \bigg( \mathbb{E}_{ \theta_{0}} \bigg[ \int_{s+r}^{t+r} \sigma_{u}^{p} \mathrm{d}u & \int_{s'+r}^{t'+r} \sigma_{u'}^{p} \mathrm{d}u' \mid \mathcal{F}_{0} \bigg] - \mathbb{E}_{ \theta_{0}} \bigg[ \int_{s}^{t} \sigma_{u}^{p} \mathrm{d}u \int_{s'}^{t'} \sigma_{u'}^{p} \mathrm{d}u' \bigg] \bigg)^{2} \Bigg] \\[0.10cm]
&\leq (t-s)(t'-s') \int_{s}^{t} \int_{s'}^{t'} \mathbb{E}_{ \theta_{0}} \big[( \mathbb{E}_{ \theta_{0}}[ \sigma_{u+r}^{p} \sigma_{u'+r}^{p} \mid \mathcal{F}_{0}] - \mathbb{E}_{ \theta_{0}}[ \sigma_{u}^{p} \sigma_{u'}^{p}])^{2} \big] \mathrm{d}u \mathrm{d}u',
\end{align*}
where, by \eqref{eq:corr-diff-1} -- \eqref{eq:corr-diff-2} and Lemma \ref{lem:lognormal},
\begin{align*}
\mathbb{E}_{ \theta_{0}} \big[( \mathbb{E}_{ \theta_{0}}[ \sigma_{u+r}^{p} \sigma_{u'+r}^{p} &\mid \mathcal{F}_{0}] - \mathbb{E}_{ \theta_{0}}[ \sigma_{u}^{p} \sigma_{u'}^{p}])^{2} \big] \\[0.10cm]
&= \xi^{2p} e^{-p \kappa(0)} \exp \bigg( \int_{- \infty}^{ \infty} K^{+}(u,u', \tau)^{2} \mathrm{d} \tau \bigg) \mathbb{E}_{ \theta_{0}} \Big[( \exp(Y_{r}^{u,u'})-1)^{2} \Big] \\[0.10cm]
&\leq 14 \xi^{2p} e^{-p \kappa(0)} \exp \bigg( \int_{- \infty}^{ \infty} K^{+}(u,u', \tau)^{2} \mathrm{d} \tau \bigg) \Big(e^{ \frac{3}{2} \overline{K}^{u,u'}(r)} + 1 \Big)^{2} \overline{K}^{u,u'}(r) \\[0.10cm]
& \leq 14 \xi^{2p} e^{\big( \frac{p^{2}}{2} - p\big) \kappa(0)} \Big(e^{ \frac{3p^{2}}{4} \kappa(0)}+1 \Big)^{2} \overline{K}^{u,u'}(r),
\end{align*}
after observing that
\begin{equation*}
\overline{K}^{u,u'}(r) \leq \int_{- \infty}^{ \infty} K^{+}(u,u', \tau)^{2} \mathrm{d} \tau \leq \frac{p^{2}}{2} \int_{0}^{ \infty} K( \tau)^{2} \mathrm{d} \tau = \frac{p^{2}}{2} \kappa(0).
\end{equation*}
Moreover, if $\tau \leq -r$ then $-\tau \geq r \geq u_{0}$, whereby $u- \tau \geq u_{0}$ and $u'- \tau \geq u_{0}$ since $u \geq s \geq 0$ and $u' \geq s' \geq 0$. Thus, by \eqref{eq:kernel-bound},
\begin{align*}
\overline{K}^{u,u'}(r) &= \frac{p^{2}}{4} \int_{- \infty}^{-r} \big(K(u- \tau) + K(u'- \tau) \big)^{2} \mathrm{d} \tau \\[0.10cm]
& \leq \frac{c^{2} p^{2}}{2} \int_{- \infty}^{-r} \Big( (u- \tau)^{-2 \gamma} +(u'- \tau)^{-2 \gamma} \Big) \mathrm{d} \tau \\[0.10cm]
& \leq c^{2}p^{2} \int_{r}^{ \infty} \tau^{-2 \gamma} \mathrm{d} \tau = \frac{c^{2} p^{2}}{1-2 \gamma} r^{-2 \gamma +1}.
\end{align*}
Consequently,
\begin{align*}
\bigg\| \mathbb{E}_{ \theta_{0}} \bigg[ \int_{s+r}^{t+r} \sigma_{u}^{p} \mathrm{d}u \int_{s'+r}^{t'+r} \sigma_{u'}^{p} \mathrm{d}u' \mid \mathcal{F}_{0} \bigg] &- \mathbb{E}_{ \theta_{0}} \bigg[ \int_{s}^{t} \sigma_{u}^{p} \mathrm{d}u \int_{s'}^{t'} \sigma_{u'}^{p} \mathrm{d}u' \bigg] \bigg\|_{L^{2}( \mathbb{P}_{ \theta_{0}})} \\[0.10cm]
&\leq (t-s)(t'-s') \bigg( \frac{14}{1-2 \gamma} \bigg)^{1/2} \xi^{p}  e^{ \big( \frac{p^{2}}{4}- \frac{p}{2} \big) \kappa(0)} \big(e^{ \frac{3p^{2}}{4} \kappa(0)}+1 \big) cpr^{- \gamma +1/2} \\[0.10cm]
&= O(r^{- \gamma +1/2}),
\end{align*}
as $r \rightarrow \infty$, which concludes the proof of (ii). \qed
\end{proof}

\begin{proof}[Proof of Proposition \ref{prop:clt}]
Thanks to Lemma \ref{lemma:mixingale-clt}, it suffices to show that the stationary process $(\widehat{ \mathbb{IV}}_{t})_{t \in \mathbb{Z}}$ is an $L^{2}$-mixingale of size $-\frac{1}{2}$ with respect to $(\mathcal{F}^{ \widehat{ \mathbb{IV}}}_{t})_{t \in \mathbb{Z}}$.

Let $r \geq 1$. First, we consider:
\begin{equation*}
\mathbb{E}_{ \theta_{0}} \Big[ \widehat{IV}_{r} \mid \mathcal{F}^{ \widehat{ \mathbb{IV}}}_{0} \Big] - g_{0}^{(1)}( \theta_{0}) = \mathbb{E}_{ \theta_{0}} \Big[IV_{r} \mid \mathcal{F}^{ \widehat{ \mathbb{IV}}}_{0} \Big] - \mathbb{E}_{ \theta_{0}}[IV_{1}] +  \mathbb{E}_{ \theta_{0}} \Big[ \varepsilon_{r} \mid \mathcal{F}^{ \widehat{ \mathbb{IV}}}_{0} \Big],
\end{equation*}
where
\begin{equation*}
\mathbb{E}_{ \theta_{0}} \Big[ \varepsilon_{r} \mid \mathcal{F}^{ \widehat{ \mathbb{IV}}}_{0} \Big] = \mathbb{E}_{ \theta_{0}} \Big[ \mathbb{E}_{ \theta_{0}}[ \varepsilon_{r} \mid \mathcal{F}^{ \sigma, \varepsilon}_{r-1}] \mid \mathcal{F}^{ \widehat{ \mathbb{IV}}}_{0} \Big] = 0
\end{equation*}
by Assumption \ref{assum:error} and the tower property of conditional expectations, which is applicable since $\mathcal{F}^{ \widehat{ \mathbb{IV}}}_{0} \subset \mathcal{F}^{ \sigma, \varepsilon}_{r-1}$. Therefore,
\begin{align} \label{eq:r-conv-1}
\begin{split}
\big\| \mathbb{E}_{ \theta_{0}} \Big[ \widehat{IV}_{r} \big| \mathcal{F}^{ \widehat{ \mathbb{IV}}}_{0} \Big] - g_{0}^{(1)}( \theta_{0}) \big\|_{L^{2}( \mathbb{P}_{ \theta_0})} &= \| \mathbb{E}_{ \theta_{0}} \Big[IV_{r} \mid \mathcal{F}^{ \widehat{ \mathbb{IV}}}_{0} \Big] - \mathbb{E}_{ \theta_{0}}[IV_{1}] \|_{L^{2}( \mathbb{P}_{ \theta_{0}})} \\[0.10cm]
&\leq \| \mathbb{E}_{ \theta_{0}} \Big[IV_{r} \mid \mathcal{F}^{W, \varepsilon}_{0} \Big] - \mathbb{E}_{ \theta_{0}}[IV_{1}] \|_{L^{2}( \mathbb{P}_{ \theta_{0}})} \\[0.10cm]
&= O(r^{- \gamma+1/2}), \quad r \rightarrow \infty, \\[0.10cm]
\end{split}
\end{align}
which follows by Lemma \ref{lem:exp-monotonicity}, again since $\mathcal{F}^{ \widehat{ \mathbb{IV}}}_{0} \subset \mathcal{F}^{W, \varepsilon}_{0}$, and Lemma \ref{lem:cond-exp-decay}(i).

Secondly,
\begin{align*}
\mathbb{E}_{ \theta_{0}} \Big[ \widehat{IV}_{r}^{2} \mid \mathcal{F}^{ \widehat{ \mathbb{IV}}}_{0} \Big]& - g_{0}^{(2)}( \theta_{0}) - c( \theta_{0}) \\[0.10cm]
&= \mathbb{E}_{ \theta_{0}} \Big[IV_{r}^{2} \mid \mathcal{F}^{ \widehat{ \mathbb{IV}}}_{0} \Big] - \mathbb{E}_{ \theta_{0}}[IV_{1}^{2}] + 2 \mathbb{E}_{ \theta_{0}} \Big[ \varepsilon_{r} IV_{r} \mid \mathcal{F}^{ \widehat{ \mathbb{IV}}}_{0} \Big] + \mathbb{E}_{ \theta_{0}} \Big[ \varepsilon_{r}^{2} \mid \mathcal{F}^{ \widehat{ \mathbb{IV}}}_{0} \Big] - \mathbb{E}_{ \theta_{0}}[ \varepsilon_{1}^{2}],
\end{align*}
where
\begin{equation*}
\mathbb{E}_{ \theta_{0}} \Big[ \varepsilon_{r}IV_{r} \mid \mathcal{F}^{ \widehat{ \mathbb{IV}}}_{0} \Big] = \mathbb{E}_{ \theta_{0}} \Big[ \mathbb{E}_{ \theta_{0}} \big[ \varepsilon_{r} \mid \mathcal{F}^{ \sigma, \varepsilon}_{r-1} \big] IV_{r} \mid \mathcal{F}^{ \widehat{ \mathbb{IV}}}_{0} \Big] = 0,
\end{equation*}
by the tower property, since $\mathcal{F}^{ \widehat{ \mathbb{IV}}}_{0} \subset \mathcal{F}^{ \sigma, \varepsilon}_{r-1}$, and Assumption \ref{assum:error}. By virtue of condition \eqref{item:epsilon-memory} in Assumption \ref{assum:error-clt} and Minkowski's inequality:
\begin{align} \label{eq:r-conv-2}
\begin{split}
\big\| \mathbb{E}_{ \theta_{0}} \Big[ \widehat{IV}_{r}^{2} \mid \mathcal{F}^{ \widehat{ \mathbb{IV}}}_{0} \Big] - g_{0}^{(2)}( \theta_{0}) - c( \theta_{0}) \big\|_{L^{2}( \mathbb{P}_{ \theta_{0}})} &= \| \mathbb{E}_{ \theta_{0}} \Big[IV_{1}^{2} \mid \mathcal{F}^{ \widehat{ \mathbb{IV}}}_{0} \Big] - \mathbb{E}_{ \theta_{0}}[IV_{1}^{2}] \|_{L^{2}( \mathbb{P}_{ \theta_{0}})} + O(r^{- \gamma+1/2}) \\[0.10cm]
&\leq \| \mathbb{E}_{ \theta_{0}} \Big[IV_{r}^{2} \mid \mathcal{F}^{W, \varepsilon}_{0} \Big] - \mathbb{E}_{ \theta_{0}}[IV_{1}^{2}]  \|_{L^{2}( \mathbb{P}_{ \theta_{0}})} + O(r^{- \gamma+1/2}) \\[0.10cm]
&= O(r^{- \gamma+1/2}),
\end{split}
\end{align}
as $r \rightarrow \infty$, by Lemma \ref{lem:exp-monotonicity} and \ref{lem:cond-exp-decay}(ii).

Lastly, for $\ell = 1,\ldots,k$ (assuming $r > k$ without loss of generality):
\begin{align*}
\mathbb{E}_{ \theta_{0}} \Big[ \widehat{IV}_{r} \widehat{IV}_{r- \ell} \mid \mathcal{F}^{ \widehat{ \mathbb{IV}}}_{0} \Big] - g_{ \ell}( \theta_{0}) &= \mathbb{E}_{ \theta_{0}} \Big[IV_{r} IV_{r- \ell} \mid \mathcal{F}^{ \widehat{ \mathbb{IV}}}_{0} \Big] - \mathbb{E}_{ \theta_{0}}[IV_{1} IV_{1- \ell}] \\[0.10cm]
&+ \mathbb{E}_{ \theta_{0}} \Big[ \varepsilon_{r} IV_{r- \ell} \mid \mathcal{F}^{ \widehat{ \mathbb{IV}}}_{0} \Big] + \mathbb{E}_{ \theta_{0}} \Big[IV_{r} \varepsilon_{r- \ell} \mid \mathcal{F}^{ \widehat{ \mathbb{IV}}}_{0} \Big] + \mathbb{E}_{ \theta_{0}} \Big[ \varepsilon_{r} \varepsilon_{r- \ell} \mid \mathcal{F}^{ \widehat{ \mathbb{IV}}}_{0} \Big],
\end{align*}
where
\begin{align*}
\mathbb{E}_{ \theta_{0}} \Big[ \varepsilon_{r} IV_{r- \ell} \mid \mathcal{F}^{ \widehat{ \mathbb{IV}}}_{0} \Big] &= \mathbb{E}_{ \theta_{0}}  \Big[ \mathbb{E}_{ \theta_{0}} \big[ \varepsilon_{r} \mid \mathcal{F}^{ \sigma, \varepsilon}_{r-1} \big]IV_{r- \ell} \mid \mathcal{F}^{ \widehat{ \mathbb{IV}}}_{0} \Big] = 0, \\[0.10cm]
\mathbb{E}_{ \theta_{0}} \Big[IV_{r} \varepsilon_{r- \ell} \mid \mathcal{F}^{ \widehat{ \mathbb{IV}}}_{0} \Big] &= \mathbb{E}_{ \theta_{0}} \Big[IV_{r} \mathbb{E}_{ \theta_{0}} \big[ \varepsilon_{r- \ell} \mid \mathcal{F}^{ \sigma, \varepsilon}_{r- \ell-1} \big] \mid \mathcal{F}^{ \widehat{ \mathbb{IV}}}_{0} \Big] = 0, \\[0.10cm]
\mathbb{E}_{ \theta_{0}} \Big[ \varepsilon_{r} \varepsilon_{r- \ell} \mid \mathcal{F}^{ \widehat{ \mathbb{IV}}}_{0} \Big] &= \mathbb{E}_{ \theta_{0}} \Big[ \mathbb{E}_{ \theta_{0}} \big[ \varepsilon_{r} \mid \mathcal{F}^{ \sigma, \varepsilon}_{r-1} \big] \varepsilon_{r- \ell} \mid \mathcal{F}^{ \widehat{ \mathbb{IV}}}_{0} \Big] =0,
\end{align*}
from tower property, because $\mathcal{F}^{ \widehat{\mathbb{IV}}}_{0} \subset \mathcal{F}^{ \sigma, \varepsilon}_{r- \ell-1} \subset \mathcal{F}^{ \sigma, \varepsilon}_{r}$, and Assumption \ref{assum:error}. Thus, applying yet again Lemma \ref{lem:exp-monotonicity}, we get
\begin{align} \label{eq:r-conv-3}
\begin{split}
\big\| \mathbb{E}_{ \theta_{0}} \Big[ \widehat{IV}_{r} \widehat{IV}_{r- \ell} \mid \mathcal{F}^{ \widehat{ \mathbb{IV}}}_{0} \Big] - g_{ \ell}( \theta_{0}) \big\|_{L^{2}( \mathbb{P}_{ \theta_0})} &= \| \mathbb{E}_{ \theta_{0}} \Big[IV_{r} IV_{r- \ell} \mid \mathcal{F}^{ \widehat{ \mathbb{IV}}}_{0} \Big] - \mathbb{E}_{ \theta_{0}}[IV_{1} IV_{1- \ell}]  \|_{L^{2}( \mathbb{P}_{ \theta_{0}})} \\[0.10cm]
&\leq \| \mathbb{E}_{ \theta_{0}} \Big[IV_{r} IV_{r- \ell} \mid \mathcal{F}^{W, \varepsilon}_{0} \Big] - \mathbb{E}_{ \theta_{0}}[IV_{1} IV_{1- \ell}] \|_{L^{2}( \mathbb{P}_{ \theta_{0}})} \\[0.10cm] &= O(r^{- \gamma+1/2}),
\end{split}
\end{align}
as $r \rightarrow \infty$, due to Lemma \ref{lem:cond-exp-decay}(ii).

Combining \eqref{eq:r-conv-1} -- \eqref{eq:r-conv-3}, we deduce that
\begin{equation*}
\Big\| \mathbb{E}_{ \theta_{0}} \Big[ \widehat{ \mathbb{IV}}_{r} \mid \mathcal{F}^{ \widehat{ \mathbb{IV}}}_{0} \Big] - \mathbb{E}_{ \theta_{0}}\Big[ \widehat{ \mathbb{IV}}_{r}  \Big]  \Big\|_{L^{2}( \mathbb{P}_{ \theta_{0}})} = \big\| \mathbb{E}_{ \theta_{0}} \Big[ \widehat{ \mathbb{IV}}_{r} \mid \mathcal{F}^{ \widehat{ \mathbb{IV}}}_{0} \Big] - G_{c}( \theta_{0})  \big\|_{L^{2}( \mathbb{P}_{ \theta_{0}})} =  O(r^{- \gamma+1/2}),
\end{equation*}
as $r \rightarrow \infty$. Since $\gamma > 1$, this implies that $(\widehat{ \mathbb{IV}}_{t})_{t \in \mathbb{Z}}$ is an $L^{2}$-mixingale of size $-\frac{1}{2}$. \qed
\end{proof}

\begin{proposition} \label{prop:exm-memory}
Suppose that Assumptions \ref{assum:gauss} and \ref{assumption:memory} hold. Then condition \emph{(ii)} of Assumption \ref{assum:error-clt} applies in Examples \ref{exm:CLT} -- \ref{exm:bipower}.
\end{proposition}

\begin{proof}
The error terms in Examples \ref{exm:CLT} and \ref{exm:bipower} differ only by scaling factor, so it suffices to look at the former. Then,
\begin{equation*}
\varepsilon_{t} = \bigg( \frac{2}{n} \int_{t-1}^{t} \sigma_{u}^{4} \mathrm{d}u \bigg)^{1/2} Z_{t}, \quad t \in \mathbb{Z},
\end{equation*}
so that, using the filtration $\mathcal{F}^{B,Z}_{t} = \sigma \{Z_{t},Z_{t-1}, \ldots \} \vee \sigma \{ B_{u} : u \leq t \}$, $t \in \mathbb{Z}$:
\begin{equation*}
\mathbb{E}_{ \theta_{0}} \Big[ \varepsilon_{r}^{2} \mid \mathcal{F}_{0}^{B,Z} \Big]- \mathbb{E}_{ \theta_{0}}[ \varepsilon_{1}^{2}] = \frac{2}{n} \bigg( \mathbb{E}_{ \theta_{0}} \bigg[ \int_{r-1}^{r} \sigma_{u}^{4} \mathrm{d}u \mid \mathcal{F}_{0}^{B,Z} \bigg] - \mathbb{E}_{ \theta_{0}} \bigg[ \int_{0}^{1} \sigma_{u}^{4} \text{d}u \bigg] \bigg), \quad r \geq 1,
\end{equation*}
since $\int_{r-1}^r \sigma_{u}^{4} \mathrm{d}u$ and $Z_{r}$ are conditionally independent on $\mathcal{F}_{0}^{B,Z}$. We can then apply Lemma \ref{lem:cond-exp-decay}(i) and \ref{lem:exp-monotonicity} to show the conjecture of this part.

In Example \ref{exm:indep},
\begin{equation*}
\varepsilon_{t} = \sum_{i=1}^{n} (Z_{t,i}^{2} -1) \int_{t-1+ \frac{i-1}{n}}^{t-1+ \frac{i}{n}} \sigma_{u}^{2} \mathrm{d}u, \quad t \in \mathbb{Z},
\end{equation*}
whereby for any $r \geq 1$,
\begin{equation*}
\mathbb{E}_{ \theta_{0}} \Big[ \varepsilon_{r}^{2} \mid \mathcal{F}_{0}^{B,Z} \Big]- \mathbb{E}_{ \theta_{0}}[ \varepsilon_{1}^{2}] = 2 \sum_{i=1}^n \Bigg( \mathbb{E}_{ \theta_{0}} \bigg[ \bigg( \int_{r-1+ \frac{i-1}{n}}^{r-1- \frac{i}{n}} \sigma_{u}^{2} \mathrm{d}u \bigg)^{2} \mid \mathcal{F}_{0}^{B,Z} \bigg] - \mathbb{E}_{ \theta_{0}} \bigg[ \bigg( \int_{ \frac{i-1}{n}}^{ \frac{i}{n}} \sigma_{u}^{2} \mathrm{d}u \bigg)^{2} \bigg] \Bigg).
\end{equation*}
Applying Minkowski's inequality and Lemmas \ref{lem:cond-exp-decay}(i) and \ref{lem:exp-monotonicity} concludes the proof. \qed
\end{proof}

\subsection{Proof of Theorem \ref{Thm:GMMCLTCorrection}}

We set $Q_{T} = \widehat{m}_{T}( \theta)^{ \intercal} \mathbb{W}_{T} \widehat{m}_{T}( \theta)$ and note that $Q_{T}$ attains its minimum value at $\widehat{ \theta}_{T}$. Combining this with the mean value theorem yields that
\begin{equation*}
0 = \nabla_{ \theta} Q_{T}( \widehat{ \theta}_{T}) = \nabla_{ \theta} Q_{T}( \theta_{0}) + \nabla_{\theta \theta}^{2} Q_{T}( \bar{ \theta}_{T})( \widehat{ \theta}_{T}- \theta_{0}),
\end{equation*}
where $\bar{\theta}_{T}$ lies between $\widehat{ \theta}_{T}$ and $\theta_{0}$. Now,
\begin{align*}
\nabla_{ \theta} Q_{T}( \theta) &= 2 \nabla_{ \theta} \widehat{m}_{T}( \theta)^{ \intercal} \mathbb{W}_{T} \widehat{m}_{T}( \theta), \\[0.10cm]
\nabla_{ \theta \theta}^{2} Q_{T}( \theta) &= 2 \nabla_{ \theta \theta} \widehat{m}_{T}( \theta)^{ \intercal} \mathbb{W}_{T} \widehat{m}_{T}( \theta)
+ 2 \nabla_{ \theta \theta} \widehat{m}_{T}( \theta)^{ \intercal} \mathbb{W}_{T} \nabla_{ \theta} \widehat{m}_{T}( \theta)
\end{align*}
This leads to:
\begin{equation}
\sqrt{T}( \widehat{ \theta}_{T}- \theta_{0})= \big( \nabla_{\theta \theta}^{2} Q_{T} (\bar \theta_{T}) \big)^{-1} 2 \nabla_{ \theta} \widehat{m}_{T}( \theta_{0})^{ \intercal} \mathbb{W}_{T} \sqrt{T} \widehat{m}_{T}( \theta_{0}).
\end{equation}
Invoking the assumptions of the theorem, it follows that $\bar{ \theta}_{T} \cp \theta_{0}$ as $T \rightarrow \infty$. In addition, and recalling Proposition \ref{prop:clt}, we deduce that as $T \rightarrow \infty$:
\begin{align*}
\sqrt{T} \widehat{m}_{T}( \theta_{0}) &\cd N \big(0, \Sigma_{ \widehat{ \mathbb{IV}}} \big), \\[0.10cm]
\nabla_{ \theta} \widehat{m}_{T}( \theta_{0}) & \cp J, \\[0.10cm]
\nabla_{ \theta \theta}^{2} Q_{T} ( \bar \theta_{T}) &\cp J^{ \intercal} \mathbb{W} J,
\end{align*}
where the last part uses that $\widehat{m}_{T}( \theta_{0}) \cp 0$. Then, Slutsky's theorem finishes the proof. \qed

\subsection{Proof of Theorem \ref{Thm:GMMCLTHAC}}
By the Cram\'{e}r--Wold device, we can reduce the proof to a univariate setting. Then, denoting $\Sigma_{T} = \Gamma(0) + 2 \sum_{ \ell=1}^{T-1} (1- \ell/T) \Gamma( \ell)$, it is enough to show that $\widehat{ \Sigma}_{T} - \Sigma_{T} \overset{ \mathbb{P}}{ \longrightarrow} 0$. This assertion follows from Theorem 2.1 of \citet*{davidson:20a}. To sketch the main ideas, we write
\begin{equation*}
\widehat{ \Sigma}_{T} / \Sigma_{T} - 1 = A_{1} + A_{2},
\end{equation*}
where
\begin{align*}
A_{1} &= \frac{1}{ \Sigma_{T}} \left[ \widehat{ \Gamma}(0)- \Gamma(0)+2 \frac{1}{T} \sum_{ \ell=1}^{T-1} w( \ell/L) \sum_{t=1}^{T- \ell} \left[ \big( \widehat{ \mathbb{IV}}_{t} - G_{c} \big) \big( \widehat{ \mathbb{IV}}_{t+ \ell} - G_{c} \big)- \Gamma( \ell) \right] \right], \\[0.10cm]
A_{2} &= \frac{1}{ \Sigma_{T}} 2 \sum_{ \ell=1}^{T-1} \big(w ( \ell/L  ) -1 \big) (1- \ell/T) \Gamma( \ell).
\end{align*}

Note that the term $A_{2}$ is deterministic. Moreover, it follows from the assumptions imposed on $w$ that $A_{2} = o(1)$.

To deal with the term $A_{1}$, we first note that based on Lemma \ref{lem:cond-exp-decay} and the proof of Proposition \ref{prop:clt}, for each $\ell \geq 0$,
the sequence $\{ \big( \widehat{ \mathbb{IV}}_{t} - G_{c} \big) \big( \widehat{ \mathbb{IV}}_{t+ \ell} - G_{c} \big)- \Gamma( \ell) \}$, is an $L^{2}$-mixingale of size $-\frac{1}{2}$. Then, Corollary 16.10 in \citet*{davidson:94a} implies that
\begin{align*}
\mathbb{E} \left[ \left( \sum_{t=1}^{T- \ell} \Big( \big( \widehat{ \mathbb{IV}}_{t} - G_{c} \big) \big( \widehat{ \mathbb{IV}}_{t+ \ell}-G_{c} \big)-\Gamma( \ell) \Big) \right)^{2} \right] \leq C T.
\end{align*}
Along with the convergence
\begin{equation*}
\frac{1}{L} \sum_{ \ell = 0}^{T-1} |w ( \ell/L  )| \longrightarrow \int_{0}^{ \infty} |w(x)| \mathrm{d}x,
\end{equation*}
we deduce that $\mathbb{E}[|A_{1}|] = O(L T^{-1/2}) = o(1)$ based on the assumption on $L$. \qed

\subsection{Proof of Theorem \ref{Thm:GMMCLTdouble}}

First, we provide a result for a general weight matrix as in (the proof of) Theorem \ref{Thm:GMMCLTCorrection}. Proceeding as in the proof there, and denoting $\widetilde{Q}_{n,T} = \widetilde{m}_{n,T}( \theta)^{ \intercal} \mathbb{W}_{n, T} \widetilde{m}_{n,T}( \theta)$, we find that
\begin{equation*}
\sqrt{T}( \widetilde{ \theta}_{n,T}- \theta_{0}) = \big( \nabla_{ \theta \theta}^{2} \widetilde{Q}_{n,T} ( \check{ \theta}_{n,T}) \big)^{-1}
2 \nabla_{ \theta} \widetilde{m}_{n,T}( \theta_{0})^{ \intercal} \mathbb{W}_{n, T} \sqrt{T} \widetilde{m}_{n,T}( \theta_0).
\end{equation*}
 Then, as $T \rightarrow \infty$ and $n \rightarrow \infty$,
\begin{align*}
\sqrt{T} \widetilde{m}_{n,T}( \theta_{0}) &\cd N(0, \Sigma_{\mathbb{IV}}), \\[0.10cm]
\nabla_{ \theta} \widetilde{m}_{n,T}( \theta_{0}) &\cp \widetilde{J}, \\[0.10cm]
\nabla_{ \theta \theta}^{2} \widetilde{Q}_{n,T}( \check{ \theta}_{n,T}) &\cp \widetilde{J}^{ \intercal} \mathbb{W} \widetilde{J},
\end{align*}
To wrap this up, we again exploit Slutsky's theorem.

To reach the final conclusion involving the estimator $\widehat{\Sigma}_{n,T}$, it is enough to further show the convergence $\widehat{\Sigma}_{n,T} \overset{ \mathbb{P}}{ \longrightarrow} \Sigma_{\mathbb{IV}}$ as $T \rightarrow \infty$ and $n \rightarrow \infty$. As in the proof of Theorem \ref{Thm:GMMCLTHAC}, this follows by noting that the error $V_{t}^{n} - IV_{t}$ is negligible in front of Assumption \ref{assum:doubleasympCLT}. \qed

\pagebreak

\section{Identification} \label{appendix:identification}

Assumption \ref{assumption:identification} is difficult to check, as it involves solving a system of nonlinear equations, which cannot be expressed in closed-form. In lieu of this, it is worthwhile to look briefly at a related question: Does equality of \emph{all} first and second-order moments, i.e.
\begin{equation} \label{eq:allmoments}
g^{(1)}_{0}( \theta_{1}) = g^{(1)}_{0}( \theta_{2}) \quad \text{and} \quad g_{ \ell}( \theta_{1}) = g_{ \ell}( \theta_{2}), \quad \text{for all $\ell \in \mathbb{N} \cup \{0 \}$},
\end{equation}
for $\theta_{1}, \theta_{2} \in \Theta$ hold if and only if $\theta_{1} = \theta_{2}$; the ``if'' direction being trivial. While an affirmative answer falls short of settling the original question, it does provide circumstantial evidence to suggest that there is no immediate issue with identification, so long as sufficiently many moments are included in the GMM estimation.

Note that since $g^{(1)}_{0}( \theta)= \xi$, the parameter $\xi$ is always identifiable.

In the fSV model, where the parameter vector is $\theta = ( \xi, \lambda, \nu,H)$, we only get a partial identification result. We first observe that, on the one hand, \eqref{eq:allmoments} implies
\begin{equation*}
\lim_{ \ell \rightarrow \infty} \frac{ \mathbb{E}_{ \theta_{1}} \big[(IV_{t} - \xi_{1})(IV_{t+l} - \xi_{1}) \big]}{ \mathbb{E}_{ \theta_{2}} \big[(IV_{t} - \xi_{2})(IV_{t+l} - \xi_{2})]} = \lim_{ \ell \rightarrow \infty} \frac{g_{ \ell}( \theta_{1}) - g^{(2)}_{0}( \theta_{1})^{2}}{g_{ \ell}( \theta_{2}) - g^{(2)}_{0}( \theta_{2})^{2}} = 1.
\end{equation*}
On the other hand, when $H_{1} \neq 0.5 \neq H_{2}$, Theorem \ref{theo:MomentsIVGeneral} and Equation (8) in \citet*{garnier-solna:18a} imply that
\begin{equation*}
\lim_{ \ell \rightarrow \infty} \frac{\mathbb{E}_{ \theta_{1}} \big[(IV_{t} - \xi_{1})(IV_{t+l} - \xi_{1}) \big]}{ \mathbb{E}_{ \theta_{2}} \big[(IV_{t} - \xi_{2})(IV_{t+l} - \xi_{2}) \big]} = \lim_{ \ell \rightarrow \infty} \frac{ \displaystyle \frac{ \xi_{1}^{2} H_{1}(2H_{1} + 1) \nu_{1}^{2}}{ \lambda_{1}^{2}} \ell^{2H_{1} - 2}}{ \displaystyle \frac{ \xi_{2}^{2} H_{2}(2H_{2} + 1) \nu_{2}^{2}}{ \lambda_{2}^{2}} \ell^{2H_{2} - 2}} = \lim_{ \ell \rightarrow \infty} \underbrace{ \frac{ \displaystyle \frac{ \xi_{1}^{2} H_{1}(2H_{1} + 1) \nu_{1}^{2}}{ \lambda_{1}^{2}}}{ \displaystyle \frac{ \xi_{2}^{2} H_{2}(2H_{2} + 1) \nu_{2}^{2}}{ \lambda_{2}^{2}}}}_{>0} \ell^{2(H_{1} - H_{2})},
\end{equation*}
where the limiting value is strictly positive and finite only if $H_{1} = H_{2}$.\footnote{If either $H_1 = 0.5$ or $H_2 = 0.5$ then either the numerator or denominator decays exponentially, so the limit of the ratio cannot equal one.} In that case, the limit is further equal to one only if
\begin{equation*}
\frac{ \nu_{1}}{ \lambda_{1}} = \frac{ \nu_{2}}{ \lambda_{2}},
\end{equation*}
since $\xi_1 = \xi_2$ by \eqref{eq:allmoments}.

We conjecture that a closer analysis of $g^{(2)}_{0}( \theta)$ and $g_{ \ell}( \theta)$, $\ell \in \mathbb{N}$, may help to uniquely identify also $\lambda$ and $\nu$, but as it does not appear completely trivial, we do not pursue this avenue further.

\pagebreak

\section{The Gamma-$\mathcal{BSS}$ process} \label{appendix:gamma-bss}

In this section, we assume that $Y$ is a Brownian semistationary ($\mathcal{BSS}$) SV process, i.e. a Gaussian process constructed with a serial correlation that is locally equivalent to a fBm, whereas the long-range dependence structure can differ a lot. We derive the expressions required to estimate the parameters of a restricted version of this model via our GMM approach.

The $\mathcal{BSS}$ process is defined as:
\begin{equation}
Y_{t} = \nu \int_{- \infty}^{t} h(t-s) \text{\upshape{d}} B_{s}, \quad t \geq 0,
\end{equation}
where $\nu > 0$ and $h: (0, \infty) \rightarrow \mathbb{R}$ is a kernel (subject to some regularity conditions).

A convenient choice is the gamma kernel $h(x) = x^{ \alpha} e^{- \lambda x}$ with $\alpha > -1/2$ and $\lambda>0$ resulting in the Gamma-$\mathcal{BSS}$. This model has local properties that are equivalent to the fSV model, and while not formally long-memory it does allow for substantial persistence in the time series. Its covariance structure was derived in \citet*{bennedsen-lunde-pakkanen:22a}, which we include in \eqref{equation:kappa-gamma-bss} for self-containedness. The autocovariance function of integrated variance in \eqref{equation_iv-gamma-bss} is new.

\begin{lemma} \label{lemma:acf-gamma-bss}
If $Y$ follows the Gamma-$\mathcal{BSS}$ process, $\kappa( \ell)$ has the form:
\begin{align} \label{equation:kappa-gamma-bss}
\begin{split}
	\kappa(0) &= \frac{ \nu^{2}}{(2 \lambda)^{2 \alpha+1}} \Gamma(2 \alpha+1), \\[0.10cm]
	\kappa( \ell) &= \frac{ \nu^{2} \Gamma( \alpha+1)}{ \sqrt{ \pi}} \left( \frac{ \ell}{2 \lambda} \right)^{ \alpha+ \frac{1}{2}} K_{ \alpha+1/2}( \lambda \ell), \quad \ell > 0,
\end{split}
\end{align}
where $K_{a}(x)$ is the Bessel function of the second kind. In addition, as $\ell \rightarrow \infty$, it follows that
\begin{equation} \label{equation_iv-gamma-bss}
\mathbb{E} \big[(IV_{t} - \xi)(IV_{t+ \ell}- \xi) \big] \sim \frac{ \nu^{2} \xi^{2} \Gamma( \alpha+1)( \exp( \lambda)-1)^2}{2^{ \alpha+1} \lambda^{ \alpha+2}} \ell^{ \alpha} \exp(- \lambda( \ell+1)).
\end{equation}
\end{lemma}

\begin{proof}
Recall that $\kappa( \ell) = \text{cov}(Y_{t}, Y_{t+ \ell})$, where for the Gamma-$\mathcal{BSS}$ process:
\begin{align*}
\kappa( \ell) = \nu^{2} \int_{0}^{ \infty} h(x) h(x+ \ell) \text{d}x.
\end{align*}
Inserting the Gamma kernel, $h(x) = x^{ \alpha}e^{- \lambda x}$, we deduce the identity:
\begin{align*}
\kappa(0) &= \nu^{2} \int_{0}^{ \infty} x^{2 \alpha} \exp(-2 \lambda x) \text{d}x \\[0.10cm]
&= \nu^{2} (2 \lambda)^{-2 \alpha-1} \int_{0}^{ \infty} z^{2 \alpha} \exp(-2z) \text{d}z \\[0.10cm]
&= \nu^{2} (2 \lambda)^{-2 \alpha-1} \Gamma(2 \alpha+1).
\end{align*}
Now, for each $\ell>0$,
\begin{align*}
\kappa( \ell) &= \nu^{2} \exp(- \lambda \ell) \int_{0}^{ \infty} x^{ \alpha} (x+ \ell)^{ \alpha} \exp(-2 \lambda x) \text{d}x \\[0.10cm]
&= \frac{ \nu^{2} \Gamma( \alpha+1)}{ \sqrt{ \pi}} \left( \frac{ \ell}{2 \lambda} \right)^{ \alpha+ \frac{1}{2}} K_{ \alpha+1/2}( \lambda \ell),
\end{align*}
where the last equality follows from \citet*{gradshteyn-ryzhik:14a}, formula 3.383 (8).

As $\kappa( \ell)$ adheres to \eqref{equation:covariance-asymptotic} with $\beta = \alpha$ and $\rho = \lambda$ by Remark 4.4 in \citet*{bennedsen-lunde-pakkanen:22a}, Theorem \ref{theo:MomentsIVGeneral} applies with
\begin{equation*}
\int_{-1}^{1}(1-|y|) \phi(y) \text{d}y = \int_{-1}^1(1-|y|) \exp(- \lambda y) \text{d} y = \frac{ \exp(- \lambda)( \exp( \lambda)-1)^{2}}{ \lambda^{2}},
\end{equation*}
so it follows that
\begin{equation*}
\gamma_{ \ell + 1,1} \sim F( \ell; \alpha, \lambda,v, \xi), \quad \ell \rightarrow \infty,
\end{equation*}
where
\begin{equation*}
F(\ell; \alpha, \lambda,v, \xi) \equiv \underbrace{ \frac{v^{2} \xi^{2} \Gamma( \alpha+1)( \exp( \lambda)-1)^{2}}{2^{ \alpha+1} \lambda^{ \alpha+2}}}_{>0} \ell^{ \alpha} \exp(- \lambda( \ell+1))
\end{equation*}
for $\ell>0$, $\alpha > - \frac{1}{2}$, $\lambda>0$, $v>0$ and $\xi>0$. \qed
\end{proof}

The Gamma-$\mathcal{BSS}$ process achieves Assumption \ref{assumption:memory} in its entire parameter space. Condition \emph{(i)} has been established for general $\mathcal{BSS}$ processes in \citet*[Proposition 2.2]{bennedsen-lunde-pakkanen:17b}. Condition \emph{(ii)} follows by noting that the modified Bessel function of the second kind, $K_{ \alpha}$, that appears in its covariance function is continuous. Moreover, the process is already expressed in the form of \eqref{equation:moving-average}.

There is a more complete result about identification for the Gamma-$\mathcal{BSS}$ process.
\begin{proposition}
In the Gamma-$\mathcal{BSS}$ process, the moment equality \eqref{eq:allmoments} holds if and only if $\theta_{1} = ( \xi_{1}, \nu_{1}, \lambda_{1}, \alpha_{1})$ equals $\theta_{2} = ( \xi_{2}, \nu_{2}, \lambda_{2}, \alpha_{2})$.
\end{proposition}

\begin{proof} Assume \eqref{eq:allmoments}, then by Lemma \ref{lemma:acf-gamma-bss}:
\begin{equation} \label{eq:ratio-lim}
\begin{split}
\lim_{ \ell \rightarrow \infty} \frac{ \mathbb{E}_{ \theta_{1}} \big[(IV_{t} - \xi_{1})(IV_{t+l} - \xi_{1}) \big]}{ \mathbb{E}_{ \theta_{2}} \big[(IV_{t} - \xi_{2})(IV_{t+l} - \xi_{2}) \big]} &= \lim_{ \ell \rightarrow \infty} \frac{ \displaystyle \frac{ \nu_{1}^{2} \xi_{1}^{2} \Gamma( \alpha_{1}+1)( \exp( \lambda_{1}) - 1)^{2}}{2^{ \alpha_{1} + 1} \lambda_{1}^{ \alpha_{1} + 2}} \ell^{ \alpha_{1}} \exp(- \lambda_{1}( \ell+1))}{ \displaystyle \frac{ \nu_{2}^{2} \xi_{2}^{2} \Gamma( \alpha_{2} + 1)( \exp( \lambda_{2}) - 1)^{2}}{2^{ \alpha_{2} + 1} \lambda_{2}^{ \alpha_{2} + 2}} \ell^{ \alpha_{2}} \exp(- \lambda_{2}( \ell + 1))} \\[0.10cm]
&= \lim_{ \ell \rightarrow \infty} \underbrace{ \frac{ \displaystyle \frac{ \nu_{1}^{2} \xi_{1}^{2} \Gamma( \alpha_{1} + 1)( \exp( \lambda_{1}) - 1)^{2}}{2^{ \alpha_{1} + 1}}}{ \displaystyle \frac{ \nu_{2}^{2} \xi_{2}^{2} \Gamma( \alpha_{2} + 1)( \exp( \lambda_{2}) - 1)^{2}}{2^{ \alpha_{2}+1}}}}_{>0} \ell^{ \alpha_{1} - \alpha_{2}} \exp(-( \lambda_{1} - \lambda_{2}) \ell).
\end{split}
\end{equation}
In order for this limit to be strictly positive and finite, it is necessary that $\lambda_{1} = \lambda_{2}$. In this vein, we also deduce that $\alpha_{1} = \alpha_{2}$ is necessary. Finally, since we already know that $\xi_{1} = \xi_{2}$ from the equality of the first moment, it remains to note that the limit \eqref{eq:ratio-lim} can only be equal to one if $\nu_{1} = \nu_{2}$.
\end{proof}

\pagebreak


\renewcommand{\baselinestretch}{1.0}
\small
\bibliographystyle{rfs}
\bibliography{userref}

\end{document}